%% file: these.tex
\documentclass[12pt,oneside,letterpaper]{book}

\usepackage[utf8]{inputenc}
\usepackage{graphicx}
\usepackage{color}
\usepackage{amsmath}
\usepackage{amssymb}
\usepackage{amsthm}  
\usepackage[colorlinks=true,linkcolor=black,citecolor=black,plainpages=false,pdfpagelabels]{hyperref}
\usepackage{palatino}
\usepackage[T1]{fontenc}

\usepackage{udem_these_eng}
\usepackage[english,french]{babel}

\newtheorem{thm}{Theorem}[chapter]
\newtheorem{cor}[thm]{Corollary}
\newtheorem{lem}[thm]{Lemma}
\newtheorem{defin}{Definition}[chapter]

\newcommand{\bra}[1]{\langle #1|}
\newcommand{\ket}[1]{|#1 \rangle}
\newcommand{\braket}[2]{\langle #1|#2\rangle}
\newcommand{\ketbra}[1]{\ket{#1}\bra{#1}}

\newcommand{\ident}{\mathbb{I}}
\DeclareMathOperator{\tr}{Tr}
\DeclareMathOperator{\rank}{rank}
\newcommand{\mdag}{^{\dag}} 
\newcommand{\demi}{\frac{1}{2}}

\DeclareMathOperator{\vect}{vec}
\DeclareMathOperator{\Span}{span}

\DeclareMathOperator{\op}{op}

\DeclareMathOperator{\Herm}{Herm}
\DeclareMathOperator{\Pos}{Pos}
\DeclareMathOperator{\DD}{D}
\DeclareMathOperator{\enc}{enc}
\DeclareMathOperator{\acc}{acc}
\DeclareMathOperator{\polylog}{polylog}
\newcommand{\LL}{\mathrm{L}}

\newcommand{\wtA}{\widetilde{A}}
\newcommand{\wtB}{\widetilde{B}}

\newcommand{\mbC}{\mathbb{C}}

\newcommand{\mbE}{\mathbb{E}}

\newcommand{\mbN}{\mathbb{N}}

\newcommand{\mbR}{\mathbb{R}}

\newcommand{\mbU}{\mathbb{U}}

\newcommand{\mfB}{\mathfrak{B}}

\newcommand{\mfD}{\mathfrak{D}}

\newcommand{\mfN}{\mathfrak{N}}

\newcommand{\mfQ}{\mathfrak{Q}}

\newcommand{\mfX}{\mathfrak{X}}
\newcommand{\mfY}{\mathfrak{Y}}

\newcommand{\sfA}{\mathsf{A}}
\newcommand{\sfB}{\mathsf{B}}
\newcommand{\sfC}{\mathsf{C}}
\newcommand{\sfD}{\mathsf{D}}
\newcommand{\sfE}{\mathsf{E}}

\newcommand{\sfX}{\mathsf{X}}

\Auteur{Frédéric}{Dupuis}

\President{M.}{Pierre}{McKenzie}{Ph.D.}

\Directeur{M.}{Gilles}{Brassard}{Ph.D.}

\CoDirecteur{M.}{Patrick}{Hayden}{Ph.D.}

\Membres{1}{M.}{Louis}{Salvail}{Ph.D.}

\Membres{2}{M.}{Renato}{Renner}{Ph.D.}

\Membres{3}{M.}{François}{Lalonde}{Ph.D.}

\PhD{The decoupling approach to quantum information theory}
    {}
    {d'informatique et de recherche op\'{e}rationnelle}
   {informatique}
    {Novembre}
    {2009}

\begin{document}

\setcounter{page}{1}
 \PagesCouverture

\setstretch{1.3}
\resume
\selectlanguage{french}

\noindent La théorie de l'information quantique étudie les limites fondamentales qu'imposent les lois de la physique sur les tâches de traitement de données comme la compression et la transmission de données sur un canal bruité.  Cette thèse présente des techniques générales permettant de résoudre plusieurs problèmes fondamentaux de la théorie de l'information quantique dans un seul et même cadre. Le théorème central de cette thèse énonce l'existence d'un protocole permettant de transmettre des données quantiques que le receveur connaît déjà partiellement à l'aide d'une seule utilisation d'un canal quantique bruité. Ce théorème a de plus comme corollaires immédiats plusieurs théorèmes centraux de la théorie de l'information quantique.

Les chapitres suivants utilisent ce théorème pour prouver l'existence de nouveaux protocoles pour deux autres types de canaux quantiques, soit les canaux de diffusion quantiques et les canaux quantiques avec information supplémentaire fournie au transmetteur. Ces protocoles traitent aussi de la transmission de données quantiques partiellement connues du receveur à l'aide d'une seule utilisation du canal, et ont comme corollaires des versions asymptotiques avec et sans intrication auxiliaire. Les versions asymptotiques avec intrication auxiliaire peuvent, dans les deux cas, être considérées comme des versions quantiques des meilleurs théorèmes de codage connus pour les versions classiques de ces problèmes.

Le dernier chapitre traite d'un phénomène purement quantique appelé \emph{verrouillage}: il est possible d'encoder un message classique dans un état quantique de sorte qu'en lui enlevant un sous-système de taille logarithmique par rapport à sa taille totale, on puisse s'assurer qu'aucune mesure ne puisse avoir de corrélation significative avec le message. Le message se trouve donc \og verrouillé \fg{} par une clé de taille logarithmique. Cette thèse présente le premier protocole de verrouillage dont le critère de succès est que la distance trace entre la distribution jointe du message et du résultat de la mesure et le produit de leur marginales soit suffisamment petite.

{\bfseries Mots clés\hspace{-3pt}: Théorie de l'information, information quantique}
\setstretch{1.4}

\selectlanguage{english}

\abstract

\vspace{-0.5cm}
Quantum information theory studies the fundamental limits that physical laws impose on information processing tasks such as data compression and data transmission on noisy channels. This thesis presents general techniques that allow one to solve many fundamental problems of quantum information theory in a unified framework. The central theorem of this thesis proves the existence of a protocol that transmits quantum data that is partially known to the receiver through a single use of an arbitrary noisy quantum channel. In addition to the intrinsic interest of this problem, this theorem has as immediate corollaries several central theorems of quantum information theory.

The following chapters use this theorem to prove the existence of new protocols for two other types of quantum channels, namely quantum broadcast channels and quantum channels with side information at the transmitter. These protocols also involve sending quantum information partially known by the receiver with a single use of the channel, and have as corollaries entanglement-assisted and unassisted asymptotic coding theorems. The entanglement-assisted asymptotic versions can, in both cases, be considered as quantum versions of the best coding theorems known for the classical versions of these problems.

The last chapter deals with a purely quantum phenomenon called \emph{locking}. We demonstrate that it is possible to encode a classical message into a quantum state such that, by removing a subsystem of logarithmic size with respect to its total size, no measurement can have significant correlations with the message. The message is therefore ``locked'' by a logarithmic-size key. This thesis presents the first locking protocol for which the success criterion is that the trace distance between the joint distribution of the message and the measurement result and the product of their marginals be sufficiently small.


\tabledesmatieres



\listedesannexes


\notation
\setstretch{1.3}

\label{notation-at-beginning-firstpage}

\newcommand{\tabstart}[1]{\noindent \begin{tabular}{p{2.95cm}p{10cm}}
    \multicolumn{2}{l}{{\bf #1}} \\ \hline \\[-2.5ex] } 

\newcommand{\tabstop}{\\ \hline \end{tabular}}

\newcommand{\tabinter}{\vspace{2ex}}

\tabstart{General}
$\log$ & Logarithm base 2.\\
$\ln$ & Natural logarithm.\\
$\mbR$ & Real numbers.\\
$\mbC$ & Complex numbers.\\
$c^*$ & Complex conjugate of $c$.\\
$\mbE_U[f(U)]$ & Expectation value of $f(U)$ over the random variable $U$.
\tabstop

\tabinter

\tabstart{Linear Algebra and Quantum Systems}
$A, B, C, \dots$ & Labels for quantum systems, or linear operators between Hilbert spaces. (Should be clear from context.)\\
$\sfA, \sfB, \sfC, \dots$ & Hilbert spaces associated with the systems $A, B, C, \dots$   \\
$|A|$ & Dimension of $\sfA$.\\
$AB$ & Composite quantum system whose associated Hilbert space is $\sfA \otimes \sfB$.\\
$A^n$ & Quantum system composed of $n$ copies of $A$.\\
$\LL(\sfA, \sfB)$ & The space of linear operators from $\sfA$ to $\sfB$ \\
$\LL(\sfA)$ & $\LL(\sfA, \sfA)$\\
$M^{A \rightarrow B}$ & Indicates that the operator $M$ is in $\LL(\sfA, \sfB)$.\\
$M\mdag$ & Adjoint of $M$\\
$M^{A \rightarrow B}_T$ & Transpose of $M$ with respect to the canonical bases of $\sfA$ and $\sfB$. This has lower priority than matrix multiplication: $AB \cdot C = (AB) C (AB)\mdag$\\
$M \cdot N$ & $MNM\mdag$\\
$\Herm(\sfA)$ & The set of Hermitian operators from $\sfA$ to $\sfA$\\
$\Pos(\sfA)$ & The subset of $\Herm(\sfA)$ consisting of positive semidefinite matrices\\
$M \leqslant N$ & If $M, N \in \Herm(\sfA)$, this means that $N - M \in \Pos(\sfA)$.
\tabstop

\tabinter

\tabstart{Linear Algebra and Quantum Systems, continued}
$\DD(\sfA)$ & The set of all density operators on $\sfA$; i.e.\ $\DD(\sfA) = \{ \rho : \rho \in \Pos(\sfA), \tr[\rho] = 1 \}$\\
$\mathcal{N}^{A \rightarrow B}, \mathcal{T}^{A \rightarrow B}, \dots$    & Superoperators (completely positive linear maps from $\LL(\sfA)$ to $\LL(\sfB)$)\\
$\ident^A$ & Identity operator on $\sfA$ or identity superoperator on $\LL(\sfA)$. (Should be clear from context.)\\
$\ket{\psi}^A, \ket{\varphi}^A, \dots$ & Vectors in $\sfA$.\\
$\psi^A, \varphi^A, \dots$ & The ``unketted'' versions denote their associated density matrices: $\psi^A = \ketbra{\psi}$. Furthermore, if we have defined a state $\psi^{AB}$, then $\psi^A = \tr_B[\psi^{AB}]$.\\
$\op_{A \rightarrow B}(\ket{\psi}^{AB})$ & Turns a vector into an operator. See Section \ref{sec:vec-op}.\\
$\vect(M^{A \rightarrow B})$ & Turns an operator into a vector. See Section \ref{sec:vec-op}.\\
$\sqrt{M}$ & If $M \in \Pos(\sfA)$ has spectral decomposition $M = \sum_i \lambda_i \ketbra{\psi_i}$, then $\sqrt{M} = \sum_i \sqrt{\lambda_i} \ketbra{\psi_i}$.\\
$\ket{\Phi}^{AA'}$ & $\frac{1}{\sqrt{|A|}}\sum_{i=1}^{|A|} \ket{i}^A \ket{i}^{A'}$, where $\ket{i}^A$ and $\ket{i}^{A'}$ are fixed canonical bases for $\sfA$ and $\sfA'$, and $\sfA \cong \sfA'$.\\
$\pi^A$ & The maximally mixed state $\frac{\ident^A}{|A|}$
\tabstop

\tabinter

\tabstart{Norms and Distance measures}
$\left\| M^{A \rightarrow B} \right\|_1$ & $\tr\sqrt{M\mdag M}$\\
$\left\| \ket{\psi} \right\|_2$ & $\sqrt{|\braket{\psi}{\psi}|}$\\
$\left\| M^{A \rightarrow B} \right\|_2$ & $\sqrt{\tr[M \mdag M]}$\\
$\left\| M^{A \rightarrow B} \right\|_{\infty}$ & Largest singular value of $M$.\\
$\left\| \mathcal{N}^{A \rightarrow B} \right\|_{\diamond}$ & Diamond norm; see Section \ref{sec:diamond-norm}.\\
$F(\rho^A, \sigma^A)$ & $\left\| \sqrt{\rho} \sqrt{\sigma} \right\|_1$. This is called the \emph{fidelity}.\\
$d_F(\rho^A, \sigma^A)$ & $\sqrt{1 - F(\rho, \sigma)^2}$. This is called the \emph{fidelity distance}.
\tabstop

\tabinter

\tabstart{Entropies}
$H(A|B)_{\rho}$ & Conditional von Neumann entropy of $A$ given $B$ on $\rho^{AB}$, see Definition \ref{def:von-neumann}.\\
$H_2(A|B)_{\rho}$ & Conditional 2-entropy of $A$ given $B$, defined as $-\log \min_{\sigma^B \in \DD(\sfB)} \tr\left[ \left( (\sigma^B \otimes \ident^A)^{-1/2} \rho^{AB}\right)^2  \right]$\\
$H_2^{\varepsilon}(A|B)_{\rho}$ & Smooth 2-entropy of $A$ given $B$, defined as $\max_{\sigma^{AB}, d_F(\rho, \sigma) \leqslant \varepsilon} H_2(A|B)_{\sigma}$\\
$H_{\min}(A|B)_{\rho}$ & Conditional min-entropy, see Definition \ref{def:cond-min-entropy}.\\
$H_{\max}(A|B)_{\rho}$ & Conditional max-entropy, see Definition \ref{def:cond-max-entropy}.\\
$H^{\varepsilon}_{\min}(A|B)_{\rho}$ & $\varepsilon$-smooth conditional min-entropy, see Definition \ref{def:smooth-cond-min-entropy}.\\
$H^{\varepsilon}_{\max}(A|B)_{\rho}$ & $\varepsilon$-smooth conditional max-entropy, see Definition \ref{def:smooth-cond-max-entropy}.\\
$I(A \rangle B)_{\rho}$ & Coherent information, see Definition \ref{def:coherent-info}.\\
$I(A; B)_{\rho}$ & Mutual information, see Definition \ref{def:mutual-info}.\\
$I(A; B|C)_{\rho}$ & Conditional mutual information, see Definition \ref{def:cond-mutual-info}.
\tabstop

\tabinter 

\tabstart{First names}
Alice & The sender in all the protocols.\\
Bob & The receiver in all the protocols.
\tabstop
\label{notation-at-beginning-lastpage}
\setstretch{1.4}


\remerciements

\setstretch{1.15}
\selectlanguage{french}
Je tiens d'abord à remercier mes deux directeurs, Gilles Brassard et Patrick Hayden de qui j'ai tant appris, pour leur soutien constant et sans qui cette thèse n'aurait jamais vu le jour.

Je dois aussi beaucoup à mes coauteurs, avec lesquels j'ai eu énormément de plaisir à collaborer: Simon-Pierre Desrosiers, Patrick Hayden, Ke Li, Debbie Leung, André Méthot, Nicolas Gisin, Avinatan Hassidim et Haran Pilpel.

Je tiens aussi à remercier Andreas Winter, Renato Renner et Aram Harrow pour leur hospitalité ainsi que pour les discussions enrichissantes que j'ai eues avec eux durant les quelques semaines où une grosse partie de la rédaction de cette thèse a été accomplie.

Je remercie aussi Mario Berta, qui m'a envoyé une copie de sa thèse de \og Diplom \fg{}, d'où j'ai appris un des trucs utilisés dans la preuve du théorème principal.

De plus, je tiens à remercier les membres de mon jury pour leur coopération, leur lecture attentive ainsi que les nombreuses corrections qu'ils ont suggérées. Je remercie également Marco Tomamichel et Oleg Szehr pour leurs corrections incluses dans la version finale.

Je voudrais aussi remercier les gens qui ont été (ou qui sont toujours) au LITQ à l'Université de Montréal ou au LaCIQ à McGill durant mon doctorat et qui ont contribué à rendre mon environnement de travail si agréable: Ashton Anderson, Somshubhro  Bandyopadhay, Laurent  Beaudou, Guido  Berlín, Hugue  Blier, Anne  Broadbent, Félix  Bussières, Michaël  Cadilhac, Claude  Crépeau, Omar  Fawzi, Jan  Florjanczyk, Sébastien  Gambs, Claude  Gravel, Charles  Hélou, Kassem  Kalach, Olivier Landon-Cardinal, Adrien  Lemaître, Abubakr  Muhammad, Roger  Müller, Nicolas Dutil, Éric  Paquette, Paul Chouha, David Pouliot, Louis  Salvail, Ivan  Savov, Jean-Raymond Simard, Alain  Tapp et Jürg  Wullschleger.

Puisqu'il faut d'abord se loger et se nourrir pour être en mesure de déchiffrer les mystères de l'univers, je tiens à remercier le CRSNG ainsi que l'ICRA pour leur soutien financier tout au long de mon doctorat.

Finalement, je tiens tout particulièrement à remercier ma famille, qui a de tout temps cru en moi et dont le soutien n'a jamais faibli depuis le début de ma vie.

\setstretch{1.4}
\selectlanguage{english}


\debutchapitres

\chapter{Introduction}

The origins of information theory go back to 1948, when Claude Shannon published ``A mathematical theory of communication'' \cite{shannon}, in which he proposed a mathematical framework to study information processing tasks such as data compression and data transmission over noisy channels. Data compression is the following task: we have a large amount of digital data, and we would like to shrink it down to a smaller size for efficient storage or transmission. If the data is sufficiently redundant, then it is possible to do this with a very small probability of decompressing it incorrectly. Data transmission over noisy channels involves the following problem: one has a communication channel in which the transmitter can select an input and the receiver receives an output that has been corrupted by noise in the channel. A concrete example of this would be the phone line between a house and the telephone central, or the radio link between a cellphone tower and the handsets. One would then like to use this channel to send a message and make sure that, with high probability, the receiver will be able to reconstruct it exactly.

Since our universe is governed by the laws of quantum mechanics, the physical limits imposed on these problems are themselves quantum mechanical. It also turns out that information can behave in counterintuitive ways under the laws of quantum mechanics: for example, one can know precisely the state of a two-particle quantum system while remaining ignorant of the state of either of the two particles separately. Furthermore, measurements made on two particles that are kept very far apart can exhibit correlations that could not be explained classically without assuming that information was transmitted faster than the speed of light. This is why a quantum version of information theory is so interesting: it is our attempt at taming these apparent paradoxes and counterintuitive facts. In this thesis, we will be concerned specifically with coding for various different types of quantum channels. In the last chapter, we will also look at the phenomenon of \emph{information locking}, in which a small key can ``unlock'' an amount of information far beyond what would be possible classically.

\section{Decoupling}
One of the most bizarre features of quantum information theory turns out to be extremely useful for solving channel coding problems. It is the notion of \emph{purification}: given any quantum system $A$ whose state is random, one can find a bigger system $AB$ such that the state on $A$ is the same as before, but where the global state on $AB$ is completely deterministic. This is impossible classically: if the state of a system is random, considering it together with another system only adds the potential of having more randomness globally. This, however, will help us tremendously. In a channel coding problem, we want to ensure that the output of the channel is strongly correlated (or ``coupled'') with the input. When we look at the purification of the final state that we want between the input and output, however, it turns out that this is equivalent to requiring that the input to the channel be completely decorrelated with the entire universe minus the channel output. This helps us because we can achieve it by \emph{destroying} correlations---and, as in other areas of life, destruction is easier to achieve than construction.

This ``decoupling'' approach---we use the term ``decouple'' to mean ``decorrelate''---has therefore become a staple of quantum information theory. It was already used to some extent in \cite{lsd3}, the first general coding theorem for the quantum capacity of quantum channels, and was used more systematically in \cite{state-merging} and \cite{FQSW}, which derived basic quantum protocols from which a large number of other, previously known protocols could be derived. In \cite{lsd-decoupling}, the results of \cite{lsd3} were revisited using a ``purer'' decoupling approach.

While we have some sense that these last three papers use the same ``trick'', they are nonetheless proven separately, and while they can be used to derive other protocols, one sometimes needs to work quite a bit to accomodate the particular forms of the theorems (using, for instance, typical projectors to limit the dimensions of various quantum systems). One of the main contributions of this thesis is to give a general decoupling theorem, from which all of the known ones can be derived very easily, and which is much more flexible. We then go on to give quantum coding theorems for different varieties of quantum channels, including quantum broadcast channels and quantum channels with side information at the transmitter. In both cases, no prior results exist regarding the particular tasks considered. Finally, we also use the main decoupling theorem to prove a result on information locking.

\section{Contributions}
This thesis is broken down into the following chapters:

\textbf{Chapter \ref{chp:preliminaries} (Preliminaries)}: This chapter contains the concepts and definitions necessary to understand the rest of the thesis. It does not contain original material.

\textbf{Chapter \ref{chp:decoupling} (The decoupling theorem)}: This chapter is devoted to the main decoupling theorem. We state it and prove it along with several variants, including a new one-shot coding theorem for quantum channels, in which Bob potentially knows part of the state before the start of the protocol. We then use it to rederive the main results of \cite{state-merging}, \cite{FQSW}, \cite{gpw04} and \cite{lsd-decoupling} in a more straightforward manner. The contents of this chapter will be published as a paper at a later date.

\textbf{Chapter \ref{chp:side-info} (Quantum channels with side information at the transmitter)}: This chapter derives new results on quantum channels with side information at the transmitter. A channel with side information at the transmitter is a channel in which the transmitter has access ahead of time to information about the noise in the channel, but where the receiver does not have access to this information. We give a one-shot coding theorem for them similar to the one for regular channels in Chapter \ref{chp:decoupling}, and show that applying it to entanglement-assisted coding for memoryless channels yields an optimal protocol. In particular, we show that the entanglement-assisted capacity of these channels admits a single-letter formula that parallels the solution to the classical version of this problem given in \cite{gelfand-pinsker}. Part of the work in this section was presented in a different form at the 2009 International Symposium on Information Theory \cite{gpquantique}.

\textbf{Chapter \ref{chp:bcast} (Quantum broadcast channels)}: This chapter contains a coding theorem for quantum broadcast channels, namely channels with one input but two outputs going to two physically separated receivers. Again, we give a general one-shot coding theorem, and we then derive from it an entanglement-assisted coding scheme for memoryless channels that parallels the best known classical coding theorem for broadcast channels given in \cite{marton}. These are the first coding theorems given for these tasks.  A different version of this work was accepted for publication in IEEE Transactions on Information Theory and is joint work with Patrick Hayden and Ke Li \cite{dhl09}.

\textbf{Chapter \ref{chp:locking} (Locking classical information in quantum states)}: This chapter deals with the purely quantum phenomenon of information locking. We show that there exists a unitary such that if we encode a classical message into a quantum state, apply this unitary to it, and remove a very small part (logarithmic in the total size), then one can get almost no information about the message by measuring the remaining part. This is done by showing that the statistical distance between the joint distribution of the message and the measurement result and a product distribution can be made very small. This is slightly stronger than what was done in prior information locking results, in which upper bounds on the mutual information between the measurement result and the message were derived. Furthermore, this is the first locking protocol in which one uses a single unitary and a quantum key instead of applying one of several unitaries and using the choice of unitary as the key. We also show that this scheme can be used to construct a quantum key distribution protocol that guarantees that the eavesdropper can gain almost no information about the key by making a measurement immediately after the execution of the protocol, but where the eavesdropper only needs to learn a very small portion of the key to be able to recover the rest. This underscores much more spectacularly than before \cite{krbm07} the need to take into account the fact that an eavesdropper might keep \emph{quantum} information after the protocol and use it only when making his actual attack. This will be published at a later date and is joint work with Patrick Hayden and Debbie Leung.

\textbf{Chapter \ref{chp:conclusion} (Conclusion)}: This chapter concludes the thesis with a recapitulation of what was done, and speculates on what the future might hold.

\chapter{Preliminaries}\label{chp:preliminaries}
This chapter explains the notation used throughout the thesis and presents some concepts one needs to understand this document.

\section{Notation}
Linear algebra is the language of quantum mechanics; we therefore start by introducing the notation we will use for linear algebraic concepts. One can find explanations of all the concepts below in any linear algebra textbook, or, to be introduced to these concepts in the setting of quantum information, in \cite{watrousnotes}. Note that a condensed version of this appears on pages \pageref{notation-at-beginning-firstpage}--\pageref{notation-at-beginning-lastpage} so the reader can refer back to it more easily. We will denote by sans-serif capital letters (such as $\sfA, \sfB,\dots$) complex finite-dimensional inner product vector spaces (which we will usually simply call Hilbert spaces following the usual quantum information convention---these are the only Hilbert spaces that we will ever consider in this thesis), and we will use regular capital letters $A, B, \dots$ to label the quantum systems associated with the spaces $\sfA, \sfB, \dots$. We will denote the dimension of $\sfA$ by $|A|$. Vectors in $\sfA$ are denoted by ``kets'' $\ket{\psi}^{A}$ with the superscript omitted when it causes no confusion. Furthermore, we will denote by $\LL(\sfA, \sfB)$ the space of linear operators from $\sfA$ to $\sfB$, and we will use the shorthand $\LL(\sfA)$ for $\LL(\sfA, \sfA)$. Elements of the dual space $\LL(\sfA, \mbC)$ of $\sfA$ are written as ``bras''; for instance the dual of $\ket{\psi}$ is written $\bra{\psi}$. We use $\mdag$ to designate the adjoint of an operator, $\Herm(\sfA)$ is the set of all Hermitian (self-adjoint) operators on $\sfA$, and $\Pos(\sfA) \subseteq \Herm(\sfA)$ is the set of all positive semidefinite operators on $\sfA$. Given two operators $M, N \in \Herm(\sfA)$, we say that $M \leqslant N$ if $N - M \in \Pos(\sfA)$. Given an operator $M \in \LL(\sfA, \sfB)$, we will use the superscript $M^{A \rightarrow B}$ to indicate its input and output spaces. We will use the symbol $\cdot$ to denote conjugation: given two operators $M^{A \rightarrow B}$ and $N^A$, we define $M \cdot N = MNM\mdag$.

We will denote by the calligraphic letters $\mathcal{N}^{A \rightarrow B}, \mathcal{T}^{A \rightarrow B}, \mathcal{S}^{A \rightarrow B} ,\dots$ completely positive linear maps from $\LL(\sfA)$ to $\LL(\sfB)$; we will call these ``superoperators''. We will also write $\ident^A$ for either the identity operator on $\sfA$, or the identity superoperator on $\LL(\sfA)$; which one is meant should be clear from the context.

In superscripts, we will simply concatenate letters to indicate the tensor product: for instance, $M^{AB \rightarrow CD} \in \LL(\sfA \otimes \sfB, \sfC \otimes \sfD)$. When applying an operator $M^{A \rightarrow B}$ to a vector $\ket{\psi}^{AC}$, we will usually omit the implicit identity: for instance, $M^{A \rightarrow B} \ket{\psi}^{AC} = (M^{A \rightarrow B} \otimes \ident^C) \ket{\psi}^{AC}$.

Given an operator $M^{AB} \in \LL(\sfA \otimes \sfB)$ on a composite Hilbert space, we can define its \emph{partial trace on $B$}, denoted either as $\tr_B[M^{AB}]$ or simply by $M^A$, omitting the $B$ in the superscript, as the unique operator $N$ in $\LL(\sfA)$ such that $\tr[ZN] = \tr[(Z \otimes \ident^B) M^{AB}]$ for all $Z \in \LL(\sfA)$. In other words, the partial trace is defined as the adjoint of the superoperator $\mathcal{F}^{A \rightarrow AB}, \mathcal{F}(N^A) = N^A \otimes \ident^B$ under the Hilbert-Schmidt inner product $\langle X, Y \rangle := \tr[X\mdag Y]$.

We will also need the concept of \emph{partial isometries}. A partial isometry is an operator $V^{A \rightarrow B}$ whose singular values are all either 1 or 0. Equivalently, they can be defined as any operator $V^{A \rightarrow B}$ such that $V\mdag V$ and $VV\mdag$ are projectors. A \emph{full-rank partial isometry} is a partial isometry $V^{A \rightarrow B}$ whose rank is $\min\{ |A|, |B| \}$.

\section{Quantum mechanics: an extremely short introduction}
Since one cannot hope to cover basic quantum mechanics in a few paragraphs, the author strongly recommends the interested reader to consult \cite{NC2000} or \cite{watrousnotes} for a more complete introduction. Nonetheless, a short introduction to the basic concepts using the notation that we will use is given here for the sake of completeness.

A quantum system $A$ is represented by a Hilbert space $\sfA$; a state of the system is a positive semidefinite operator $\rho^A \in \Pos(\sfA)$ such that $\tr[\rho^A] = 1$. We also call these states \emph{density operators} or \emph{density matrices} and denote the set of all density operators on $A$ as $\DD(\sfA)$. A state $\rho$ is considered \emph{pure} if $\rank \rho = 1$, in which case there exists a $\ket{\psi} \in \sfA$ of norm 1 and such that $\rho = \ketbra{\psi}$. We sometimes also call general states \emph{mixed states} when we want to emphasize the fact that the state is not necessarily pure. Let $\sfA$ and $\sfB$ be Hilbert spaces corresponding to quantum systems $A$ and $B$; we can then consider them as a single composite system $AB$ with $\sfA \otimes \sfB$ as its associated Hilbert space. By convention in this document, we will write systems on which a quantum state is defined as a superscript; for instance, $\rho^{AB} \in \DD(\sfA \otimes \sfB)$. The same convention will apply to all operators. If the input and output spaces of an operator are different, we will write an arrow in the superscript to indicate this: for example, $M^{A \rightarrow B} \in \LL(\sfA, \sfB)$. When we want to consider only part of a composite system, we take its partial trace on the system we want to eliminate.

The operations that can be applied to a quantum system without making it interact with other systems correspond to the unitary operators on the associated Hilbert space, namely all transformations of the form $\rho \rightarrow U \rho U\mdag$, where $U$ is unitary. Since conjugation will be used so often in this thesis, we will use the notation $A \cdot B$ to denote $ABA\mdag$. Transformations involving interactions with other systems can be simulated by adding an ancillary system, applying a unitary on the composite system, and then tracing out part of the remaining system.

Such a transformation can also be represented by a trace-preserving superoperator (sometimes called CPTP map, which stands for ``completely positive trace-preserving'' map). It can be shown that a linear map $\mathcal{N}^{A \rightarrow B}$ is completely positive if and only if it can be written as $\mathcal{N}(\rho) = \sum_i N_i \rho N_i\mdag$, where $N_i \in \LL(\sfA, \sfB)$; furthermore, any such linear map is trace-preserving (i.e.\ $\tr[\mathcal{N}(M)] = \tr[M] \hspace{3mm} \forall M$) if $\sum_i N_i \mdag N_i = \ident^A$. We will sometimes call trace-preserving superoperators ``quantum channels'' when we want to emphasize that this is a transformation over which we have no control and wish to view as a noisy channel.

There is also a class of operations that leaves the quantum system intact but changes its underlying Hilbert space. For instance, suppose we have a state $\rho^A$ and want to embed the information it contains into the system $B$. An operation that does this is a partial isometry $V^{A \rightarrow B}$ such that $\tr[V \rho V\mdag] = 1$ (i.e.\ the image of $V\mdag$ must contain the support of $\rho$). We will sometimes call these simply ``isometries'' since they act as isometries on the part of $\sfA$ in which $\rho$ lies. Such an operation can be implemented by the superoperator $\mathcal{V}^{A \rightarrow B}(\rho^A) = V\rho V\mdag + \sum_i N_i \rho N_i\mdag$, where the $N_i$ are such that $V\mdag V + \sum_i N_i \mdag N_i = \ident^A$ and $\tr[N_i \rho N_i\mdag] = 0$ for all $i$. 

Quantum systems can also be measured, yielding a classical output. In addition to the measurement result, a measurement can also have a quantum residue, in case the measurement does not completely measure the state. To represent this, we will use a special type of trace-preserving superoperator that we will call a \emph{measurement superoperator}. A measurement superoperator is a superoperator of the form $\mathcal{M}^{A \rightarrow BX}(\sigma^A) = \sum_x \ketbra{x}^X \otimes N_x \sigma^A {N_x}\mdag$, where the $\ket{x}$ are all part of the same orthonormal basis for $\sfX$, and the $N_x^{A \rightarrow B}$ are arbitrary operators such that $\mathcal{M}$ is CPTP. The interpretation for this is that we get the measurement result $x$ with probability $\tr[N_x \sigma N_x\mdag]$, $X$ is a classical register that holds the measurement result, and, if the measurement result was $x$, the $B$ register gets the state $N_x \sigma N_x\mdag/\tr[N_x \sigma N_x\mdag]$. If we are not interested in the quantum residue and only care about the classical result, we only need to describe the set of positive semidefinite operators $\{ N_x\mdag N_x \}$ to describe the measurement, which is then called a \emph{positive operator valued measure}, or POVM. We call a measurement superoperator \emph{complete} if all of the $N_x$ are of rank 1, in which case the quantum residue is superfluous since it can be reconstructed from the classical result only.

One particularly strange and interesting feature of quantum mechanics is the concept of entanglement. We say that a bipartite state $\rho^{AB}$ is \emph{entangled} if it cannot be written in the form $\rho^{AB} = \sum_i \alpha_i \sigma^A_i \otimes \omega^B_i$. In other words, a state on $AB$ in entangled if it cannot be expressed as a probabilistic mixture of separate states on $A$ and $B$. An example of such a state is the following pure maximally entangled state that will be of great importance throughout the thesis: $\Phi^{AB} = \ketbra{\Phi}^{AB}$ with $\ket{\Phi}^{AB} = \frac{1}{\sqrt{|A|}} \sum_i \ket{i}^A \otimes \ket{i}^{B}$, where $|A| = |B|$ and the $\ket{i}^A$ and $\ket{i}^B$ are standard orthonormal bases for $A$ and $B$. When $|A|=|B|=2$, we will call this state an \emph{EPR pair} \cite{epr}, after Einstein, Podolski and Rosen who first noticed the phenomenon of entanglement and defined this state. With some abuse of terminology, we will call higher-dimensional instances of this state ``EPR pairs'' even when the dimension is not a power of two.

Of central importance to this thesis is the concept of \emph{purification}. Given a mixed state $\rho^A$, it is always possible to find a pure state $\omega^{AB}$ on a larger system such that $\rho^A = \tr_B[\omega^{AB}]$. Note that this is also a purely quantum phenomenon: if one has a probability distribution $p$ over a set $\mathfrak{X}$, it is impossible to find a single element $(x,y)$ of $\mathfrak{X} \times \mathfrak{Y}$ which then somehow ends up being distributed as $p$ when we stop looking at the $y$ part of it!

An analogous fact holds for quantum channels: given a completely positive superoperator $\mathcal{N}^{A \rightarrow B}$, it is possible to find a partial isometry $U_{\mathcal{N}}^{A \rightarrow BE}$ such that $\mathcal{N}(X) = \tr_E[U_{\mathcal{N}} \cdot X]$ for every $X \in \LL(\sfA)$. In other words, one can find a deterministic operation that takes the input $A$ to two output systems: the actual output of the channel $B$, and an environment system $E$. When we ignore the environment system, we get exactly the same channel. We call such a partial isometry a \emph{Stinespring dilation} of $\mathcal{N}$.

\section{Distance measures}
We will often need a notion of distance between quantum states, usually to state that the result of a particular protocol that we developed is ``close'' to some ideal output state that we would like to get. The distance we will use most of the time is called the \emph{trace distance}; the trace distance between two states $\rho$ and $\sigma$ is $\| \rho - \sigma \|_1$, where $\| M \|_1 := \tr\sqrt {M\mdag M}$ for any $M^{A \rightarrow B}$. In other words, it is equal to the sum of the absolute values of the eigenvalues of the matrix $\rho - \sigma$. The reason for which this is a meaningful measure of distance is that it characterizes how easy it is for someone to determine through a measurement whether an unknown state is $\rho$ or $\sigma$, as was discovered by Helstrom:
\begin{thm}[Helstrom's theorem \cite{helstrom}]\label{thm:helstrom}
	Let $\rho^A$ and $\sigma^A$ be two density operators on $A$, and suppose one holds $\rho^A$ with probability $\demi$ and $\sigma^A$ with probability $\demi$, and one tries to determine which one it is by performing a measurement on $A$. Then, the best possible measurement will give the correct answer with probability $\demi + \frac{1}{4} \| \rho - \sigma \|_1$.
\end{thm}
\begin{proof}
	Let $\{ M_{\rho}^A, M_{\sigma}^A \}$ be a POVM used to guess which state we have (since there are only two possible answers, one only needs two POVM operators). Then, the probability of guessing correctly is
	\begin{align*}
		\demi \tr[M_{\rho} \rho] + \demi \tr[M_{\sigma} \sigma] &= \demi \tr[M_{\rho} \rho +  (\ident - M_{\rho}) \sigma]\\ 
		&= \demi \tr[M_{\rho} \rho + \sigma - M_{\rho} \sigma ]\\
		&= \demi + \demi \tr[M_{\rho}(\rho - \sigma)]\\
		&\leqslant \demi + \demi \tr[M_{\rho}P_+(\rho - \sigma)P_+]\\
		&\leqslant \demi + \demi \tr[P_+(\rho - \sigma)P_+]\\
		&= \demi + \frac{1}{4} \| \rho - \sigma \|_1
	\end{align*}
	where $P_+$ is a projector onto the eigenspaces of $\rho - \sigma$ corresponding to positive eigenvalues. The first inequality is due to the operator inequality $P_+(\rho - \sigma)P_+ \geqslant \rho - \sigma$, and the second inequality, to the operator inequality $M_{\rho} \leqslant \ident$. The last equality is due to the fact that, since $\rho - \sigma$ has zero trace, $P_+(\rho - \sigma)P_+$ must correspond to exactly half of the trace distance. Of course, equality can be attained if $M_{\rho} = P_+$.
\end{proof}
This means that, if the trace distance between two states is very small, someone trying to determine which of the states an unknown state is in will be scarcely better off by doing the optimal measurement than by guessing randomly. In particular, if the output of a quantum protocol is $\varepsilon$-close in trace distance to the output of an ideal protocol, then, regardless of what we use the protocol for, we will almost never be able to tell the difference.

A related notion is the \emph{fidelity} between quantum states:
\begin{defin}[Fidelity]
Given two states $\rho$ and $\sigma$, their fidelity is defined as $F(\rho, \sigma) := \left\| \sqrt{\rho} \sqrt{\sigma} \right\|_1$.
\end{defin}
One can easily see that the fidelity approaches one when two states get closer together; in fact, $F(\rho, \rho) = \| \rho \|_1 = 1$.  An important property of the fidelity is that it is stable under purifications: given two states $\rho^{A}$ and $\sigma^A$, and a purification $\rho^{AB}$ of $\rho^A$, then $F(\rho^A, \sigma^A) = \max_{\sigma^{AB}} F(\rho^{AB}, \sigma^{AB})$, where we maximize over all purifications of $\sigma^A$. This is due to Uhlmann's theorem \cite{uhlmann} and will be proven in the next chapter as Theorem \ref{thm:uhlmann}. One can also define a distance measure based on the fidelity:
\begin{defin}[Fidelity distance]
Let $\rho$ and $\sigma$ be two density operators. Then, their \emph{fidelity distance} is defined as
\[ d_F(\rho, \sigma) := \sqrt{1 - F(\rho, \sigma)^2}. \]
\end{defin}

The fidelity distance is essentially equivalent to the trace distance, as shown by the Fuchs-van de Graaf inequalities \cite{fuchs-vdg}:
\begin{lem}[Fuchs-van de Graaf inequalities]\label{lem:fuchs-vdg}
Let $\rho \in \DD(\sfA)$ and $\sigma \in \DD(\sfA)$ be density operators on $A$. Then,
\begin{equation*}
1 - \demi \| \rho - \sigma \|_1 \leqslant F(\rho, \sigma) \leqslant \sqrt{1 - \frac{1}{4}\| \rho - \sigma \|_1^2}.
\end{equation*}
This implies that
\[ \demi \| \rho - \sigma \|_1 \leqslant d_F(\rho, \sigma) \leqslant \sqrt{\| \rho - \sigma \|_1}. \]
\end{lem}

\subsection{The diamond norm}\label{sec:diamond-norm}
It will also be convenient on a few occasions to be able to compare two superoperators. To do this, we introduce the so-called \emph{diamond norm}:

\begin{defin}[Diamond norm]
Let $\mathcal{N}: \LL(\sfA) \rightarrow \LL(\sfB)$ be any linear operator from $\LL(\sfA)$ to $\LL(\sfB)$. Then, we define its diamond norm to be
\[ \left\| \mathcal{N} \right\|_{\diamond} := \max_{\sigma^{AA'} \in \DD(\sfA \otimes \sfA')} \left\| (\mathcal{N}^{A \rightarrow B} \otimes \ident^{A' \rightarrow A'})(\sigma^{AA'}) \right\|_1 \]
where the maximization is taken over all mixed states $\sigma^{AA'}$, and where $\sfA' \cong \sfA$.
\end{defin}

This norm is usually called the \emph{completely bounded trace norm} in operator theory and has been an object of study in that field for many years (see, for example, \cite{paulsen-operator-theory} for an introduction to the area), but it was introduced to quantum information theory by Kitaev \cite{diamond-norm} as the ``diamond norm''.

The main reason for using the diamond norm to define a notion of distance on quantum channels is essentially the same as for using the trace norm on quantum states: it characterizes the optimal probability of successfully distinguishing two channels. Just as Theorem \ref{thm:helstrom} shows that the optimal probability of distinguishing the quantum states $\rho$ and $\sigma$ is $\demi + \frac{1}{4}\|\rho - \sigma \|_1$, it is possible to show that the optimal probability of distinguishing the quantum channels $\mathcal{N}$ and $\mathcal{M}$ is given by $\frac{1}{2} + \frac{1}{4} \|\mathcal{N} - \mathcal{M} \|_{\diamond}$.

\section{Information measures}

\subsection{von Neumann entropy and derived quantities}
To be able to give solutions to information theory problems, we must have ways of measuring amounts of information. The fundamental quantity is the \emph{von Neumann entropy} of a quantum state:

\begin{defin}[von Neumann entropy \cite{vonneumann-grundlagen}]\label{def:von-neumann}
	The von Neumann entropy of a quantum state $\rho^A$ is defined as $H(A)_{\rho} := -\tr[\rho^A \log \rho^A]$ (where, if $\rho^A = \sum_i \lambda_i \ketbra{\psi_i}^A$ is a spectral decomposition of $\rho^A$, $\log \rho^A = \sum_i \log(\lambda_i) \ketbra{\psi_i}^A$, and we interpret $0 \log 0$ as 0).
\end{defin}
The von Neumann entropy measures the amount of information present in sequences of many copies of the same state, i.e.\ in $\rho^{\otimes n}$. More specifically, it has been shown by Schumacher that an i.i.d.\ state ${\rho^{A}}^{\otimes n}$ can be compressed into $n[H(A)_{\rho} - \delta]$ qubits with an error rate going to zero as $n \rightarrow \infty$ for any $\delta > 0$ \cite{schumi95}. Hence, the higher the entropy, the less certain we are about the state and the more space we need to store it.

Many other information measures are derived from the von Neumann entropy. The first one is the conditional von Neumann entropy:
\begin{defin}[Conditional von Neumann entropy]
Given a state $\rho^{AB}$, the conditional von Neumann entropy of $A$ given $B$ is defined as
\[ H(A|B)_{\rho} := H(AB)_{\rho} - H(B)_{\rho}. \]
\end{defin}
This is meant to describe the amount of uncertainty that we have about $A$ if we already possess $B$. This interpretation was problematic for a long time, however, given that it can be negative (for instance, $H(A|A')_{\Phi^{AA'}} = -1$, where $|A| = |A'| = 2$ and $\Phi^{AA'} = \demi \sum_{i,j = 1}^2 \ket{ii}\bra{jj}$). However, it turns out to give the solution to the following problem: given a state ${\rho^{AB}}^{\otimes n}$ between Alice and Bob, how many EPR pairs are required between Alice and Bob to teleport Alice's $n$ shares to Bob (with free classical communication)? This task, called \emph{state merging} \cite{state-merging}, is possible if we have $n[H(A|B)_{\rho} + \delta]$ EPR pairs, and the error goes to zero as $n \rightarrow \infty$ for every $\delta > 0$. When $H(A|B)_{\rho}$ is negative, we can teleport while \emph{generating} EPR pairs at this rate.

Another information measure derived from the von Neumann entropy is the quantum mutual information:
\begin{defin}[Quantum mutual information]\label{def:mutual-info}
Let $\rho^{AB}$ be a quantum state. Then, the mutual information between $A$ and $B$ is defined as
\begin{align*}
I(A;B)_{\rho} &:= H(A)_{\rho} + H(B)_{\rho} - H(AB)_{\rho}\\
&= H(A)_{\rho} - H(A|B)_{\rho}\\
&= H(B)_{\rho} - H(B|A)_{\rho}.
\end{align*}
\end{defin}
Without going into details, this quantity gives the entanglement-assisted classical capacity of memoryless quantum channels (channels of the form $\mathcal{N}^{\otimes n}$) \cite{BSST02}. It is always nonnegative.

We can also define a conditional version of quantum mutual information:
\begin{defin}[Conditional quantum mutual information]\label{def:cond-mutual-info}
Let $\rho^{ABC}$ be a quantum state. Then, the mutual information between $A$ and $B$ given $C$ is defined as
\[ I(A;B|C)_{\rho} := I(A;BC)_{\rho} - I(A;C)_{\rho}. \]
\end{defin}
This is also never negative \cite{strong-subadditivity}, a fact that turns out to be highly nontrivial to prove. While we will only use this quantity in a technical proof in Chapter \ref{chp:side-info}, it nonetheless has a natural interpretation through the task of \emph{state redistribution} \cite{redistribution}.

The coherent information is another measure that is important for channel coding problems. Unlike the previous one, this one has no classical analogue:
\begin{defin}[Coherent information]\label{def:coherent-info}
Let $\rho^{AB}$ be a quantum state. Then, the coherent information from $A$ to $B$ is defined as
\[ I(A \rangle B)_{\rho} := -H(A|B)_{\rho}.  \]
\end{defin}
This is only positive when conditional entropy is negative, which only happens when a state is entangled. This quantity gives the best known general rate for unassisted transmission of quantum data through i.i.d.\ quantum channels.

\subsection{Properties of the von Neumann entropy}\label{sec:entropy-properties}
The family of entropic quantities defined above have a number of useful properties. In all of the statements below, let $\psi^{ABC}$ be any pure state with respect to which all entropic quantities are computed:
\begin{itemize}
\item $H(A) = H(BC)$
\item $H(AB) = H(A) + H(B|A)$
\item $H(A) = \demi I(A;B) + \demi I(A;C)$
\item $I(A \rangle B) = \demi I(A;B) - \demi I(A;C)$
\item $H(A|B) = -H(A|C)$
\end{itemize}
All of the above can be easily proven from the definitions. One can also show that, on a \emph{mixed} state $\rho^{ABC}$, the following holds:
\begin{itemize}
\item $I(A;BC) \geqslant I(A;B)$.
\end{itemize}
In other words, the mutual information is monotonic under the addition of more subsystems; by taking into consideration an additional system $C$ in addition to $B$, one cannot lose information about $A$. This comes from the strong subadditivity of the von Neumann entropy \cite{strong-subadditivity} and its proof is rather involved compared to the previously stated properties.

\subsection{One-shot information measures}
All of the above quantities were relevant for tasks involving $n$ copies of a state, or $n$ uses of a quantum channel, and where we then take the limit as $n \rightarrow \infty$. However, in this thesis, we will generally start from protocols involving a single use of an arbitrary channel on a given arbitrary state, and then derive this special case by considering a ``single use'' of the channel $\mathcal{N}^{\otimes n}$. We will therefore need information measures that are relevant for the one-shot case and that reduce to the above quantities in the case of multiple copies.

The first one is the min-entropy of a quantum state:
\begin{defin}[Quantum min-entropy]
Let $\rho^A$ be a quantum state. Then, its min-entropy is defined as
\[ H_{\min}(A)_{\rho} := -\log \min_{\lambda \in \mbR} \{ \lambda : \rho^A \leqslant \lambda \ident^A \}. \]
\end{defin}
In other words, the min-entropy is the negative logarithm of the largest eigenvalue. Classically, this definition goes back to Chor and Goldreich \cite{chor-goldreich}, and was generalized to quantum information by Renner \cite{renner-phd}.

Renner also defined a conditional version of the quantum min-entropy:
\begin{defin}[Quantum conditional min-entropy]\label{def:cond-min-entropy}
Let $\rho^{AB} \in \Pos(\sfA \otimes \sfB)$. Then, the conditional min-entropy of $A$ given $B$ is defined as
\[ H_{\min}(A|B)_{\rho} := -\log \min \left\{ \tr[\sigma^B] : \sigma^B \in \Pos(\sfB), \rho^{AB} \leqslant \ident^A \otimes \sigma^B \right\}. \]
\end{defin}

This quantity measures how much uniform and private randomness we can extract from a random variable that is correlated with a quantum state that an attacker might possess, as shown in \cite{renner-phd}. It is also the quantity that governs how many bits of key must be used to encrypt the $A$ part of a quantum state $\rho^{AB}$ against an adversary that knows $B$ \cite{sec-entropique}.

While much more is known about the min-entropy, the following slightly more unwieldy quantity is used in many proofs:

\begin{defin}[Quantum conditional 2-entropy]
Let $\rho^{AB} \in \Pos(\sfA \otimes \sfB)$. Then, the conditional 2-entropy of $A$ given $B$ is defined as
\[ H_2(A|B)_{\rho} := -\log \inf_{\sigma^B \in \DD(\sfB), \sigma^B > 0} \tr\left[ \left( (\sigma^B \otimes \ident^A)^{-1/2} \rho^{AB} \right)^2 \right]. \]
\end{defin}

Note that the conditional 2-entropy is always lower-bounded by the conditional min-entropy:
\begin{lem}
Let $\rho^{AB} \in \Pos(\sfA \otimes \sfB)$; then $H_{\min}(A|B)_{\rho} \leqslant H_2(A|B)_{\rho}$.
\end{lem}
\begin{proof}
Let $\lambda = 2^{-H_{\min}(A|B)_{\rho}}$, and let $\sigma^B$ be a normalized density operator such that $\rho^{AB} \leqslant \lambda \ident^A \otimes \sigma^B$; assume without loss of generality that $\sigma^B$ is positive definite (otherwise redefine $B$ as the support of $\rho^B$). Also, let $P^{AB} = \ident^{AB} - (\rho^{AB})^{0}$ (i.e.\ $P$ is a projector onto the kernel of $\rho^{AB}$). Then, using the fact that $X \leqslant Y \Rightarrow X^{-1/2} \geqslant Y^{-1/2}$ (which one can derive from Propositions V.I.6 and V.I.8 in \cite{bhatia}), we have that $\lambda^{1/2} (\rho^{AB} + \varepsilon P)^{-1/2} \geqslant (\ident^A \otimes \sigma^B + \varepsilon P)^{-1/2}$, and therefore
\begin{align*}
\tr\left[ \left( (\ident^A \otimes \sigma^B + \varepsilon P)^{-1/2} \rho^{AB} \right)^2 \right] &\leqslant \lambda \tr\left[ \left( (\rho^{AB} + \varepsilon P)^{-1/2} \rho^{AB} \right)^2 \right]\\
&= \lambda \tr[\rho^{AB}]\\
&= \lambda.
\end{align*}
Taking the limit as $\varepsilon \rightarrow 0$ yields the lemma.
\end{proof}

One can also define a ``max-entropy'' as done in \cite{min-max-entropy} in the following manner:
\begin{defin}[Quantum conditional max-entropy]\label{def:cond-max-entropy}
Let $\psi^{ABC}$ be a pure state. Then, the conditional max-entropy of $A$ given $B$ is defined as
\[ H_{\max}(A|B)_{\psi} := -H_{\min}(A|C)_{\psi}. \]
Since $H_{\min}(A|C)_{\psi}$ is invariant under unitaries on $C$, this does not depend on the particular choice of purification.
\end{defin}

Note that there are at present two competing definitions of the max-entropy in circulation, at least in the non-conditional case. The other one is simply the logarithm of the rank of a state. However, the author feels that Definition \ref{def:cond-max-entropy} is more compelling given the various results in \cite{min-max-entropy} and \cite{duality-min-max-entropy}, as well as the results in this thesis.

In \cite{min-max-entropy}, the authors give a nice direct interpretation of both the min- and the max-entropy: given a state $\rho^{AB}$ the conditional min-entropy $H_{\min}(A|B)_{\rho}$ quantifies how close to a maximally entangled state we can make $\rho^{AB}$ by applying an arbitrary CPTP map on $B$:
\[ 2^{-H_{\min}(A|B)_{\rho}} = |A| \max_{\mathcal{F}^{B \rightarrow A'}} F\left( (\ident^A \otimes \mathcal{F})(\rho^{AB}), \Phi^{AA'} \right)^2 \]
where $\mathcal{F}$ ranges over all CPTP maps from $\LL(\sfB)$ to $\LL(\sfA')$, and $A'$ is a quantum system of the same dimension as $A$.  Likewise, the max-entropy $H_{\max}(A|B)_{\rho}$ characterizes how close the state is to being decoupled and uniform on $A$:
\[ 2^{H_{\max}(A|B)_{\rho}} = |A| \max_{\sigma^B \in \DD(\sfB)} F\left( \rho^{AB}, \pi^A \otimes \sigma^B \right). \]
It can also be shown that $H_{\min}(A|B)_{\rho} \leqslant H(A|B) \leqslant H_{\max}(A|B)_{\rho}$ (\cite{tcr08}, Lemma 2).

One problem with all of the above quantities is that they are very sensitive to small variations in the state on which they are defined, whereas most of the quantities that we are bounding with them are not. Hence, if we use these quantities directly, we can end up with very poor bounds in certain cases. For this reason, we define ``smooth'' versions of these entropies. Instead of computing the entropic quantities directly on the state we are given, we optimize them over an $\varepsilon$-ball around the state; this idea was introduced by Renner and Wolf in \cite{smoothing}. For any $\rho^{AB} \in \DD(\sfA \otimes \sfB)$, define
\[ \mfB(\rho, \varepsilon) := \{ \tilde{\rho}^{AB} : \tr[\tilde{\rho}] \leqslant 1, d_F(\rho, \tilde{\rho}) \leqslant \varepsilon \}. \]
We then define the following quantities:

\begin{defin}[Smooth conditional min-entropy]\label{def:smooth-cond-min-entropy}
Let $\rho^{AB}$ be a quantum state. Then, the $\varepsilon$-smooth conditional min-entropy of $A$ given $B$ is defined as
\[ H^{\varepsilon}_{\min}(A|B)_{\rho} := \max_{\sigma^{AB} \in \mfB(\rho,\varepsilon)} H_{\min}(A|B)_{\sigma}. \]
\end{defin}

\begin{defin}[Smooth conditional 2-entropy]\label{def:smooth-cond-2-entropy}
Let $\rho^{AB}$ be a quantum state. Then, the $\varepsilon$-smooth conditional 2-entropy of $A$ given $B$ is defined as
\[ H^{\varepsilon}_{2}(A|B)_{\rho} := \max_{\sigma^{AB} \in \mfB(\rho, \varepsilon)} H_{2}(A|B)_{\sigma}. \]
\end{defin}

\begin{defin}[Smooth conditional max-entropy]\label{def:smooth-cond-max-entropy}
Let $\rho^{AB}$ be a quantum state. Then, the $\varepsilon$-smooth conditional max-entropy of $A$ given $B$ is defined as
\[ H^{\varepsilon}_{\max}(A|B)_{\rho} := \min_{\sigma^{AB} \in \mfB(\rho, \varepsilon)} H_{\max}(A|B)_{\sigma}. \]
\end{defin}

As mentioned before, these quantities reduce to von Neumann quantities in the i.i.d.\ case. This is formalized in the following theorem, called the \emph{fully quantum asymptotic equipartition property} \cite{tcr08} by Tomamichel, Colbeck and Renner:

\begin{thm}[Fully Quantum Asymptotic Equipartition Property]\label{thm:fully-quantum-aep}
	Let $\rho^{AB}$ be a density operator, $\varepsilon > 0$, $\eta \leqslant 2^{-\demi H_{\min}(A|B)_{\rho}} + 2^{\demi H_{\max}(A|B)_{\rho}} + 1 \leqslant 2\sqrt{|A|} + 1$, and $n \in \mbN$. Then, if $n \geqslant \frac{8}{5} \log \frac{2}{\varepsilon^2}$,
	\[ \frac{1}{n} H^{\varepsilon}_{\min}(A^n | B^n)_{\rho^{\otimes n}} \geqslant H(A|B)_{\rho} - 4 \log \eta \sqrt{\frac{\log(2/\varepsilon^2)}{n}}. \]
\end{thm}

\section{Quantum channel capacities}
There are many variants of quantum channel capacities that can be defined, reflecting the large number of possible data transmission scenarios in which quantum channels can be useful. The two main ones are the classical capacity (the best rate at which we can send classical data through a quantum channel) and the quantum capacity (the best rate at which we can send arbitrary qubits through the channel). We can also define entanglement-assisted capacities of these two problems, in which the sender and the receiver share an arbitrary number of EPR pairs that they can use for free to help them transmit either classical or quantum data through the quantum channel, as the case may be. We will not be concerned with the classical capacity of quantum channels in this thesis, however, and the entanglement-assisted classical capacity is simply twice the entanglement-assisted quantum capacity. We shall therefore only talk about the unassisted and entanglement-assisted quantum capacities.

Furthermore, we can either consider what we can do with a single use of an arbitrary channel (which is the most general version of the problem), or we can restrict ourselves to i.i.d.\ channels (i.e.\ $n$ copies of a relatively small channel $\mathcal{N}$). The i.i.d.\ case, in addition to being a practically relevant special case, is also typically much easier to solve. We will consider both problems in this thesis: for all of the problems that we will consider, we will first prove a theorem for a single use of an arbitrary channel, and we will then apply it to an i.i.d.\ channel. The goal of this section is to define these problems and say a few words about them.

\subsection{One-shot capacities}
We first consider the simpler case of one-shot capacity. (Simpler to define, not to solve!) By the term ``one-shot'', we mean that we will consider protocols involving a single use of a channel, as opposed to using the same channel $n$ times. Suppose Alice would like to send an arbitrary quantum system $M$ to Bob using the channel $\mathcal{N}^{A' \rightarrow C}$ a single time. Alice will therefore encode her message $M$ into the channel input $A'$ using some encoding CPTP map $\mathcal{E}^{M \rightarrow A'}$. Upon receiving the channel output $C$, Bob will attempt to recover $M$ using a decoding CPTP map $\mathcal{D}^{C \rightarrow M}$. One would like to make sure that, regardless of the actual state of the message system $M$, Bob gets that same state at the output. When we consider every possible state of $M$, we must also include cases in which the contents of $M$ are entangled with another system. While $M$ can be entangled with an arbitrarily large system, it is mathematically equivalent to consider only entanglement with another system $R$ of dimension $|R| = |M|$. 

Of course, since we are only using the channel once, one cannot hope in general to have no error whatsoever. We must therefore decide on an error level that we are willing to tolerate, and then look at how big a message we can transmit given this constraint.

Taking all this into consideration, our goal is to find an encoder-decoder pair that satisfies
\begin{equation*}
\left\| (\mathcal{D} \circ \mathcal{N} \circ \mathcal{E})(\psi^{RM}) - \psi^{RM} \right\|_1 \leqslant \varepsilon
\end{equation*}
for every pure state $\psi^{RM}$. An alternative way of writing this is via the \emph{diamond norm} on superoperators \cite{diamond-norm}:
\begin{equation}\label{eqn:def-quantum-transmission}
\left\| \mathcal{D} \circ \mathcal{N} \circ \mathcal{E} - \ident^M \right\|_{\diamond} \leqslant \varepsilon.
\end{equation}
In other words, the composition of the encoder, channel, and decoder, must be nearly indistinguishable from the identity channel.

This quantity, however, is rather difficult to bound directly because of the optimization over the input state. Fortunately, there exists an essentially equivalent criterion which is much easier to establish for a given protocol. Instead of considering the worst case input, one can consider only a fixed maximally entangled state $\Phi^{RM}$ between $R$ and $M$:
\begin{equation}\label{eqn:def-entanglement-generation}
\left\| (\mathcal{D} \circ \mathcal{N} \circ \mathcal{E})(\Phi^{RM}) - \Phi^{RM} \right\|_{1} \leqslant \varepsilon.
\end{equation}
Requiring that an encoder and decoder fulfill this condition is weaker, but it turns out that by slightly reducing the dimension of the input system of the channel, one can turn an encoder-decoder pair that fulfills (\ref{eqn:def-entanglement-generation}) into one that fulfills (\ref{eqn:def-quantum-transmission}) \cite{tema-con-variazioni}.

Alice and Bob might also have EPR pairs at their disposal to help them increase the transmission rate; we call this the \emph{entanglement-assisted} capacity. The setting is the same as above, except that Alice and Bob start out with an additional state $\Phi^{\wtA B}$ that they are allowed to consume at will to help them in their task. We now have an encoder $\mathcal{E}^{M \wtA \rightarrow A'}$ and a decoder $\mathcal{D}^{CB \rightarrow M}$, and we want to ensure that
\begin{equation*}
\left\| (\mathcal{D} \circ \mathcal{N} \circ \mathcal{E})(\Phi^{RM} \otimes \Phi^{\wtA B}) - \Phi^{RM} \right\|_{1} \leqslant \varepsilon.
\end{equation*}

Classical one-shot capacities were first considered by Han and Verdú \cite{han-verdu}, who pioneered the so-called information-spectrum approach to the capacity of general non-i.i.d.\ channels. Using an approach that is much closer to the one used in this thesis, Renner, Wolf and Wullschleger \cite{single-serving} used classical versions of min- and max-entropies to derive bounds for the one-shot capacity of classical channels. On the quantum side, Buscemi and Datta \cite{buscemi-datta} consider the one-shot capacity using different tools from the ones used here.

\subsection{Capacities of memoryless channels}
We now wish to consider coding for channels of the form $\mathcal{N}^{\otimes n}$, where $n$ grows arbitrarily. In this case, we will want to find the best rate (number of qubits sent divided by number of channel uses) at which we can send quantum data such that the error rate goes to zero as $n \rightarrow \infty$. We call such channels \emph{memoryless channels}, since the channel behaves exactly the same way from one use to the next without ``remembering'' previous inputs; we also sometimes call this the ``i.i.d.\ case''. The definitions in this case are slightly more involved due to the fact that we need a series of encoders and decoders that grows with $n$. We begin by defining the unassisted quantum capacity:

\begin{defin}[Quantum code]
	An $(n,R)$-code for a quantum channel $\mathcal{N}^{A' \rightarrow C}$ is an encoding superoperator $\mathcal{E}^{M \rightarrow {A'}^{n}}$ and a decoding superoperator $\mathcal{D}^{C^{n} \rightarrow M}$, where $M$ is a $2^{nR}$-dimensional quantum system.
\end{defin}

\begin{defin}[Achievable rate]
	A rate $R$ is said to be achievable for a channel $\mathcal{N}^{A' \rightarrow C}$ if there exists a sequence of $(n, R)$-codes $(\mathcal{E}_n, \mathcal{D}_n)$ such that
	\begin{equation}
		\lim_{n \rightarrow \infty} \left\| \left( \mathcal{D}_n \circ \mathcal{N}^{\otimes n} \circ \mathcal{E}_n \right) - \ident^M \right\|_{\diamond} = 0.
	\end{equation}
\end{defin}

\begin{defin}[Quantum capacity]
	The quantum capacity $Q(\mathcal{N})$ of a quantum channel $\mathcal{N}$ is the supremum of all achievable rates for this channel.
\end{defin}

Despite considerable efforts, we do not yet have a satisfactory characterization of the quantum capacity. We do have a general lower bound for the capacity \cite{lsd1,lsd2,lsd3} which is given by the coherent information (see Theorem \ref{thm:direct-vanilla-channels}). This bound is known not to be tight, however \cite{superadd-qcap}, and a very strange phenomenon appears: this capacity is not additive. More specifically, there exist pairs of quantum channels $\mathcal{N}$ and $\mathcal{M}$ such that the capacity of both $\mathcal{N}$ and $\mathcal{M}$ is zero, while the capacity of $\mathcal{N} \otimes \mathcal{M}$ is strictly positive \cite{smith-yard}.

We now turn to the definitions relevant for entanglement-assisted capacity:
\begin{defin}[Quantum entanglement-assisted code]
	An $(n,R,E)$-code for a quantum channel $\mathcal{N}^{A' \rightarrow C}$ is an encoding superoperator $\mathcal{E}^{M \wtA \rightarrow {A'}^{n}}$, and an associated decoding superoperator $\mathcal{D}^{C^{n} \wtB \rightarrow M}$, such that $|M| = 2^{nR}$ and $|\wtA| = |\wtB| = 2^{nE}$.
\end{defin}

\begin{defin}[Achievable rate]
	A rate $R$ is said to be achievable for a channel $\mathcal{N}^{A' \rightarrow C}$ if there exists a sequence of $(n, R, E)$-codes $(\mathcal{C}_n, \mathcal{D}_n)$ (for arbitrary finite $E$) such that 
	\begin{equation}
		\lim_{n \rightarrow \infty} \left\| \mathcal{M}_n - \ident^M \right\|_{\diamond}  = 0.
	\end{equation}
	where $\mathcal{M}_n$ is the superoperator $\mathcal{M}_n(\rho) = (\mathcal{D}_n \circ \mathcal{N}^{\otimes n} \circ \mathcal{E}_n)(\Phi^{\wtA \wtB} \otimes \rho)$.
\end{defin}

\begin{defin}[Entanglement-assisted quantum capacity]
	The entanglement-assisted quantum capacity $Q_E(\mathcal{N})$ of a quantum channel $\mathcal{N}$ is the supremum of all achievable rates for this channel.
\end{defin}

\subsection{Regularization and single-letter converses}
When we set out to characterize the capacity of a type of memoryless channel, we ultimately want to get an expression that can be efficiently computed from the description of the channel. Unfortunately, in quantum information theory, we seldom achieve this ideal. What usually happens is that we are able to give an easily computable \emph{achievable rate region}, meaning a set of transmission rates that we know can be achieved, and we can often give an uncomputable expression for the true capacity. In some cases, such as for the unassisted transmission of quantum data through quantum channels, we know that there is a gap between the two expressions \cite{superadd-qcap}. The same is true for the transmission of classical data through quantum channels \cite{hastings-nonadditivity}. In other cases, such as the entanglement-assisted transmission through quantum multiple-access channels \cite{qmac}, we do not know whether this is the case. Only in a few rare instances can we show that the two coincide; the main example is the entanglement-assisted quantum (and classical) capacities of quantum channels.

The achievable rate region and the uncomputable expression for the capacity usually take particular forms. For the sake of concreteness, we will consider the case of the transmission of quantum data through quantum channels, but the situation tends to be very similar in other settings. The best known achievable rate for this task is expressed in the following theorem (\cite{lsd1,lsd2,lsd3}, see also Theorem \ref{thm:direct-vanilla-channels} for a proof):

\begin{thm}\label{thm:lsd}
Let $\mathcal{N}^{A' \rightarrow C}$ be a quantum channel, let $\sigma^{AA'}$ be any pure state with $\sfA' \cong \sfA$, and let $\rho^{AC} = \mathcal{N}^{A' \rightarrow C}(\sigma)$. Then, any rate $R < I(A \rangle C)_{\rho}$ is achievable for the transmission of quantum data through $\mathcal{N}$.
\end{thm}

The main feature of this theorem is that it states the existence of protocols that send quantum data using the channel $n$ times, but whose rates can be computed by looking at a single instance of $\mathcal{N}$. Indeed, the state $\rho$ on which we compute $I(A \rangle C)_{\rho}$ is a state produced by a single application of $\mathcal{N}$. Furthermore, the proof of the theorem shows that these protocols can be constructed by choosing codes that ``look'' like the state $({\sigma^{A'}})^{\otimes n}$ at the channel input. It therefore gives us some information about the structure of codes that achieve the rates advertised in the theorem statement. 

The main question at this point is whether this theorem is optimal or not: is it possible to create codes that go beyond the highest rate this theorem can give? Since the above theorem holds for any channel, it is certainly possible to look at the rates we obtain for channels of the form $\mathcal{N}^{\otimes k}$ (i.e.\ if we regard $k$ uses of $\mathcal{N}$ as a single channel). If the above theorem were optimal, then doing this should never give us a higher rate than simply looking at $\mathcal{N}$ alone. In \cite{superadd-qcap}, the authors show that it is in fact possible to get a higher rate this way, thereby showing Theorem \ref{thm:lsd} to be suboptimal.

This raises a further question: can we in fact get the optimal rate only by using the above theorem on some large number of copies of the same channel, or do we need to do something altogether different? The answer is that taking many copies is sufficient, and is expressed in the following theorem:

\begin{thm}
Let $\mathcal{N}^{A' \rightarrow C}$ be a quantum channel. Then, the capacity of $\mathcal{N}$ is given by
\begin{equation}\label{eqn:lsd-reg-converse}
 C = \sup_{n, \sigma^{A{A'}^{n}}} \frac{1}{n}I(A \rangle C^n)_{\rho}
\end{equation}
where $\sigma^{A {A'}^{n}}$ ranges over all pure states, $\sfA \cong (\sfA')^{\otimes n}$, $\rho^{AC^n} = \mathcal{N}^{\otimes n}(\sigma)$ and $n$ ranges over all positive integers.
\end{thm}

This is what we call a \emph{regularized converse} or \emph{multiletter converse}: it is a converse of Theorem \ref{thm:lsd}, provided that we ``regularize'' it by considering many copies of the channel. This is not a very strong characterization of the capacity. One reason for this is that we cannot compute it: we have no bound on how large $n$ has to be to get within a given factor of the capacity. Another perhaps even more depressing reason is the way we prove this last theorem: we look at an arbitrary code achieving quantum transmission, use it as the state $\sigma$ in the above theorem, and show that the resulting $\frac{1}{n} I(A \rangle C)_{\rho}$ is lower-bounded by the rate of the code. Since there exists a code for every rate below the capacity (by definition), the right-hand side of Equation (\ref{eqn:lsd-reg-converse}) can never be lower than the capacity. This makes the above theorem nearly tautological: if we choose the best possible code, then we reach the capacity. It says nothing whatsoever about the structure of capacity-achieving codes, which is perhaps the main motivation for studying channel capacity problems.

As mentioned earlier, however, it is sometimes possible to prove that regularization is not necessary, and that a theorem that considers only one copy (like Theorem \ref{thm:lsd}) is optimal. When this is the case, we say that we have a \emph{single-letter converse}. The main example in quantum information theory is the entanglement-assisted quantum and classical capacities (see again Theorem \ref{thm:direct-vanilla-channels}). A further example is the entanglement-assisted quantum and classical capacities of quantum channels with side-information at the transmitter, which is studied in Chapter \ref{chp:side-info}, the single-letter converse being given as Theorem \ref{thm:side-info-single-letter}. When we have a single-letter converse, it means that the code structure used in the proof of the corresponding theorem is in fact the optimal way to code for this type of channels. We can then say that we have a good grasp on how the channel carries information.

Finding expressions for the various capacities that are easily computable and that give us information about the structure of optimal codes is one of the main goals of information theory, and it has generally been rather difficult to achieve in the quantum setting. Finding such expressions for the most basic quantum capacities (such as the unassisted quantum and classical capacities of quantum channels) is one of the most central open problems in the field today.

\section{The duality between vectors and operators}\label{sec:vec-op}
Periodically throughout this thesis it will be extremely useful to turn multipartite pure states into operators, and vice versa. This is simply a generalization of turning a ``ket'' into a ``bra'': if we have a vector $\ket{\psi} \in \sfA$, then we can turn it into an operator $\bra{\psi} \in \LL(\sfA,\mbC)$ from vectors to the complex numbers, by defining $\bra{\psi}$ as the only operator in $\LL(\sfA, \mbC)$ such that $\braket{\psi}{\varphi} = \langle \ket{\psi}, \ket{\varphi} \rangle$, where $\langle \cdot, \cdot \rangle$ denotes the inner product in $\sfA$.  We can turn multipartite states into more interesting operators, however. Endow $\sfA$ and $\sfB$ with standard orthonormal bases $\{ \ket{a_i}^A \}$ and $\{ \ket{b_i}^B \}$ respectively, and let $\op_{A \rightarrow B}: \sfA \otimes \sfB \rightarrow \LL(\sfA, \sfB)$ be defined as
\[ \op_{A \rightarrow B}(\ket{a_i}\ket{b_j}) = \ket{b_j}\bra{a_i} \hspace{1cm} \forall i,j.\]
This operation depends on the choice of standard basis; therefore, whenever it is used, a particular choice of basis is implied. Since this choice will never matter in this thesis, we shall not explicitly define these bases.

The following properties of the $\op$ transformation will be needed:
\begin{lem}\label{lem:op-switcheroo}
	Let $\ket{\psi}^{AB}$ and $\ket{\varphi}^{AC}$ be any vectors in $\sfA \otimes \sfB$ and $\sfA \otimes \sfC$ respectively. Then, $\op_{A \rightarrow B}(\ket{\psi}^{AB})\ket{\varphi}^{AC} = \op_{A \rightarrow C}(\ket{\varphi}^{AC})\ket{\psi}^{AB}$.
\end{lem}
\begin{proof}
	Let $\{ \ket{a_i} \}$, $\{ \ket{b_i} \}$, and $\{ \ket{c_i} \}$ be the canonical bases for $\sfA$, $\sfB$ and $\sfC$ respectively, and let
	\begin{align*}
		\ket{\psi}^{AB} &= \sum_{ij} \alpha_{ij} \ket{a_i}\ket{b_j}\\
		\ket{\varphi}^{AC} &= \sum_{ij} \beta_{ij} \ket{a_i}\ket{c_j}.
	\end{align*}
	Then,
	\begin{align*}
		\op_{A \rightarrow B}(\ket{\psi}^{AB}) \ket{\varphi}^{AC} &= \sum_{ijkl} \alpha_{ij} \beta_{kl} \ket{b_j}\braket{a_i}{a_k}\ket{c_l}\\
		&= \sum_{ijl} \alpha_{ij} \beta_{il} \ket{b_j}\ket{c_l}\\
		\op_{A \rightarrow C}(\ket{\varphi}^{AC}) \ket{\psi}^{AB} &= \sum_{ijkl} \alpha_{ij} \beta_{kl} \ket{c_l}\braket{a_k}{a_i}\ket{b_j}\\
		&= \sum_{ijl} \alpha_{ij} \beta_{il} \ket{b_j}\ket{c_l}.
	\end{align*}
\end{proof}

\begin{lem}\label{lem:op-identity}
	Let $\ket{\psi}^{AB}$ be any vector in $\sfA \otimes \sfB$, let $A'$ be a system of equal dimension to $A$, and let $\ket{\Phi}^{AA'} = \frac{1}{\sqrt{|A|}}\sum_i \ket{a_i}\ket{a'_i}$, where the $\ket{a_i}$'s and $\ket{a'_i}$'s are the canonical bases of $A$ and $A'$ respectively. Then,
	\[ \sqrt{|A|} \op_{A \rightarrow B}(\ket{\psi}^{AB})\ket{\Phi}^{AA'} = \ket{\psi}^{A'B}. \]
\end{lem}
\begin{proof}
	Let $\ket{\psi}^{AB} = \sum_{ij} \alpha_{ij} \ket{a_i}\ket{b_j}$; we then get that
	\begin{align*}
		\sqrt{|A|} \op_{A \rightarrow B}(\ket{\psi}^{AB}) \ket{\Phi}^{AA'} &= \sum_{ijk} \alpha_{ij} \ket{b_j}\braket{a_i}{a_k}\ket{a'_k}\\
		&= \sum_{ij} \alpha_{ij} \ket{a'_i}^{A'} \ket{b_j}^B\\
		&= \ket{\psi}^{A'B}.
	\end{align*}
\end{proof}

\begin{lem}
	For any $\ket{\psi} \in \sfA \otimes \sfB$ and any $M^{A \rightarrow C}$, we have that
	\begin{align*}
		\op_{B \rightarrow C}(M \ket{\psi}) &= M \op_{B \rightarrow A}(\ket{\psi})\\
		\op_{C \rightarrow B}(M \ket{\psi}) &= \op_{A \rightarrow B}(\ket{\psi}) M_T\\
	\end{align*}
where the $T$ subscript denotes transposition.
\end{lem}
\begin{proof}
	Let $\ket{\psi} = \sum_{ij} \alpha_{ij} \ket{a_i} \ket{b_j}$ and $M = \sum_{kl} \gamma_{kl} \ket{c_k} \bra{a_l}$. Then,
	\begin{align*}
		\op_{B \rightarrow C}(M \ket{\psi}) &= \op_{B \rightarrow C}\left( \sum_{ijkl} \alpha_{ij} \gamma_{kl} \ket{c_k} \braket{a_l}{a_i} \ket{b_j}  \right)\\
		&= \op_{B \rightarrow C}\left( \sum_{ijk} \alpha_{ij} \gamma_{ki} \ket{c_k} \ket{b_j}  \right)\\
		&= \sum_{ijk} \alpha_{ij} \gamma_{ki} \ket{c_k} \bra{b_j}.
	\end{align*}
	Likewise,
	\begin{align*}
		M \op_{B \rightarrow A}(\ket{\psi}) &= \sum_{ijkl} \alpha_{ij}\gamma_{kl} \ket{c_k}\braket{a_l}{a_i} \bra{b_j}\\
		&= \sum_{ijk} \alpha_{ij}\gamma_{ki} \ket{c_k} \bra{b_j}.
	\end{align*}
	The other statement is proven in the same manner:
	\begin{align*}
		\op_{C \rightarrow B}(M \ket{\psi}) &= \sum_{ijkl} \alpha_{ij} \gamma_{kl} \op_{C \rightarrow B}\left( \ket{c_k}\braket{a_l}{a_i} \ket{b_j} \right)\\
		&= \sum_{ijk} \alpha_{ij} \gamma_{ki} \ket{b_j} \bra{c_k}
	\end{align*}
	and
	\begin{align*}
		\op_{A \rightarrow B}(\ket{\psi}) M_T &= \sum_{ijkl} \alpha_{ij} \gamma_{kl} \ket{b_j} \braket{a_i}{a_l} \bra{c_k}\\
		&= \sum_{ijk} \alpha_{ij} \gamma_{ki} \ket{b_j} \bra{c_k}.
	\end{align*}
\end{proof}

\begin{lem}
	Let $\ket{\psi} \in \sfA \otimes \sfB$. Then, $\tr_B[\psi^{AB}] = \op_{B \rightarrow A}(\ket{\psi}) \op_{B \rightarrow A}(\ket{\psi})\mdag$.
\end{lem}
\begin{proof}
	Let $\ket{\psi} = \sum_{i} \alpha_i \ket{\psi_i}^A \ket{\varphi_i}^B$ be the Schmidt decomposition of $\ket{\psi}$.
	\begin{align*}
		\tr_B[\psi^{AB}] = \sum_i \alpha_i^2 \ketbra{\psi_i}^A
	\end{align*}
	and
	\begin{align*}
		\op_{B \rightarrow A}(\ket{\psi}) \op_{B \rightarrow A}(\ket{\psi})\mdag &= \sum_{ij} \alpha_i \alpha_j \ket{\psi_i} \braket{\varphi^*_i}{\varphi^*_j} \bra{\psi_j}\\
		&= \sum_{i} \alpha_i^2 \ketbra{\psi_i}^A.
	\end{align*}
\end{proof}

We will also need to turn operators into vectors through the same process. For any pair of systems $A$ and $B$, define $\vect: \LL(\sfA, \sfB) \rightarrow \sfA \otimes \sfB$ as the transformation:
\[ \vect(\ket{b_j}\bra{a_i}) = \ket{a_i}\ket{b_j}. \]
It is simply the inverse of $\op$.

We will need the following property of the $\vect$ transformation:
\begin{lem}\label{lem:vec-isometry}
	Let $M^{A \rightarrow B}$ and $N^{A \rightarrow B}$ be arbitrary operators. Then, $\tr[N\mdag M] = \vect(N)\mdag \vect(M)$.
\end{lem}
\begin{proof}
	Let $M = \sum_{ij} m_{ij} \ket{b_j} \bra{a_i}$ and $N = \sum_{ij} n_{ij} \ket{b_j} \bra{a_i}$. Then,
	\begin{align*}
		\tr[N\mdag M] &= \sum_{ijkl} \tr[m_{ij} n_{kl}^* \ket{a_k}\braket{b_l}{b_j}\bra{a_i}]\\ 
		&= \sum_{ij} m_{ij} n_{ij}^*
	\end{align*}
	and
	\begin{align*}
		\vect(N)\mdag \vect(M) &= \sum_{ijkl} m_{ij} n_{kl}^* \braket{a_k}{a_i}\braket{b_l}{b_j}\\
		&= \sum_{ij} m_{ij} n_{ij}^*.
	\end{align*}
\end{proof}

\chapter{The decoupling theorem}\label{chp:decoupling}
One peculiar feature of quantum information theory is that some of the simplest coding theorems that we know come from theorems that tell us how to \emph{remove} correlations, even though the goal of an error-correcting code is to establish correlations between the sender and the receiver. The basic idea is the following: to prove a coding theorem, we generally need to assert the existence of a decoder of some sort; this decoder must be able to reproduce a particular state with good fidelity given only partial or noisy information. By purifying all systems, we can consider all subsystems that are not held by the decoder. These will generally include a subsystem purifying the state that the decoder needs to produce, as well as systems considered as part of the environment or that we otherwise don't care about. It turns out that, in such a case, a decoder exists if and only if the system purifying the desired state and the ``environment'' are close to a product state. The theorem that ensures this is called Uhlmann's theorem, and is the subject of the next section.  Of course, for this approach to work, we need a way to ensure that two systems are close to a product state. Section \ref{sec:bertha} will present a very general decoupling theorem with which we will prove all of the coding theorems in this thesis.

Although some elements of this approach were already used earlier, this method came into its own with the discovery of the \emph{state merging} protocol \cite{state-merging}, and later, the \emph{Fully Quantum Slepian-Wolf (FQSW)} \cite{FQSW} protocol. A whole array of results, including the ``mother'' and ``father'' \cite{mother-father}, can be easily derived from either of these protocols, such as the quantum reverse Shannon theorem \cite{reverse-shannon}, the Lloyd-Shor-Devetak (LSD) theorem \cite{lsd1} \cite{lsd2} \cite{lsd3}, one-way entanglement distillation \cite{DW05}, and distributed compression \cite{FQSW}. This chapter will present a generalization of both FQSW and state merging that is much more flexible and which can therefore be used in more diversified contexts. 

\section{Uhlmann's theorem}
Before starting, we will need to formally define what we mean by \emph{purification}:
\begin{defin}[Purification]
	Let $\rho^A \in \DD(\sfA)$ be any normalized density operator. Then, a purification of $\rho^A$ is any normalized vector $\ket{\psi} \in \sfA \otimes \sfB$, with $B$ an arbitrary quantum system, such that $\tr_B[\ketbra{\psi}^{AB}] = \rho^A$. We then call $B$ the \emph{purifying system}.
\end{defin}
For any density operator, a purification exists, and is unique up to isometries on the purifying system.

Uhlmann's theorem was first shown in \cite{uhlmann}; the proof given here essentially follows the one in \cite{watrousnotes}.
\begin{thm}[Uhlmann]\label{thm:uhlmann}
Let $\rho^A$ and $\sigma^A$ be two quantum states, and let $\ket{\psi}^{AB}$ and $\ket{\varphi}^{AC}$ be purifications of $\rho^A$ and $\sigma^A$ respectively (the purifying systems $B$ and $C$ need not be isomorphic). Then,
\begin{equation}
F(\rho^A, \sigma^A) = \max_{V^{B\rightarrow C}} \left| \bra{\psi}V\mdag\ket{\varphi} \right|
\label{eqn:uhlmann-orig}
\end{equation} 
where the maximization is over all partial isometries from $B$ to $C$.
\end{thm}
\begin{proof}
Let $U^{A\rightarrow B}$ and $W^{A \rightarrow C}$ be partial isometries such that $\ket{\psi}^{AB} = \vect(U \sqrt{\rho})$ and $\ket{\varphi}^{AC} = \vect(W \sqrt{\sigma})$. Then,
\begin{align}
F(\rho^A, \sigma^A) &= \left\| \sqrt{\rho} \sqrt{\sigma} \right\|_1\\
   &= \left\| U \sqrt{\rho} \sqrt{\sigma} W\mdag \right\|_1\\
   \label{eqn:uhlmann-proof-line3} &= \max_{V^{B \rightarrow C}} \left| \tr \left[V U \sqrt{\rho} \sqrt{\sigma} W\mdag \right] \right|\\
   \label{eqn:uhlmann-proof-line4} &= \max_{V^{B \rightarrow C}} \left| \vect(VU \sqrt{\rho})\mdag \vect(W \sqrt{\sigma}) \right|\\ 
   &= \max_{V^{B \rightarrow C}} \left| \bra{\psi}V\mdag\ket{\varphi} \right|
\end{align}
where we have used Lemma \ref{lem:tracenorm-maxu} from the appendix on line (\ref{eqn:uhlmann-proof-line3}) and \ref{lem:vec-isometry} on line (\ref{eqn:uhlmann-proof-line4}).
\end{proof}

The main use of this theorem for coding purposes is that it often gives us a decoder ``for free''. Indeed, assume that, at the end of the execution of a channel coding protocol, we have a tripartite pure state $\ket{\psi}^{BER}$, with the three subsystems representing the shares of Bob, the environment, and a ``reference'' system which purifies the qubits that Alice wanted to send to Bob. Now, suppose that we were able to show that the environment is nearly uncorrelated with the reference: $F\left( \psi^{RE}, \rho^R \otimes \sigma^E \right) \geqslant 1 - \varepsilon$. Then, given a product purification $\ket{\varphi}^{R\bar{B}} \otimes \ket{\xi}^{E\hat{B}}$ of $\rho^R \otimes \sigma^E$, there exists a partial isometry $V^{B \rightarrow \bar{B} \hat{B}}$ such that $F\left( V \ket{\psi}^{BER}, \ket{\varphi}^{R \bar{B}} \otimes \ket{\xi}^{E \hat{B}} \right) \geqslant 1 - \varepsilon$.

Since we generally use the trace distance rather than the fidelity, the following corollary of Uhlmann's theorem (Lemma 2.2 in \cite{DHW05}) will be very useful to us:

\begin{cor}
    Let $\ket{\psi}^{AB}$ and $\ket{\varphi}^{AC}$ be two quantum states such that $\left\| \psi^A - \varphi^A \right\|_1 \leqslant \varepsilon$. Then there exists an isometry $U^{B \rightarrow C}$ such that $\left\| \left(U^{B \rightarrow C} \cdot \psi^{AB}\right) - \varphi^{AC} \right\|_1 \leqslant 2 \sqrt{\varepsilon}$.
\end{cor}
\begin{proof}
If $\| \psi^A - \varphi^A \|_1 \leqslant \varepsilon$, then by the Fuchs-van de Graaf inequalities \cite{fuchs-vdg} (Lemma \ref{lem:fuchs-vdg}) we have that $F(\psi^A, \varphi^A) \geqslant 1 - \demi \epsilon$. By Uhlmann's theorem, this means that there exists a partial isometry $U^{B \rightarrow C}$ such that $F(U^{B\rightarrow C} \cdot \psi^{AB}, \varphi^{AC}) \geqslant 1 - \demi \varepsilon$. A second application of the Fuchs-van de Graaf inequalities concludes the proof.
\end{proof}

\section{The decoupling theorem}\label{sec:bertha}
To be able to use Uhlmann's theorem to derive a coding scheme, we need a way to ensure that two quantum systems are nearly uncorrelated. The main theorem of this section will achieve this for us.

Suppose Alice holds the $A$ share of a mixed state $\rho^{AR}$. We would like to perform an operation on Alice's system to ensure that her share is decoupled from the reference. We will consider a very general operation: a fixed unitary transformation followed by an arbitrary completely positive superoperator $\mathcal{T}^{A \rightarrow E}$. We will show that if we choose the unitary transformation randomly according to the Haar measure (which can essentially be viewed as the uniform distribution over all unitaries), then the resulting protocol will on average perform well. This generalizes all of the decoupling theorems in the literature that the author is aware of, including the Fully Quantum Slepian-Wolf theorem \cite{FQSW}, which corresponds to the special case in which $\mathcal{T}$ traces out part of the system, as well as the state merging \cite{state-merging} theorem, in which $\mathcal{T}^{A \rightarrow EX}$ corresponds to making a rank-$|E|$ measurement and then storing the measurement result in the classical register $X$ and the residual quantum state in $E$. One advantage of this generalization is that it allows us to choose $\mathcal{T}$ to be a very complex operation; one especially interesting example is to pick $\mathcal{T}$ to be the complementary channel (the channel to the environment) of a channel we are interested in coding for. Another advantage is the use of (smooth) conditional 2-entropies rather than purities and dimension bounds as was done in all of these theorems (although, in the case of state merging, this was already done in \cite{diploma-berta} and \cite{merging-berta-etal}, and, in the case of FQSW, by Hayden in \cite{hayden-fqsw-2entropy}). This theorem allows to show directly that the environment is decoupled from any system of interest, which is usually what we need to show.

We will calculate how close the remaining state on $ER$ is to a product state in the main theorem of this section (Theorem \ref{thm:bertha}). To get to it, however, we will first need the following four technical lemmas. The first one is simply a trick that we will use to compute the trace of the square of a matrix:

\begin{lem}[Swap trick]\label{lem:swap-trick}
	Given two operators $M \in \LL(\mathsf{A})$ and $N \in \LL(\sfA)$, then $\tr[MN] = \tr[(M \otimes N) F]$, where $F$ swaps the two copies of the $A$ subsystem.
\end{lem}
\begin{proof}
Write $M$ and $N$ in the standard basis for $\sfA$: $M = \sum_{ij} m_{ij} \ket{i}\bra{j}$ and $N = \sum_{kl} n_{kl} \ket{k}\bra{l}$. Then,
	\begin{align}
		\tr[(M \otimes N) F] &= \tr\left[ \left( \sum_{ijkl} m_{ij} n_{kl} \ket{i}\bra{j} \otimes \ket{k}\bra{l} \right)F \right]\\
&= \tr\left[ \sum_{ijkl} m_{ij} n_{kl} \ket{i}\bra{l} \otimes \ket{k}\bra{j} \right]\\
&= \sum_{ij} m_{ij} n_{ji}\\
&= \tr[MN].
	\end{align}
\end{proof}

The second lemma involves averaging over Haar-distributed unitaries. While it would take us too far afield to formally introduce the Haar measure, it can simply be thought of as the uniform probability distribution over the set of all unitaries on a Hilbert space. The following then tells us the expected value of $U^{\otimes 2} \cdot M$ (with $M \in \LL(\sfA)$) when $U$ is selected ``uniformly at random'':
\begin{lem}\label{lem:haar-integral}
	Given an operator $M \in \LL(\sfA^{\otimes 2})$, we have that
	\begin{equation}
		\mbE(M) := \int_{\mbU(A)} (U^{\otimes 2} \cdot M)  dU = \alpha \ident^{A A'} + \beta F^{A}
	\end{equation}
	where $\alpha$ and $\beta$ are such that $\tr[M] = \alpha |A|^2 + \beta |A|$ and $\tr[MF] = \alpha |A| + \beta |A|^2$, and where $dU$ is the normalized Haar measure on $\mbU(A)$.
\end{lem}
\begin{proof}
	This is a standard result in Schur-Weyl duality. This is a special case of, for instance, Proposition 2.2 in \cite{collins-sniady}. To see this, note that Proposition 2.2 states that $\mbE : \LL(\sfA^{\otimes 2}) \rightarrow \LL(\sfA^{\otimes 2})$ is an orthogonal projection onto $\Span\{ \ident, F \}$ under the inner product $\langle A,B \rangle = \tr[A\mdag B]$. Hence, $\mbE(M)$ can be written as $\alpha \ident^{AA'} + \beta F^A$ as claimed, and the conditions $\tr[\ident \mbE(M)] = \tr[M]$ and $\tr[F \mbE(M)] = \tr[FM]$ must be fulfilled, and these lead to the two conditions on $\alpha$ and $\beta$.
\end{proof}

The following bounds the ratio of the purity of a bipartite state and the purity of the reduced state on one subsystem:
\begin{lem}\label{lem:tr2-bounded}
	Let $\xi^{AB} \in \Pos(\sfA \otimes \sfB)$ be any positive semidefinite operator. Then
	\begin{equation}
		\frac{1}{|A|} \leqslant \frac{\tr\left[ {\xi^{AB}}^2 \right]}{\tr\left[ {\xi^B}^2 \right]} \leqslant |A|.
	\end{equation}
\end{lem}
\begin{proof}
	Letting $A'$ be a system isomorphic to $A$, we first prove the left-hand side:
	\begin{align}
		\tr\left[ {\xi^B}^2 \right] &= \tr\left[ \tr_A\left[ \xi^{AB} \right]^2 \right]\\
		&= \tr\left[ \tr_A\left[ \xi^{AB} \right] \tr_{A'}\left[ \xi^{A'B} \right] \right]\\
		&= \tr\left[ \xi^{AB} \left( \tr_{A'}\left[ \xi^{A'B}\right]  \otimes \ident^A \right) \right]\\
		&= \tr\left[ (\xi^{AB} \otimes \ident^{A'})(\xi^{A'B}  \otimes \ident^A ) \right]\\
		&\leqslant \sqrt{\tr \left[ (\xi^{AB} \otimes \ident^{A'})^2 \right] \tr\left[ (\xi^{A'B} \otimes \ident^{A})^2 \right] }\\
		&= \tr\left[ {\xi^{AB}}^2 \otimes \ident^{A'} \right]\\
		&= |A| \tr\left[ {\xi^{AB}}^2 \right]
	\end{align}
	where the inequality is due to an application of Cauchy-Schwarz. The right-hand side follows from the fact that $\xi^{AB} \leqslant |A| \ident^A \otimes \xi^B$. This can in turn be seen from the fact that $|A| \ident^A \otimes \xi^B = \sum_{i=1}^{|A|^2} U_i^A \cdot \xi^{AB}$, where the $U_i$'s are Weyl operators with $U_1 = \ident$.
\end{proof}

In the main proof, we will need to bound the trace distance between two states using the 2-norm. The following lemma will allow us to do this:
\begin{lem}\label{lem:pseudo-jensen-renato}
Let $M \in \LL(\sfA)$ be any operator and let $\sigma \in \Pos(\sfA)$ be a positive definite operator. Then,
\begin{equation}
\| M \|_1 \leqslant \sqrt{\tr[\sigma] \tr[\sigma^{-1/4} M \sigma^{-1/2} M\mdag \sigma^{-1/4}]}.
\end{equation}
In particular, if $M$ is Hermitian, then
\begin{equation}
\| M \|_1 \leqslant \sqrt{\tr[\sigma] \tr[(\sigma^{-1/4} M \sigma^{-1/4})^2]}.
\end{equation}
\end{lem}
This is a slight generalization of Lemma 5.1.3 in \cite{renner-phd}; we give a different proof here for completeness:
\begin{proof}
\begin{align}
\| M \|_1 &= \max_{U} \left| \tr[UM] \right|\\
&= \max_U \left| \tr[(\sigma^{1/4} U \sigma^{1/4})(\sigma^{-1/4} M \sigma^{-1/4})] \right|\\
&\leqslant \max_U \sqrt{\tr[(\sigma^{1/4} U \sigma^{1/4})(\sigma^{1/4} U\mdag (\sigma^{1/4})] \tr\left[ \sigma^{-1/4} M \sigma^{-1/2}  M\mdag \sigma^{-1/4}\right]}\\
&= \sqrt{\max_U \tr[\sigma^{1/2} U \sigma^{1/2} U\mdag] \tr\left[ \sigma^{-1/4} M \sigma^{-1/2} M\mdag \sigma^{-1/4} \right]}\\
&= \sqrt{\tr[\sigma] \tr\left[ \sigma^{-1/4} M \sigma^{-1/2} M\mdag \sigma^{-1/4} \right]}
\end{align}
where the first equality is an application of Lemma \ref{lem:tracenorm-maxu} and the inequality results from an application of Cauchy-Schwarz, and the maximizations are over all unitaries on $A$. The last equality follows from
\begin{align*}
\max_U \tr[\sigma^{1/2} U \sigma^{1/2} U\mdag] &\leqslant \max_U \sqrt{\tr[\sigma] \tr[U \sigma^{1/2} U\mdag U \sigma^{1/2} U\mdag]}\\
&= \tr[\sigma]\\
&\leqslant \max_U \tr[\sigma^{1/2} U \sigma^{1/2} U\mdag].
\end{align*}
\end{proof}

We are now ready to prove the main theorem:

\begin{thm}\label{thm:bertha}
	Let $\rho^{AR}$ be a density operator, $\mathcal{T}^{A \rightarrow E}$ be any completely positive superoperator, and define $\omega^{A'E} := (\mathcal{T} \otimes \ident^{A'})(\Phi^{AA'})$. Then,
\begin{equation}\label{eqn:bertha-trd}
	\int_{\mbU(A)} \left\| \mathcal{T}(U \cdot \rho^{AR}) - \omega^E \otimes \rho^R \right\|_1 dU \leqslant 2^{-\demi H_2(A'|E)_{\omega} - \demi H_2(A|R)_{\rho}}
\end{equation}
where $\int \cdot dU$ denotes the integral over the Haar measure over unitaries $U^A$ acting on $A$.
\end{thm}
\begin{proof}
Throughout the proof, we will denote with a prime the ``twin'' subsystems used when we take tensor copies of operators, and $F^S$ denotes a swap between $S$ and $S'$.

We first use Lemma \ref{lem:pseudo-jensen-renato}. Letting $\sigma^E$ and $\zeta^R$ be any normalized, positive definite density matrices on $E$ and $R$ respectively, we get:
\begin{multline}
\left\| \mathcal{T}(U \cdot \rho^{AB}) - \omega^E \otimes \rho^R \right\|_1\\
\begin{split}
 &\leqslant \sqrt{\tr\left[ \left( (\sigma^E \otimes \zeta^R)^{-1/4} (\mathcal{T}(U \cdot \rho^{AR}) - \omega^E \otimes \rho^R)(\sigma^E \otimes \zeta^R)^{-1/4}\right)^2 \right]}.
\end{split}
\end{multline}
Define the CP map $\tilde{\mathcal{T}}^{A \rightarrow E}$ as $\tilde{\mathcal{T}}(\xi) = {\sigma^E}^{-1/4} \mathcal{T}(\xi) {\sigma^E}^{-1/4}$, the state $\tilde{\rho}^{AR}$ as $\tilde{\rho}^{AR} = {\zeta^R}^{-1/4} \rho^{AR} {\zeta^R}^{-1/4}$, and the state $\tilde{\omega}^{A'E}$ as $\tilde{\omega}^{A'E} = \tilde{\mathcal{T}}(\Phi^{A'A})$. We then rewrite the above as
\begin{equation}
\left\| \mathcal{T}(U \cdot \rho^{AR}) - \omega^E \otimes \rho^R \right\|_1 \leqslant \sqrt{\tr\left[ \left( (\mathcal{\tilde{T}}(U \cdot \tilde{\rho}^{AR}) - \tilde{\omega}^E \otimes \tilde{\rho}^R)\right)^2 \right]}.
\end{equation}
Using Jensen's inequality, we can get
\begin{equation}\label{eqn:bertha-inter1}
\int \left\| \mathcal{T}(U \cdot \rho^{AR}) - \omega^E \otimes \rho^R \right\|_1 dU \leqslant \sqrt{\int \tr\left[ \left( (\mathcal{\tilde{T}}(U \cdot \tilde{\rho}^{AR}) - \tilde{\omega}^E \otimes \tilde{\rho}^R)\right)^2 \right] dU}.
\end{equation}
We now simplify the integral:
\begin{multline}
	\int \tr \left[ \left( \tilde{\mathcal{T}}(U \cdot \tilde{\rho}^{AR}) - \tilde{\omega}^E \otimes \tilde{\rho}^R \right)^2 \right] dU\\
	\begin{split}
		&= \int \tr\left[ \left( \tilde{\mathcal{T}}(U \cdot \tilde{\rho}^{AR}) \right)^2 \right]dU - 2 \int \tr\left[ \tilde{\mathcal{T}}(U \cdot \tilde{\rho}^{AR}) \left( \tilde{\omega}^E \otimes \tilde{\rho}^R \right) \right]dU + \tr\left[ \left( \tilde{\omega}^E \otimes \tilde{\rho}^R \right)^2 \right]\\
		&= \int \tr\left[ \left( \tilde{\mathcal{T}}(U \cdot \tilde{\rho}^{AR}) \right)^2 \right]dU - 2 \tr\left[ \tilde{\mathcal{T}}\left(\int U \cdot \tilde{\rho}^{AR}dU\right)  \left( \tilde{\omega}^E \otimes \tilde{\rho}^R \right) \right] + \tr\left[ \left( \tilde{\omega}^E \otimes \tilde{\rho}^R \right)^2 \right]\\
		&= \int \tr\left[ \left( \tilde{\mathcal{T}}(U \cdot {\tilde{\rho}}^{AR}) \right)^2 \right]dU - \tr\left[ (\tilde{\omega}^E)^2 \right] \tr\left[ (\tilde{\rho}^R)^2 \right].
	\end{split}
\end{multline}

We attack the first term as follows:
\begin{multline}\label{eqn:attack-first-term}
	\int \tr\left[ \left( \tilde{\mathcal{T}}(U \cdot \tilde{\rho}^{AR}) \right)^2 \right]dU\\
\begin{split}
	&= \int \tr\left[ \left( \tilde{\mathcal{T}}(U \cdot \tilde{\rho}^{AR})  \right)^{\otimes 2} F^{ER}\right] dU\\
	&= \int \tr\left[ \left( \tilde{\mathcal{T}}^{\otimes 2}(U^{\otimes 2} \cdot (\tilde{\rho}^{AR})^{\otimes 2})  \right) F^{ER}\right] dU\\
	&= \int \tr\left[ (\tilde{\rho}^{AR})^{\otimes 2}  \left( \left\{ {U\mdag}^{\otimes 2} \cdot (\tilde{\mathcal{T}}\mdag)^{\otimes 2}(F^E) \right\} \otimes F^R \right) \right] dU\\
	&= \tr\left[ (\tilde{\rho}^{AR})^{\otimes 2} \left( \int \left\{ {U\mdag}^{\otimes 2} \cdot  (\tilde{\mathcal{T}}\mdag)^{\otimes 2}(F^{E})  \right\} dU \otimes F^R  \right)\right].\\
\end{split}
\end{multline}
where we have used Lemma \ref{lem:swap-trick} in the first equality, and the definition of the adjoint of a superoperator in the third equality. We now compute the integral using Lemma \ref{lem:haar-integral}:
\begin{equation}
	\int {U\mdag}^{\otimes 2} \cdot (\tilde{\mathcal{T}}\mdag)^{\otimes 2}(F^{E}) dU = \alpha \ident^{AA'} + \beta F^{A}
\end{equation}
where $\alpha$ and $\beta$ satisfy the following equations:
\begin{align}
	\alpha |A|^2 + \beta |A| &= \tr\left[ (\tilde{\mathcal{T}}\mdag)^{\otimes 2} (F^{E})  \right]\\
	&= \tr\left[ F^E (\tilde{\mathcal{T}})^{\otimes 2}(\ident^{AA'}) \right]\\
	&= |A|^2 \tr\left[ F^E (\tilde{\omega}^E)^{\otimes 2} \right]\\
	&= |A|^2 \tr\left[ (\tilde{\omega}^E)^2 \right]
\end{align}
and
\begin{align}
	\alpha |A| + \beta |A|^2 &= \tr\left[ (\tilde{\mathcal{T}}\mdag)^{\otimes 2}(F^{E})  F^A \right]\\
	&=  \tr\left[ F^{E} (\tilde{\mathcal{T}})^{\otimes 2}(F^A) \right]\\
	&= |A|^2 \tr\left[ F^E \tr_{AA'} \left[ (\tilde{\omega}^{AE})^{\otimes 2} (F^A \otimes \ident^{EE'}) \right] \right]\\
	&= |A|^2 \tr\left[ (\ident^{AA'} \otimes F^E) (\tilde{\omega}^{AE})^{\otimes 2}(F^A \otimes \ident^{EE'}) \right]\\
	&= |A|^2 \tr\left[ F^{AE} (\tilde{\omega}^{AE})^{\otimes 2} \right]\\
	&= |A|^2 \tr\left[ (\tilde{\omega}^{A'E})^2 \right].
\end{align}
where, $\tilde{\omega}^{AE}$ is simply $\tilde{\omega}^{A'E}$ with $A$ and $A'$ reversed. In the third equality, we have used the fact that $|A| \tilde{\omega}^{AE}$ is a Choi-Jamio\l{}kowski \cite{cj-choi,cj-jamiolkowski} representation of $\tilde{\mathcal{T}}$; the fourth equality is due to the fact that the adjoint of the partial trace is tensoring with the identity.

Solving this system of equations yields
\begin{align}
	\alpha &= \tr\left[ (\tilde{\omega}^E)^2 \right] \left( \frac{|A|^2 - \frac{|A|\tr\left[ (\tilde{\omega}^{A'E})^2 \right]}{\tr\left[ (\tilde{\omega}^E)^2 \right]}}{|A|^2-1} \right)\\
	\beta &= \tr\left[ (\tilde{\omega}^{A'E})^2 \right] \left( \frac{|A|^2 - \frac{|A|\tr\left[ (\tilde{\omega}^E)^2 \right]}{\tr\left[ (\tilde{\omega}^{A'E})^2 \right]}}{|A|^2-1} \right).
\end{align}
By applying Lemma \ref{lem:tr2-bounded}, we can simplify this to $\alpha \leqslant \tr\left[ (\tilde{\omega}^E)^2 \right]$ and $\beta \leqslant \tr\left[ (\tilde{\omega}^{A'E})^2 \right]$. Substituting this into (\ref{eqn:attack-first-term}) and using Lemma \ref{lem:swap-trick} twice, and then substituting into (\ref{eqn:bertha-inter1}) yields
\begin{equation}
\int \left\| \mathcal{T}(U \cdot \rho^{AR}) - \omega^E \otimes \rho^R \right\|_1 dU \leqslant \sqrt{\tr\left[ (\tilde{\omega}^{A'E})^2 \right] \tr\left[ (\tilde{\rho}^{AR})^2 \right]}.
\end{equation}

We then get the theorem by using the definitions of $\tilde{\omega}$, $\tilde{\rho}$ and the definition of $H_2$.
\end{proof}

We now prove a version of the theorem that allows us to replace the $H_2$ in the upper bound by the smoothed versions of $H_2$. Among other things, this allows us to use the fully quantum AEP (Theorem \ref{thm:fully-quantum-aep}) and therefore to use the theorem directly on i.i.d.\ states and channels.
\begin{thm}\label{thm:smooth-bertha}
	Let $\rho^{AR}$ be a density operator, $\mathcal{T}^{A \rightarrow E}$ be any completely positive superoperator, let $\omega^{A'E} = (\mathcal{T} \otimes \ident^{A'})(\Phi^{AA'})$, and let $\varepsilon > 0$. Then,
\begin{equation}
	\int_{\mbU(A)} \left\| \mathcal{T}(U \cdot \rho^{AR}) - \omega^E \otimes \rho^R \right\|_1 dU \leqslant 2^{-\demi H_{2}^{\varepsilon}(A'|E)_{\omega} - \demi H_{2}^{\varepsilon}(A|R)_{\rho}} + 8 \varepsilon
\end{equation}
where $\int \cdot dU$ denotes the integral over the Haar measure on all unitaries $U^A$.
\end{thm}
\begin{proof}
	Let $U_{\mathcal{T}}^{A \rightarrow CE}$ be a Stinespring extension of $\mathcal{T}$, and let $\widehat{\omega}^{A'E}$ be such that $d_F(\widehat{\omega}, \omega) \leqslant \varepsilon$ and $H_{2}(A'|E)_{\widehat{\omega}} = H_{2}^{\varepsilon}(A'|E)_{\omega}$. Also, let $\widehat{\rho}^{AR}$ be such that $d_F(\widehat{\rho}, \rho) \leqslant \varepsilon$ and $H_2(A|R)_{\widehat{\rho}} = H_2^{\varepsilon}(A|R)_{\rho}$. Write $\widehat{\omega} - \omega = \Delta_{+} - \Delta_{-}$ where $\Delta_{\pm} \in \Pos(\sfA' \otimes \sfE)$ have orthogonal support. Since $d_F(\widehat{\omega}, \omega) \leqslant \varepsilon$, $\| \widehat{\omega} - \omega \|_1 \leqslant 2\varepsilon$ (see Lemma \ref{lem:fuchs-vdg}) and $\left\| \Delta_{\pm} \right\|_1 \leqslant 2\varepsilon$. We now define $\widehat{\omega}' := \omega - \Delta_{-}$. By the definition of $H_2$ and the fact that $\widehat{\omega}' \leqslant \widehat{\omega}$, we have that $H_2(A'|E)_{\widehat{\omega}'} \geqslant H_2(A'|E)_{\widehat{\omega}}$.

	Let $P^C \leqslant \ident^C$ be a positive semidefinite operator such that $\tr_C[P U_{\mathcal{T}} \cdot \Phi^{AA'}] = \widehat{\omega}'$ (whose existence is guaranteed by Lemma \ref{lem:fiou2}, since $\widehat{\omega}' \leqslant \omega$) and define $\widehat{\mathcal{T}}(\xi) = \tr_C[P U_{\mathcal{T}} \cdot \xi]$. Then, using the previous theorem, we get
\begin{multline*}
2^{-\demi H_{2}^{\varepsilon}(A'|E)_{\omega} - \demi H_{2}^{\varepsilon}(A|R)_{\rho}}\\
\begin{split}
&\geqslant \int_{\mbU(A)} \left\| \widehat{\mathcal{T}}(U \cdot \widehat{\rho}^{AR}) - \widehat{\omega}'^E \otimes \widehat{\rho}^R \right\|_1 dU\\
&\geqslant \int_{\mbU(A)} \left\| \widehat{\mathcal{T}}(U \cdot \rho^{AR}) - \omega^E \otimes \rho^R \right\|_1 dU - 6\varepsilon.\\
&\geqslant \int_{\mbU(A)} \left\| \mathcal{T}(U \cdot \rho^{AR}) - \omega^E \otimes \rho^R \right\|_1 dU - \int_{\mbU(A)} \left\| \mathcal{T}(U \cdot \rho^{AR}) - \widehat{\mathcal{T}}(U \cdot \rho^{AR}) \right\|_1 - 6\varepsilon.\\
\end{split}
\end{multline*}
We now deal with the second term above:
\begin{align*}
	\int_{\mbU(A)} \left\| \mathcal{T}(U \cdot \rho^{AR}) - \widehat{\mathcal{T}}(U \cdot \rho^{AR}) \right\|_1 dU &= \int \left\| \tr_C\left[ (\ident^C - {P^C}^2)(U_{\mathcal{T}}U \cdot \rho^{AR}) \right] \right\|_1 dU\\
	&= \int \tr\left[ (\ident^C - {P^C}^2)(U_{\mathcal{T}}U \cdot \rho^{AR}) \right] dU\\
	&= \tr\left[ (\ident^C - {P^C}^2) \left(U_{\mathcal{T}} \cdot \int U \cdot \rho^{AR} dU \right) \right]\\
	&= \tr\left[ (\ident^C - {P^C}^2) \left(U_{\mathcal{T}} \cdot \pi^A \otimes \rho^R \right) \right]\\
&= \tr\left[ \omega - \widehat{\omega}' \right]\\
&\leqslant 2\varepsilon,
\end{align*}
where the second equality follows from the fact that $P^C \leqslant \ident^C$. This results in:
\begin{equation*}
\int_{\mbU(A)} \left\| \mathcal{T}(U \cdot \rho^{AR}) - \omega^E \otimes \rho^R \right\|_1 dU \leqslant 2^{-\demi H_{2}^{\varepsilon}(A'|E)_{\omega} - \demi H_{2}^{\varepsilon}(A|R)_{\rho}} + 8\varepsilon
\end{equation*}
which concludes the proof.
\end{proof}

It is also possible to show that, with very high probability, the value of the left-hand side in Theorem \ref{thm:bertha} is very close to its expected value. This is shown in the next theorem. First, however, we must define the \emph{Lipschitz constant} of a function:

\begin{defin}[Lipschitz constant]
Let $f: \mfX \rightarrow \mfY$ be a function from the set $\mfX$ to the set $\mfY$ endowed with distance measures $d_{\mfX}$ and $d_{\mfY}$. Then, the Lipschitz constant of $f$ is defined as
\[ \sup_{x_1, x_2 \in \mfX} \frac{d_{\mfY}(f(x_1), f(x_2))}{d_{\mfX}(x_1,x_2)}. \]
If the above quantity is not bounded, the constant is not defined.
\end{defin}

\begin{thm}\label{thm:bertha-concentration}
In the scenario described in the statement of Theorem \ref{thm:bertha}, we have that
\begin{equation}
\Pr \left\{ \left\| \mathcal{T}(U \cdot \rho^{AR}) - \omega^E \otimes \rho^R \right\|_1  \geqslant 2^{-\demi H_2(A'|E)_{\omega} - \demi H_2(A|R)_{\rho}} + r \right\} \leqslant 2e^{-\frac{|A|r^2}{16K^2\|\rho^A\|_{\infty}}}
\end{equation}
where $K = \max \{ \| \mathcal{T}(X) \|_1 : X \in \Herm(\sfA), \|X\|_1 \leqslant 1 \}$, and where the probability is computed over the choice of $U$.
\end{thm}
\begin{proof}
This is a corollary of Corollary 4.4.28 in \cite{zeitouni-random-matrices}, which states that, for a $c$-Lipschitz function $f : \mbU(A) \rightarrow \mbR$,
\begin{equation}
\Pr_U \left\{ \left| f(U) - \mbE f \right| \geqslant \delta \right\} \leqslant 2e^{-|A| \delta^2/4c^2}.
\end{equation}
We are interested in the function $f(U) = \left\| \mathcal{T}(U \cdot \rho^{AR}) - \omega^E \otimes \rho^R \right\|_1$; we therefore need to bound its Lipschitz constant. Let $\ket{\rho}^{ABR}$ be a purification of $\rho$ and, without loss of generality, assume $f(U) \geqslant f(V)$, and
\begin{align*}
f(U) - f(V) &= \left\| \mathcal{T}(U \cdot \rho^{AR}) - \omega^E \otimes \rho^R \right\|_1 - \left\| \mathcal{T}(V \cdot \rho^{AR}) - \omega^E \otimes \rho^R \right\|_1\\
&\leqslant \left\| \mathcal{T}(U \cdot \rho^{AR}) - \mathcal{T}(V \cdot \rho^{AR}) \right\|_1\\
&\leqslant K \left\| U \cdot \rho^{ABR} - V\cdot \rho^{ABR} \right\|_1\\
&\leqslant 2K\left\| (U - V) \ket{\rho}^{ABR} \right\|_2\\
&= 2K \left\| \op_{BR \rightarrow A}( (U-V) \ket{\rho}^{ABR}) \right\|_2\\
&= 2K\left\| (U - V) \op_{BR \rightarrow A}(\ket{\rho}^{ABR}) \right\|_2\\
&\leqslant 2K\left\| U - V \right\|_2 \left\| \op_{BR \rightarrow A}(\ket{\rho}^{ABR}) \right\|_{\infty}\\
&= 2K\left\| U - V \right\|_2 \sqrt{\left\| \rho^A \right\|_{\infty}}.
\end{align*}
where the third inequality comes from Lemma \ref{lem:onenorm-twonorm}, the last inequality is an application of Lemma \ref{lem:norm-prod-matrices}, and $\| \cdot \|_{\infty}$ denotes the largest singular value of a matrix. Hence, the Lipschitz constant of $f$ is upper bounded by $2K \sqrt{\| \rho^A \|_{\infty}}$ and the theorem follows.
\end{proof}

All of the constructions given in this section involve selecting unitary operators randomly according to the Haar measure. This wouldn't be very practical in real life: there is no guarantee that a matrix chosen this way could be implemented efficiently by a quantum circuit; in fact there is an extremely high chance that it wouldn't be. However, for most of the above theorems, there is a way out of this: apart from Theorem \ref{thm:bertha-concentration}, the Haar measure can be replaced in all of the above theorems by a \emph{unitary 2-design}, which is defined as follows:

\begin{defin}[Unitary 2-design]
We call a finite set of unitaries $\mfD \subset \mbU(A)$ a \emph{unitary 2-design} if
\[ \frac{1}{|\mfD|} \sum_{U \in \mfD} U^{\otimes 2} \cdot M = \int_{\mbU(A)} (U^{\otimes 2} \cdot M) dU \]
for every $M \in \LL(\sfA^{\otimes 2})$, where the integral is taken over the Haar measure on $\mbU(A)$.
\end{defin}

Since all of the theorems of this section with the exception of Theorem \ref{thm:bertha-concentration} only involve the Haar measure in integrals of this type, the Haar measure can be replaced by a unitary 2-design without affecting the rest of the theorem statements. An example of a unitary 2-design is the Clifford group (the group of unitaries that take Pauli operators to Pauli operators) \cite{quantum-data-hiding}.

For Theorem \ref{thm:bertha-concentration}, which makes use of the concentration properties of the Haar measure, we do not yet know how to replace the Haar measure by something constructive.

\section{Corollaries of the decoupling theorem}
As mentioned previously, many well-known results can be shown to be special cases of this theorem, including the Fully Quantum Slepian-Wolf theorem \cite{FQSW}, as well as state merging \cite{state-merging}. We present them here for completeness:

\begin{cor}[FQSW \cite{FQSW}]\label{cor:fqsw}
	Let $\rho^{AR}$ be a mixed state, and $\sfA = \sfA_1 \otimes \sfA_2$. Then, we have that
	\begin{equation}
		\int \left\| \tr_{A_2}\left[ U \cdot \rho^{AR} \right] - \pi^{A_1} \otimes \rho^R \right\|_1 dU \leqslant \sqrt{ \frac{|A_1|}{|A_2|} 2^{-H_2(A|R)_{\rho}}}
	\end{equation}
	where the integral is over the Haar measure on $\mbU(A)$, and $\pi^{A_1}$ denotes the completely mixed state on $A_1$.
\end{cor}
\begin{proof}
	Consider the superoperator $\mathcal{T}^{A \rightarrow A_1}(\xi) = \tr_{A_2}[\xi]$ and define $\omega^{A'A_1} = \tr_{A_2}[\Phi^{AA'}]$, and then apply Theorem \ref{thm:bertha}.  It is easy to show that $\tr[{\omega^{A'A_1}}^2] = 1/|A_2|$, from which the result follows.
\end{proof}

\begin{cor}[State merging]\label{cor:state-merging}
	Let $\rho^{AR}$ be a mixed state, and let $\{ M_i^{A \rightarrow E} : M_i \in \LL(\sfA, \sfE), i \in \{1,\dots,n\} \}$ be a set of measurement operators (i.e.\ $\sum_i M_i\mdag M_i = \ident^A$) such that each $M_i$ is a rank-$|E|$ partial isometry. Then,
	\begin{equation}
		\int \sum_i \left\| M_i U \cdot \rho^{AR} - \frac{\pi^E}{n} \otimes \rho^R \right\|_1 dU \leqslant \sqrt{ |E| 2^{-H_2(A|R)_{\rho}}}
	\end{equation}
	where we integrate over $\mbU(A)$.
\end{cor}
For simplicity, we do not consider the case where $|A|$ is not divisible by $|E|$; the extension to the general case is straightforward.
\begin{proof}
	Let $X$ be an $n$-dimensional subsystem, and let $\mathcal{T}^{A \rightarrow EX}$ be a superoperator such that
	\begin{equation}
		\mathcal{T}(\sigma^A) = \sum_i \ketbra{i}^X \otimes (M_i^{A \rightarrow E} \cdot \sigma^{A})
	\end{equation}
	and define the state $\omega^{A'EX} = \mathcal{T}(\Phi^{AA'})$. It can easily be shown that $\tr[{\omega^{A'EX}}^2] = 1/n$, from which we get
	\begin{equation}
		\int \left\| \sum_i \ketbra{i}^X \otimes (M_i U \cdot \rho^{AR}) - \pi^{EX} \otimes \rho^R \right\|_1 dU \leqslant \sqrt{\frac{|E||X|}{n} 2^{-H_2(A|R)_{\rho}}}
	\end{equation}
	and the result follows. 
\end{proof}

The next corollary is an intermediate lemma from the state merging paper \cite{state-merging}, and also forms the basis of \cite{lsd-decoupling}:
\begin{cor}[Random subspaces]\label{cor:random-subspace}
	Let $\rho^{AR}$ be a mixed state and let $V^{A \rightarrow E}$ be a fixed rank-$|E|$ partial isometry. Then,
	\begin{equation}
		\int \left\| \frac{|A|}{|E|} VU \cdot \rho^{AR} - \pi^E \otimes \rho^R \right\|_1 dU \leqslant \sqrt{|E| 2^{-H_2(A|R)_{\rho}}}.
	\end{equation}
\end{cor}
\begin{proof}
	Consider the superoperator $\mathcal{T}^{A \rightarrow E}$ such that $\mathcal{T}(\sigma^A) = \frac{|A|}{|E|}V \cdot \sigma^A$. The result follows immediately from Theorem \ref{thm:bertha} and the fact that $\mathcal{T}(\Phi^{AA'})$ is a pure state.
\end{proof}

One can also come up with interesting blends of the above. For instance, one can mix FQSW and the random subspaces theorem above. The operation $\mathcal{T}$ we will consider is the following: we apply a fixed unitary operator on Alice's share, then restrict her system to only a subspace by applying a fixed projector, and then trace out part of the remaining state. We call the post-restriction system $E$, which is divided into two shares $E_1$ and $E_2$; we trace out $E_2$ and Alice is left with only $E_1$. The result is the following protocol, which will be presented and shown to be essentially optimal in \cite{merging-berta-etal}:
\begin{cor}\label{cor:aieul}
	Let $\rho^{AR}$ be a density operator, and let $\mathcal{T}^{A \rightarrow E_1}$ be a completely positive superoperator such that
	\begin{equation*}
		\mathcal{T}(\sigma^A) = \frac{|A|}{|E|}\tr_{E_2}[V^{A \rightarrow E_1E_2} \cdot \sigma^A]
	\end{equation*}
	where $V$ is a partial isometry, and $\sfE = \sfE_1 \otimes \sfE_2$. Then,
\begin{equation}\label{eqn:aieul-trd}
	\int_{\mbU(A)} \left\| \mathcal{T}(U \cdot \rho^{AR}) - \pi^{E_1} \otimes \rho^R \right\|_1 dU \leqslant \sqrt{\frac{|E_1|}{|E_2|} 2^{-H_2(A|R)_{\rho}}}
\end{equation}
where $\int \cdot dU$ denotes the integral over the Haar measure on all unitaries $U^A$.
\end{cor}
\begin{proof}
	Follows trivially from Theorem \ref{thm:bertha} and the calculation of $H_2(A'|E_1)_{\mathcal{T}(\Phi^{AA'})} = \log|E_2| - \log|E_1|$.
\end{proof}

\section{Quantum coding theorems via decoupling}
In this section, we use the theorems from Section \ref{sec:bertha} to derive a one-shot coding theorem for quantum channels. As explained earlier, our strategy will be to show that the complementary channel (i.e.\ the channel to the environment) completely breaks all correlations with a system that purifies the input. As a result, we get that Bob is able to reconstruct the message. 

We will consider the following problem: Alice and Bob share a pure state $\psi^{ABR}$; Alice holds $A$, Bob holds $B$, and $R$ is a reference system which purifies the state. Alice would like to send her share of the state to Bob through a single use of the quantum channel $\mathcal{N}^{A' \rightarrow C}$. To accomplish this, we need to find encoding and decoding superoperators $\mathcal{E}^{A \rightarrow A'}$ and $\mathcal{D}^{CB \rightarrow AB}$ such that
\begin{equation}\label{eqn:goal}
\left\| (\mathcal{D} \circ \mathcal{N} \circ \mathcal{E})(\psi) - \psi \right\|_1 \leqslant \varepsilon
\end{equation}

Note that Buscemi and Datta \cite{buscemi-datta} have considered a similar problem but without any $B$ system already at Bob's. The following generalizes their result to the case where Bob already has a share of the system:

\begin{thm}\label{thm:vanilla-channels-oneshot}
	Let $\psi^{ABR}$ be a pure state, $\mathcal{N}^{A' \rightarrow C}$ be any completely positive trace-preserving superoperator with Stinespring dilation $U_{\mathcal{N}}^{A' \rightarrow CE}$ and complementary channel $\bar{\mathcal{N}}^{A' \rightarrow E}$, let $\omega^{A''CE} = U_{\mathcal{N}} \cdot \sigma^{A''A'}$, where $\sigma$ is any pure state and $\sfA'' \cong \sfA'$ , and let $\varepsilon > 0$. Then, there exists an encoding partial isometry $V^{A \rightarrow A'}$ and a decoding superoperator $\mathcal{D}^{CB \rightarrow AB}$ such that
	\begin{equation*}
		\left\| \bar{\mathcal{N}}(V \cdot \psi^{AR}) - \omega^E \otimes \psi^R \right\|_1 \leqslant 2 \sqrt{\delta_1} + \delta_2
	\end{equation*}
	and
	\begin{equation*}
		\left\| (\mathcal{D} \circ \mathcal{N})(V \cdot \psi^{ABR}) - \psi^{ABR} \right\|_1 \leqslant 2\sqrt{2\sqrt{\delta_1} + \delta_2}
	\end{equation*}
	where
	\begin{align*}
		\delta_1 &= 3 \times 2^{\demi H_{\max}^{\varepsilon}(A)_{\psi} - \demi H_2^{\varepsilon}(A'')_{\omega}} + 24\varepsilon\\
		\delta_2 &= 3 \times 2^{-\demi H_2^{\varepsilon}(A''|E)_{\omega} - \demi H_2^{\varepsilon}(A|R)_{\psi}} + 24\varepsilon.
	\end{align*}
\end{thm}
Here, $\delta_1$ determines how closely the input $\psi^A$ can be made to fit the target input distribution $\omega^{A''}$, whereas $\delta_2$ depends on the difference between the amount of information that must be transmitted ($-H_2^{\varepsilon}(A|R)_{\psi}$) and the information-carrying capability of the channel ($H_2^{\varepsilon}(A''|E)_{\omega}$). See Figures \ref{fig:reg-1shot-real} and \ref{fig:reg-1shot-omega} for an illustration of the theorem.

\begin{figure}
    \centering
    \scalebox{1.0}{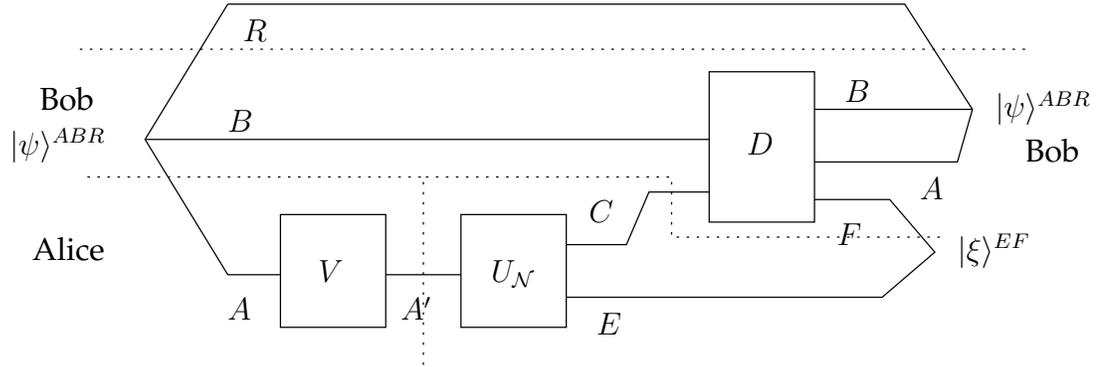}
    \caption{Diagram illustrating Theorem \ref{thm:vanilla-channels-oneshot}. Each line represents a quantum system, boxes represent isometries, and the horizontal axis represents the passage of time. Lines joined together at either end of the diagram represent pure states. Alice used $V$ to encode her message $A$ into the input to the channel $A'$, and Bob uses the channel output $C$ together with the $B$ that he had since the beginning to decode $A$ (and $B$) back. The decoder also produces a system $F$ that purifies the environment.}
    \label{fig:reg-1shot-real}
\end{figure}

\begin{figure}
    \centering
    \scalebox{1.0}{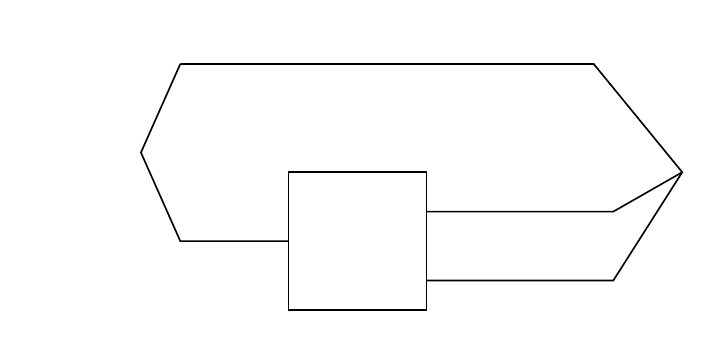}
    \caption{Diagram illustrating the states $\sigma$ and $\omega$ in Theorem \ref{thm:vanilla-channels-oneshot}. Each line represents a quantum system, boxes represent isometries, and the horizontal axis represents the passage of time. Lines joined together at either end of the diagram represent pure states.}
    \label{fig:reg-1shot-omega}
\end{figure}

\begin{proof}
	Let $W^{A \rightarrow A''}$ be any full-rank partial isometry, and consider the superoperator $\mathcal{T}^{A'' \rightarrow E}$ defined as $\mathcal{T}(\xi) = |A''|\bar{\mathcal{N}}( \op_{A'' \rightarrow A'}(\ket{\sigma}) \cdot \xi)$. Theorem \ref{thm:smooth-bertha} then tells us that
	\begin{multline*}
		\int \left\| |A''|\bar{\mathcal{N}}( \op_{A'' \rightarrow A'}(\ket{\sigma})U^{A''}W \cdot \psi^{AR}) - \omega^E \otimes \psi^R \right\|_1 dU\\
		\begin{split}
			&= \int \left\| \mathcal{T}(UW \cdot \psi^{AR}) - \omega^{E} \otimes \psi^R \right\|_1 dU\\
			&\leqslant 2^{-\demi H_2^{\varepsilon}(A''|E)_{\omega} - \demi H_2^{\varepsilon}(A|R)_{\psi}} + 8\varepsilon.
		\end{split}
	\end{multline*}
	We must now prove that there exists a $U$ such that conjugating $\psi$ by $\sqrt{|A''|} \op_{A'' \rightarrow A'}(\ket{\sigma}) UW$ can be approximated by an isometry and for which the above inequality holds. For this, we use Theorem \ref{thm:smooth-bertha} again, with $\mathcal{E}^{A'' \rightarrow G}$, $\mathcal{E}(\xi) = |A''|\tr[\op_{A'' \rightarrow A'}(\ket{\sigma}) \cdot \xi]$ ($G$ is a dummy 1-dimensional system):
	\begin{align*}
		\int \left\| |A''|\tr_{A'}[\op_{A'' \rightarrow A'}(\ket{\sigma})UW \cdot \psi^{ABR}] - \psi^{BR} \right\|_1 dU &\leqslant 2^{-\demi H_2^{\varepsilon}(A|BR)_{\psi} - \demi H_2^{\varepsilon}(A''|G)_{\mathcal{E}(\Phi)}} + 8\varepsilon\\
		&\leqslant 2^{\demi H_{\max}^{\varepsilon}(A)_{\psi} - \demi H_2^{\varepsilon}(A'')_{\omega}} + 8\varepsilon
	\end{align*}
	where we have used Lemma \ref{lem:op-identity} to establish that $\mathcal{E}(\Phi) = \omega$. Now, by Markov's inequality (Lemma \ref{lem:markov}), we can choose a $U$ such that
	\begin{align*}
		\int \left\| |A''|\tr_{A'}[\op_{A'' \rightarrow A'}(\ket{\sigma})UW \cdot \psi^{ABR}] - \psi^{BR} \right\|_1 dU &\leqslant 3 \times 2^{\demi H_{\max}^{\varepsilon}(A)_{\psi} - \demi H_2^{\varepsilon}(A'')_{\omega}} + 24\varepsilon\\
		\int \left\| |A''|\bar{\mathcal{N}}( \op_{A'' \rightarrow A'}(\ket{\sigma})  U^{A''}W \cdot \psi^{AR}) - \omega^E \otimes \psi^R \right\|_1 dU &\leqslant 3 \times 2^{-\demi H_2^{\varepsilon}(A''|E)_{\omega} - \demi H_2^{\varepsilon}(A|R)_{\psi}} + 24\varepsilon.
	\end{align*}
	
	The first of these two inequalities allows us to use Uhlmann's theorem (Theorem \ref{thm:uhlmann}) to find the encoding isometry: there exists a $V^{A \rightarrow A'}$ such that
	\begin{align*}
		\left\| |A''|\op_{A'' \rightarrow A'}(\ket{\sigma})UW \cdot \psi^{ABR} - V \cdot \psi^{ABR} \right\|_1 &\leqslant 2\sqrt{3 \times 2^{\demi H_{\max}^{\varepsilon}(A)_{\psi} - \demi H_2^{\varepsilon}(A'')_{\omega}} + 24\varepsilon}\\
		&= 2 \sqrt{\delta_1}.
	\end{align*}
	Using the triangle inequality and the monotonicity of trace distance under CPTP maps, we get that
	\begin{equation}\label{eqn:vanilla-oneshot-decoupling}
		\left\| \bar{\mathcal{N}}(V \cdot \psi^{AR}) - \omega^{E} \otimes \psi^R \right\|_1 \leqslant 2\sqrt{\delta_1} + \delta_2.
	\end{equation}
	To finish, we use Uhlmann's theorem again on this last inequality to get the decoder: there exists a partial isometry $D^{CB \rightarrow FAB}$ such that
	\begin{equation*}
		\left\| DU_{\bar{\mathcal{N}}}V \cdot \psi^{ABR} - \xi^{EF} \otimes \psi^{ABR} \right\|_1 \leqslant 2\sqrt{2\sqrt{\delta_1} + \delta_2}
	\end{equation*}
	for some state $\xi^{EF}$. Finally, we trace over $EF$ to get the theorem.
\end{proof}

We now turn to the case of memoryless channels used to transmit arbitrary quantum information with some entanglement assistance. This consists of the following special situation: the state $\psi^{ABR}$ has the form $\Phi^{RM} \otimes \Phi^{\wtA B}$ where $M\wtA$ plays the role of $A$ from the previous theorem, and the channel is $\left( \mathcal{N}^{A' \rightarrow C} \right)^{\otimes n}$. We will say that a pair $(Q, E)$ is achievable if there exists a sequence of codes of length $n$, with encoders $\mathcal{E}_n^{M_n \wtA_n \rightarrow {A'}^n}$ and decoders $\mathcal{D}_n^{C^n B_n \rightarrow M_n}$, with $|M_n| = |R_n| = 2^{nQ}$, $|\wtA_n| = |B_n| = 2^{nE}$, such that
\begin{equation*}
	\lim_{n \rightarrow \infty} \left\| \left( \mathcal{D}_n \circ \mathcal{N}^{\otimes n} \circ \mathcal{E}_n \right)\left( \Phi^{R_nM_n} \otimes \Phi^{\wtA_n B_n} \right) - \Phi^{R_n M_n} \right\|_1 = 0.
\end{equation*}

A rate $Q$ is achievable for entanglement-assisted transmission if there exists an $E\geqslant0$ such that $(Q,E)$ is achievable, and it is achievable for unassisted transmission if $(Q,0)$ is achievable. The capacity region is the closure of the convex hull of all achievable points.

The achievability of the coherent information for unassisted transmission was proven with increasing standards of rigour by Lloyd \cite{lsd1}, Shor \cite{lsd2}, and Devetak \cite{lsd3}. Since then, several other proofs have been published; the proof given below shares some similarities with the one by Hayden, Horodecki, Yard and Winter \cite{lsd-decoupling}. The entanglement-assisted capacity was first given by Bennett, Shor, Smolin and Thapliyal \cite{BSST02}. Theorem \ref{thm:direct-vanilla-channels}, which interpolates between those two results, can also be obtained by time-sharing between the completely assisted and unassisted protocols.

\begin{thm}\label{thm:direct-vanilla-channels}
	For any quantum channel $\mathcal{N}^{A' \rightarrow C}$, any pure state $\sigma^{AA'}$ with $\sfA' \cong \sfA$, the rate pair $(Q,E)$ is achievable for quantum transmission with rate-limited entanglement assistance through $\mathcal{N}$ if
	\begin{align*}
		Q + E &< H(A)_{\sigma}   && \text{and}  & Q - E &< I(A \rangle C)_{\mathcal{N}(\sigma)}.
	\end{align*}
	As a corollary, if we do not limit the rate of entanglement assistance, $Q < \demi I(A;C)_{\mathcal{N}(\sigma)}$ is achievable.
\end{thm}
The first condition (that $Q+E < H(A)_{\sigma}$) says that both the quantum information to be transmitted and Alice's share of the EPR pairs must fit into the input to the channel, and the second condition says that the channel can carry $I(A \rangle C)$ qubits per transmission when no entanglement is used, but the rate can be ``boosted'' at the rate of one ebit per qubit until the first condition is saturated. If we saturate the first condition, then we get the entanglement-assisted capacity of $\demi I(A;C)_{\mathcal{N}(\sigma)}$.
\begin{proof}
	The proof essentially consists of using the previous theorem on $\mathcal{N}^{\otimes n}$ and using the fully quantum AEP (Theorem \ref{thm:fully-quantum-aep}) to bound the various conditional entropies. Let $U_{\mathcal{N}}^{A' \rightarrow CE}$ be a Stinespring dilation of $\mathcal{N}$, and let $R$ and $M$ be subsystems of dimension $2^{nQ}$, $M$ storing the quantum message Alice wants to transmit, and $R$ being its purifying system. Likewise, let $\wtA$ and $B$ be systems storing Alice's and Bob's part of the shared entanglement respectively, both of dimension $2^{nE}$. Now, the input state we will consider is $\psi^{M \wtA B R} = \Phi^{RM} \otimes \Phi^{\wtA B}$, where $M \wtA$ play the role of $A$ from the previous theorem. We are now in a position to use the previous theorem with $\psi$ as the input state, $\mathcal{N}^{\otimes n}$ as the channel, and $\omega^{ {A}^{n}C^nE^n} = U_{\mathcal{N}}^{\otimes n} \cdot \sigma^{\otimes n}$ to conclude that there exists an isometry $V^{M \wtA \rightarrow {A'}^n}$ and a CPTP map $\mathcal{D}^{C^nB \rightarrow M}$ such that
	\begin{equation*}
		\left\| (\mathcal{D} \circ \mathcal{N}^{\otimes n})(V \cdot \psi^{M\wtA BR}) - \psi^{ABR} \right\|_1 \leqslant 2\sqrt{2\sqrt{\delta_1} + \delta_2}
	\end{equation*}
	where
	\begin{align*}
             \delta_1 &= 3 \times 2^{\demi H_{\max}^{\varepsilon}(M \wtA)_{\psi} - \demi H_2^{\varepsilon}({A}^{n})_{\omega}} + 24\varepsilon\\
	     \delta_2 &= 3 \times 2^{-\demi H_2^{\varepsilon}({A}^{n}|E^n)_{\omega} - \demi H_2^{\varepsilon}(M \wtA|R)_{\psi}} + 24\varepsilon
	\end{align*}

	Now we simply need to ensure that both $\delta_1$ and $\delta_2$ go down to zero as $n \rightarrow \infty$. We then have that:
\begin{align}
H_{\max}^{\varepsilon}(M \wtA)_{\psi} &\leqslant nQ + nE\\
H_2^{\varepsilon}({A}^n)_{\omega} &\geqslant n\left[H(A)_{\sigma} - \Delta_1\right]\\
H_2^{\varepsilon}({A}^n|E^n)_{\omega} &\geqslant n\left[H(A|E)_{U_{\mathcal{N}} \cdot \sigma} - \Delta_2\right] = n\left[I(A \rangle C)_{\mathcal{N}(\sigma)} - \Delta_2\right]\\
H_2^{\varepsilon}(M \wtA|R)_{\psi} &\geqslant -nQ + nE.
\end{align}
where the $\Delta$ can be computed from the statement of the fully quantum AEP (Theorem \ref{thm:fully-quantum-aep}) and can be made arbitrarily close to zero for large enough $n$. Hence, we get that
\begin{align*}
	\delta_1 &= 3 \times 2^{\frac{n}{2}[Q + E - H(A)_{\sigma} + \Delta_1]} + 24\varepsilon\\
	\delta_2 &= 3 \times 2^{\frac{n}{2}[Q - E - I(A \rangle C)_{\mathcal{N}(\sigma)} + \Delta_2]} + 24\varepsilon.
\end{align*}

Hence, for any pair $(Q,E)$ such that $Q+E < H(A)_{\sigma}$ and $Q-E < I(A \rangle C)_{\mathcal{N}(\sigma)}$, there exists a protocol for which the error goes down to zero as $n \rightarrow \infty$.

To get the corollary on fully entanglement-assisted transmission, we simply add the two constraints to get $2Q < H(A)_{\sigma} + I(A \rangle C)_{\mathcal{N}(\sigma)} = I(A;C)_{\mathcal{N}(\sigma)}$ and hence $Q < \demi I(A;C)_{\mathcal{N}(\sigma)}$.
\end{proof}

\section{Destroying correlations by adding classical randomness}
In \cite{gpw04}, the authors discuss the following question: given a quantum state ${\rho^{AB}}^{\otimes n}$, how many bits of classical randomness must be added to it to turn it into a product state? By ``adding $k$ bits of randomness'' to a quantum state, we mean applying one of $2^k$ unitaries uniformly at random to either $A^n$ or $B^n$. They find a method such that, as long as $k \geqslant n[I(A;B)_{\rho} + \delta]$, the distance to a decoupled state goes to zero as $n \rightarrow \infty$. This theorem constitutes one of the first direct operational intepretations of the quantum mutual information for arbitrary density operators. We can recover both this result and a one-shot version of it from Theorem \ref{thm:vanilla-channels-oneshot}:

\begin{thm}
	Let $\rho^{AB} \in \DD(\sfA \otimes \sfB)$ be any quantum state. Then, for any $\varepsilon > 0$, there exists a set of $2^k$ unitaries $U_i^A$, with $k \leqslant 2H_{\max}^{\varepsilon}(A)_{\rho} + 4\log(1/\varepsilon) + 4$,  and a $\xi^A \in \DD(\sfA)$ such that
	\begin{equation*}
		\left\| 2^{-k} \sum_{i=1}^{2^k} U_i^A \cdot \rho^{AB} - \xi^A \otimes \rho^B \right\|_1 \leqslant 3 \times 2^{\demi[H_{\max}^{\varepsilon}(A)_{\rho} - H_2^{\varepsilon}(A|B)_{\rho} - k + 2\log(1/\varepsilon)]} + 2\sqrt{27 \varepsilon} + 24\varepsilon.
	\end{equation*}
\end{thm}
\begin{proof}
	Let $P^A$ be a projector onto a subspace of $\sfA$ of dimension $D \geqslant \sqrt{2^k}$, and let $\{ V^A_i \}_{i=1}^{2^k}$ be a set of Weyl operators (unitaries such that $\tr[V_i\mdag V_j] = 0$ for every $i \neq j$) on the support of $P^A$. Now, define the superoperator $\mathcal{T}^{A \rightarrow A}$ as
	\[ \mathcal{T}(\xi) = 2^{-k} \sum_{i=1}^{2^k} V_i \cdot \xi.  \]
	We can now apply Theorem \ref{thm:vanilla-channels-oneshot} with $\mathcal{T}$ playing the role of $\bar{\mathcal{N}}$, $\rho^{AB}$ as the input state and with input distribution $\sigma^{A''A} = \frac{|A|}{D}P^A \Phi^{A''A}P^A$ to get that there exists a unitary $U^A$ such that
	\begin{equation*}
		\left\| \mathcal{T}(U^A \cdot \rho^{AB}) - \omega^A \otimes \rho^B \right\|_1 \leqslant 2 \sqrt{\delta_1} + \delta_2
	\end{equation*}
	where $\omega^{A''A} = \mathcal{T}(\sigma^{A''A})$, and
	\begin{align*}
		\delta_1 &= 3 \times 2^{\demi H_{\max}^{\varepsilon}(A)_{\rho} - \demi \log D} + 24\varepsilon\\
		\delta_2 &= 3 \times 2^{-\demi H_{2}^{\varepsilon}(A|B)_{\rho} + \demi \log D - \demi k} + 24\varepsilon
	\end{align*}
	since $H_2(A'')_{\omega} = \log D$ and $H_2(A''|A)_{\omega} = -\log D + k$. In other words,
	\begin{equation*}
		\left\| 2^{-k} \sum_{i} V_iU^A \cdot \rho^{AB} - \omega^A \otimes \rho^B \right\|_1 \leqslant 2 \sqrt{\delta_1} + \delta_2.
	\end{equation*}
	We can now define $U_i^A := V_i^A U^A$ to get our desired set of unitaries. All that is left to do is to let $\log D = H_{\max}^{\varepsilon}(A)_{\rho} + 2\log(1/\varepsilon)$ to get that $\delta_1 = 27\varepsilon$, and hence, the theorem.
\end{proof}
One can also show using the fully quantum AEP (Theorem \ref{thm:fully-quantum-aep}) that, for an i.i.d.\ state ${\rho^{AB}}^{\otimes n}$, $H_{\max}^{\varepsilon}(A^n)_{\rho^{\otimes n}} - H_2^{\varepsilon}(A^n|B^n)_{\rho^{\otimes n}} \rightarrow n[H(A)_{\rho} - H(A|B)_{\rho}] = nI(A;B)_{\rho}$, which allows us to recover the theorem of \cite{gpw04}.

\chapter{Quantum channels with side information at the transmitter}
\label{chp:side-info}

\section{Introduction}\label{sec:side-info-intro}
Consider the following problem: we have a noisy quantum memory device that can store $n$ qubits and that contains a certain fraction of defective cells. The cells that do work can be modelled as depolarizing channels, but the defective ones always output $\ket{0}$. We can test which cells are defective before writing to the memory device, but this information is not necessarily available when reading from it. What is the best asymptotic rate at which we can store qubits reliably on this device? This problem can be generalized to any channel where the transmitter has access to side information about the channel state while the receiver does not.

The corresponding classical problem has been solved by Gel'fand and Pinsker in \cite{gelfand-pinsker}. They consider channels modelled as a conditional probability distribution $p_{Y|XS}(y|x,s)$, $x \in \mathcal{X}, s \in \mathcal{S}, y \in \mathcal{Y}$, where $x$, $y$ and $s$ represent the input, output and state of the channel respectively. The channel state is i.i.d.\ and distributed according to $p_S(s)$. The encoder has access to the entire sequence of channel states ahead of time whereas the decoder does not. They have shown that the capacity of such a channel is given by
\begin{equation}
	C = \max_{q_{USX} \in \mathcal{P}} \left[ I(U;Y) - I(U;S) \right]
	\label{eqn:gp-classique}
\end{equation}
where $\mathcal{P}$ is the set of all probability distributions on $\mathcal{U} \times \mathcal{X} \times \mathcal{S}$ such that the marginal on $\mathcal{S}$ is equal to $p_S(s)$; $\mathcal{U}$ is an arbitrary set that can be chosen such that $|\mathcal{U}| \leqslant |\mathcal{X}| + |\mathcal{S}|$. The mutual informations are computed for the distribution $p_{Y|XS} \cdot q_{USX}$.

Here we shall generalize this result to quantum channels and potentially quantum side-information using the methods developed in Chapter \ref{chp:decoupling}. Namely, we will prove that the entanglement-assisted quantum capacity of quantum channels with side information at the transmitter has the same form as (\ref{eqn:gp-classique}) and, a relatively rare fact in quantum information theory, has a single-letter converse. Along the way, we will prove a one-shot coding theorem as well as a coding theorem for quantum transmission with rate-limited entanglement assistance for such channels, both in the same spirit as those proven at the end of the last chapter.

\section{Definition of quantum channels with side information at the transmitter}

A quantum channel with side information at the transmitter is defined by a superoperator $\mathcal{N}^{A'S \rightarrow C}$ and a quantum state $\ket{\phi}^{SS'}$; this quantum state represents the side information. Alice has access to $S'$ and can input a state of her choice into $A'$. One way to view this is to say that Alice shares entanglement with the channel itself. This framework allows us to consider both quantum and classical side information about the channel in a unified manner.

To illustrate this, consider the example of the depolarizing channel with defects given in the introduction. For this case, we can choose $\ket{\phi}$ to be $\sqrt{p} \ket{00} + \sqrt{1-p} \ket{11}$. The superoperator $\mathcal{N}$ then measures the $S$ subsystem, and outputs $\ket{0}$ if the outcome is $0$. If the outcome is $1$, it applies the depolarizing channel to $A'$ and sends the output to Bob.

In this chapter, we will first be interested in the following one-shot task: Alice and Bob initially share the $A$ and $B$ parts of the state $\psi^{ABR}$, and Alice would like to use the channel $(\mathcal{N}^{A'S \rightarrow C}, \ket{\phi}^{SS'})$ to send $A$ to Bob. Hence, we want to ascertain the existence of an encoder $\mathcal{E}^{AS' \rightarrow A'}$ and decoder $\mathcal{D}^{CB \rightarrow AB}$ such that
\begin{equation*}
	\left\| \left( \mathcal{D} \circ \mathcal{N} \circ \mathcal{E} \right)\left( \psi^{ABR} \otimes \phi^{SS'} \right) - \psi^{ABR} \right\|_1 \leqslant \varepsilon
\end{equation*}
with a $\varepsilon$ small enough for our purposes.

We will then specialize our one-shot theorem to the i.i.d.\ case, in which the state $\psi^{M\wtA BR}$ has the form $\Phi^{RM} \otimes \Phi^{\wtA B}$ where $M\wtA$ plays the role of $A$, and the channel is $\mathcal{N}^{\otimes n}$. We will say that a pair $(Q, E)$ is achievable if there exists a sequence of codes of length $n$, with encoders $\mathcal{E}_n^{M_n \wtA_n {S'}^{n}\rightarrow {A'}^n}$ and decoders $\mathcal{D}_n^{C^n B_n \rightarrow M_n}$, with $|M_n| = |R_n| = 2^{nQ}$, $|\wtA_n| = |B_n| = 2^{nE}$, such that
\begin{equation*}
	\lim_{n \rightarrow \infty} \left\| \left( \mathcal{D}_n \circ \mathcal{N}^{\otimes n} \circ \mathcal{E}_n \right)\left( \Phi^{R_nM_n} \otimes \Phi^{\wtA_n B_n} \otimes (\phi^{SS'})^{\otimes n} \right) - \Phi^{R_n M_n} \right\|_1 = 0.
\end{equation*}

A rate $Q$ is achievable for entanglement-assisted transmission if there exists a $E\geqslant0$ such that $(Q,E)$ is achievable, and it is achievable for unassisted transmission if $(Q,0)$ is achievable. 

The capacity region is the closure of the convex hull of all achievable points.

The goal of this chapter is to establish the following theorems:
\begin{thm}
Let $(\mathcal{N}^{A'S \rightarrow C}, \ket{\phi}^{SS'})$ be a quantum channel with side-information at the transmitter, and let $\sigma^{AA'S}$ be any mixed state with $\sigma^S = \phi^S$. Then, any rate point $(Q,E)$ such that
\begin{align*}
Q + E &< H(A|S)_{\sigma}  &&\text{and} & Q - E &< I(A \rangle C)_{\mathcal{N}(\sigma)}
\end{align*}
is achievable for transmission with rate-limited entanglement assistance. As a corollary, any rate $Q$ such that $Q < \demi[ I(A;C)_{\mathcal{N}(\sigma)} - I(A;S)_{\sigma}]$ is achievable for entanglement-assisted transmission.
\end{thm}

\begin{thm}\label{thm:sideinfo-asst}
	The entanglement-assisted quantum capacity of a quantum channel with side information at the transmitter $( \mathcal{N}^{A'S \rightarrow C}, \ket{\phi}^{SS'})$ is
	\begin{equation}\label{eqn:cap}
		C = \sup_{\sigma} \left\{ \demi I(A;C)_{\omega} - \demi I(A;S)_{\sigma} \right\}.
	\end{equation}
	The supremum is taken over all mixed states of the form $\sigma^{AA'S}$ where $\sigma^S = \phi^S$ and $\omega^{AC} = \mathcal{N}^{A'S \rightarrow C}(\sigma^{AA'S})$. In other words, the previous theorem with an i.i.d.\ input distribution is optimal for coding for entanglement-assisted i.i.d.\ channels.
\end{thm}
This theorem also entails that the entanglement-assisted classical capacity of quantum channels with side information at the transmitter is
\begin{equation}
	C = \sup_{\sigma} \left\{ I(A;C)_{\mathcal{N}(\sigma)} - I(A;S)_{\sigma} \right\}
\end{equation}
via super-dense coding.

\section{Direct coding theorem}\label{sec:side-info-direct}
We begin with the one-shot coding theorem:

\begin{thm}\label{thm:side-info-channels-oneshot}
	Let $\psi^{ABR}$ be a pure state, $(\mathcal{N}^{A'S \rightarrow C},\ket{\phi}^{SS'})$ be any channel with side-information at the transmitter with $U_{\mathcal{N}}^{A'S \rightarrow CE}$ as Stinespring dilation, and let $\omega^{A''CED} = U_{\mathcal{N}} \cdot \sigma^{A''A'SD}$, where $\sigma$ is any pure state with $\sigma^S = \phi^S$. Then, there exists a encoding CPTP map $\mathcal{E}^{AS' \rightarrow A'}$ and a decoding CPTP map $\mathcal{D}^{CB \rightarrow AB}$ such that
	\begin{equation*}
		\left\| (\mathcal{D} \circ \mathcal{N} \circ \mathcal{E})(\psi^{ABR} \otimes \phi^{SS'})^{ABR} - \psi^{ABR} \right\|_1 \leqslant 2\sqrt{2\sqrt{\delta_1} + \delta_2}
	\end{equation*}
where
\begin{align*}
	\delta_1 &= 3 \times 2^{\demi H_{\max}^{\varepsilon}(A)_{\psi} - \demi H_2^{\varepsilon}(A''|S)_{\sigma}} + 24\varepsilon\\
	\delta_2 &= 3 \times 2^{-\demi H_{2}^{\varepsilon}(A''|ED)_{\omega} - \demi H_2^{\varepsilon}(A|R)_{\psi}} + 24\varepsilon.
\end{align*}
\end{thm}
Hence, to have a good code, one must ensure that both $\delta_1$ and $\delta_2$ are sufficiently small. Both of these quantities have natural interpretations: $\delta_1$ characterizes the difference between how ``big'' the message is ($H_{\max}(A)_{\psi}$) and how much space there is in the input to the channel ($H_{2}^{\varepsilon}(A''|S)_{\sigma}$), and $\delta_2$ depends on the difference between how hard the state is to transmit ($-H_2^{\varepsilon}(A|R)_{\psi}$) and how good the channel to the environment is at destroying correlations ($H_2^{\varepsilon}(A''|ED)_{\omega}$). See Figures \ref{fig:si-1shot-real} and \ref{fig:si-1shot-omega} for illustrations of the protocol.

\begin{figure}
    \centering
    \scalebox{0.9}{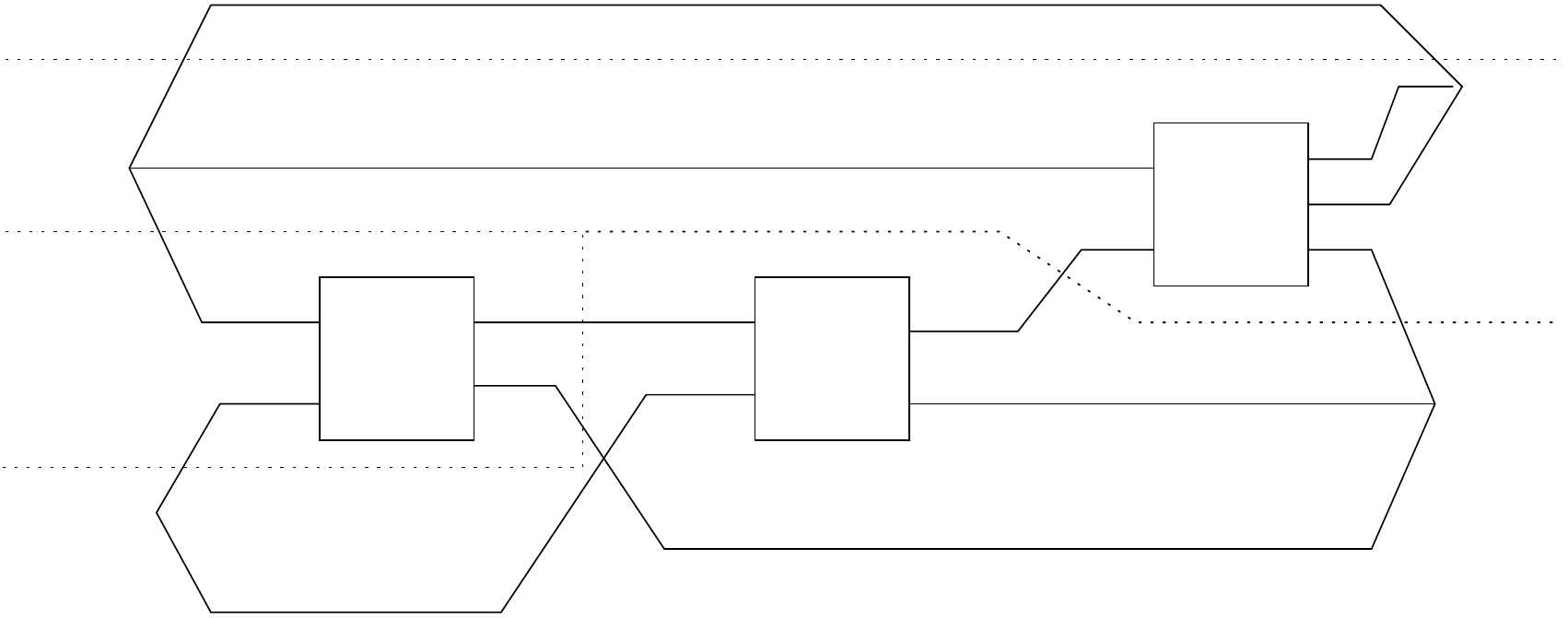}
    \caption{Diagram illustrating Theorem \ref{thm:side-info-channels-oneshot}, with encoder, channel and decoder purified. Each line represents a quantum system, boxes represent isometries, and the horizontal axis represents the passage of time. Lines joined together at either end of the diagram represent pure states. $V$ represents Alice's encoder: she uses the side information $S'$ to encode the message $A$ into the channel input $A'$ and discards a system $D$. The decoder $D$ takes the channel output $C$ together with Bob's initial system $B$ and produces $A$ and $B$ as output; the result being close to the initial state $\psi$.}
    \label{fig:si-1shot-real}
\end{figure}

\begin{figure}
    \centering
    \scalebox{1.0}{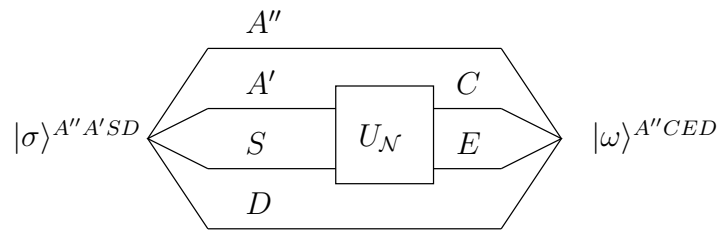}
    \caption{Diagram illustrating the states $\omega$ and $\sigma$ which define the input distribution in Theorem \ref{thm:side-info-channels-oneshot}. Each line represents a quantum system, boxes represent isometries, and the horizontal axis represents the passage of time. Lines joined together at either end of the diagram represent pure states.}
    \label{fig:si-1shot-omega}
\end{figure}

\begin{proof}
	Let $W^{A \rightarrow A''}$ any full-rank partial isometry, and consider the superoperator $\mathcal{T}^{A'' \rightarrow ED}$ defined as
	\[\mathcal{T}(\xi) = |A''|\tr_{C}[U_{\mathcal{N}} \op_{A'' \rightarrow A'SD}(\ket{\sigma}) \cdot \xi].\]
	Theorem \ref{thm:smooth-bertha} then tells us that
	\begin{align*}
		\int \left\| \mathcal{T}(UW \cdot \psi^{AR}) - \omega^{ED} \otimes \psi^R \right\|_1 dU &\leqslant 2^{-\demi H_2^{\varepsilon}(A''|ED)_{\omega} - \demi H_2^{\varepsilon}(A|R)_{\psi}} + 8\varepsilon.
	\end{align*}
	We must now prove that there exists a $U$ such that conjugating $\psi$ by $\sqrt{|A''|}\op_{A'' \rightarrow A'SD}(\ket{\sigma})UW$ can be approximated by an isometry of the form $V^{AS' \rightarrow A'D}$ acting on $\psi^{ABR} \otimes \phi^{SS'}$ and for which the above inequality holds. For this, we use Theorem \ref{thm:smooth-bertha} again, with $\mathcal{E}^{A'' \rightarrow S}$, $\mathcal{E}(\xi) = |A''|\tr_{A'D}[\op_{A'' \rightarrow A'SD}(\ket{\sigma}) \cdot \xi]$:
	\begin{multline*}
		\int \left\| |A''|\tr_{A'D}[\op_{A'' \rightarrow A'SD}(\ket{\sigma}) UW \cdot \psi^{ABR}] - \psi^{BR} \otimes \phi^S \right\|_1 dU\\
	\begin{split}
		&\leqslant 2^{-\demi H_2^{\varepsilon}(A|BR)_{\psi} - \demi H_2^{\varepsilon}(A''|S)_{\omega}} + 8\varepsilon\\
		&\leqslant 2^{\demi H_{\max}^{\varepsilon}(A)_{\psi} - \demi H_2^{\varepsilon}(A''|S)_{\omega}} + 8\varepsilon
	\end{split}
	\end{multline*}
	where we have used Lemma \ref{lem:op-identity} to deduce that $\mathcal{E}(\Phi) = \omega$. We would like to have a $U^{A''}$ that satisfies both inequalities. We can do this using Markov's inequality (Lemma \ref{lem:markov}): there exists a $U^{A''}$ such that:
	\begin{align*}
		\int \left\| \mathcal{T}(UW \cdot \psi^{AR}) - \omega^{ED} \otimes \psi^R \right\|_1 dU &\leqslant 3 \times 2^{-\demi H_2^{\varepsilon}(A''|E)_{\omega} - \demi H_2^{\varepsilon}(A|R)_{\psi}} + 24\varepsilon\\
&= \delta_2
\end{align*}
and
\begin{multline*}
		\int \left\| |A''|\tr_{A'D}[\op_{A'' \rightarrow A'SD}(\ket{\sigma}) U \cdot \psi^{ABR}] - \psi^{BR} \otimes \phi^S \right\|_1 dU\\
\begin{split}
 &\leqslant 3 \times 2^{\demi H_{\max}^{\varepsilon}(A)_{\psi} - \demi H_2^{\varepsilon}(A''|S)_{\omega}} + 24\varepsilon\\
&= \delta_1.
\end{split}
	\end{multline*}

This last condition allows us to use Uhlmann's theorem (Theorem \ref{thm:uhlmann}) to find the encoding isometry: there exists a $V^{AS' \rightarrow A'D}$ such that
	\begin{equation*}
		\left\| |A''|\op_{A'' \rightarrow A'SD}(\ket{\sigma})UW \cdot \psi^{ABR} - V \cdot (\psi^{ABR} \otimes \phi^{SS'}) \right\|_1 \leqslant 2\sqrt{\delta_1}
	\end{equation*}
	Using the triangle inequality and the monotonicity of trace distance under superoperators, we get that
	\begin{equation*}
		\left\| (U_{\mathcal{N}}V \cdot (\psi^{AR} \otimes \phi^{SS'}))^{EDR} - \omega^{ED} \otimes \psi^R \right\|_1 \leqslant 2\sqrt{\delta_1} + \delta_2.
	\end{equation*}

Finally, we use Uhlmann's theorem a second time to get a decoding partial isometry $D^{CB \rightarrow ABG}$:
\begin{equation*}
\left\| \left(D U_{\mathcal{N}}V \cdot \left(\psi^{ABR} \otimes \phi^{SS'}\right)\right) - \xi^{GED} \otimes \psi^{ABR} \right\|_1 \leqslant 2\sqrt{2\sqrt{\delta_1} + \delta_2}
\end{equation*}
for some state $\xi^{GED}$. We then take a partial trace over $GED$ inside the trace distance to get the theorem.
\end{proof}

We now move on to the memoryless case:
\begin{thm}\label{thm:side-info-iid}
Let $(\mathcal{N}^{A'S \rightarrow C}, \ket{\phi}^{SS'})$ be a quantum channel with side-information at the transmitter, and let $\sigma^{AA'S}$ be any mixed state with $\sigma^S = \phi^S$. Then, any rate point $(Q,E)$ such that
\begin{align*}
Q + E &< H(A|S)_{\sigma}   && \text{and} & Q - E &< I(A \rangle C)_{\mathcal{N}(\sigma)}
\end{align*}
is achievable for transmission with rate-limited entanglement assistance. As a corollary, any rate $Q$ such that $Q < \demi[ I(A;S)_{\sigma} - I(A;C)_{\mathcal{N}(\sigma)}]$ is achievable for entanglement-assisted transmission.
\end{thm}
Again, the first condition corresponds to how closely we can make the input fit the target input distribution $\sigma$, and the second one is the limit imposed by the channel noise. Once again, we can trade ebits for qubits at a one-to-one ratio until we reach the limit imposed by the first condition.
\begin{proof}
	The proof essentially consists of using the previous theorem on $\mathcal{N}^{\otimes n}$ and using the fully quantum AEP (Theorem \ref{thm:fully-quantum-aep}) to bound the various conditional entropies. Let $R$ and $M$ be subsystems of dimension $2^{nQ}$, $M$ storing the quantum message Alice wants to transmit, and $R$ being its purifying system. Likewise, let $\wtA$ and $B$ be systems storing Alice's and Bob's parts of the shared entanglement respectively, both of dimension $2^{nE}$. Now consider the input state $\psi^{M \wtA B R} = \Phi^{RM} \otimes \Phi^{\wtA B}$, where $M \wtA$ plays the role of $A$ in the one-shot theorem. We now use the previous theorem with $\psi$ as the input state, $\mathcal{N}^{\otimes n}$ as the channel, and $\omega^{ {A}^nC^nE^nD^n} = \mathcal{N}^{\otimes n}(\sigma^{\otimes n})$ to conclude that there exist encoding and decoding CPTP maps $\mathcal{E}^{M \wtA S' \rightarrow {A'}^n}$ and $\mathcal{D}^{C^n B \rightarrow AB}$ such that
	\begin{equation*}
		\left\| \left(\mathcal{D} \circ \mathcal{N}^{\otimes n} \circ \mathcal{E}\right)\left(\psi^{M \wtA BR} \otimes (\phi^{SS'})^{\otimes n}\right) - \psi^{ABR} \right\|_1 \leqslant 2\sqrt{2\sqrt{\delta_1} + \delta_2}
\end{equation*}
with
\begin{align*}
	\delta_1 &= 3 \times 2^{\demi H_{\max}^{\varepsilon}(M \wtA)_{\psi} - \demi H_2^{\varepsilon}(A^n|S^n)_{\sigma^{\otimes n}}} + 24\varepsilon\\
	\delta_2 &= 3 \times 2^{-\demi H_{2}^{\varepsilon}(A^n|E^nD^n)_{\omega} - \demi H_2^{\varepsilon}(M \wtA|R)_{\psi}} + 24\varepsilon.
\end{align*}

We can bound all the entropic terms above. We have that:
\begin{align}
H_{\max}^{\varepsilon}(M \wtA)_{\psi} &\leqslant nQ + nE\\
H_2^{\varepsilon}(A^n|S^n)_{\sigma^{\otimes n}} &\geqslant n\left[H(A|S)_{\sigma} - \Delta_1\right]\\
H_2^{\varepsilon}(A^n|E^nD^n)_{\omega} &\geqslant n\left[H(A|ED)_{\mathcal{N}(\sigma)} - \Delta_2\right] = n\left[I(A \rangle C)_{\mathcal{N}(\sigma)} - \Delta_2\right]\\
H_2^{\varepsilon}(M \wtA|R)_{\psi} &\geqslant -nQ + nE
\end{align}
where the $\Delta$ can be computed from the statement of the fully quantum AEP (Theorem \ref{thm:fully-quantum-aep})) and can be made arbitrarily close to zero for large enough $n$. 

Hence, as long as $Q + E < H(A|S)_{\sigma} - \Delta_1$ and $Q - E < I(A \rangle C)_{\mathcal{N}(\sigma)} - \Delta_2$, both $\delta_1$ and $\delta_2$ go down to zero as $n$ grows. The first condition corresponds to the fact that the message qubits and Alice's share of the entanglement must fit in the input, and the second condition means that the transmission rate minus the amount of entanglement must not exceed the coherent information.

To get the entanglement-assisted rate, we can simply add the two inequalities so as to eliminate $E$. The result is that $2Q < H(A|S)_{\sigma} + I(A \rangle C)_{\mathcal{N}(\sigma)}$ and a simple calculation reveals this to be equivalent to $Q < \demi [I(A;C)_{\mathcal{N}(\sigma)} - I(A;S)_{\sigma}]$.

\end{proof}

\section{Optimality for entanglement-assisted coding}\label{sec:sideinfo-converse}
We shall now prove that the previous theorem is optimal for entanglement-assisted coding. In other words, for any achievable rate $Q$ for entanglement-assisted transmission, there exists a state $\sigma^{AA'S}$ as in Theorem \ref{thm:sideinfo-asst} for which $Q = \demi I(A;C)_{\mathcal{N}(\sigma)} - \demi I(A;S)_{\sigma}$. This part is essentially the same as in \cite{gelfand-pinsker}, with a few adaptations to the quantum case. In particular, one must pay close attention to which state the various mutual informations are defined on, since we will be dealing with states where only some fraction of the $n$ instances of the channel has been applied.

\begin{thm}\label{thm:side-info-single-letter}
	The entanglement-assisted quantum capacity of a quantum channel with side information at the transmitter $( \mathcal{N}^{A'S \rightarrow C}, \ket{\phi}^{SS'})$ is
	\begin{equation}
		C = \sup_{\sigma} \left\{ \demi I(A;C)_{\omega} - \demi I(A;S)_{\sigma} \right\}.
	\end{equation}
	The supremum is taken over all mixed states of the form $\sigma^{AA'S}$ where $\sigma^S = \phi^S$ and $\omega^{AC} = \mathcal{N}^{A'S \rightarrow C}(\sigma^{AA'S})$.
\end{thm}
\begin{proof}
The achievability of this rate follows directly from Theorem \ref{thm:side-info-iid}. We therefore now need to prove that one cannot go above this rate. First, let $\mathcal{E}^{M_n \wtA_n {S'}^n \rightarrow {A'}_n}$ and $\mathcal{D}^{C^n B_n \rightarrow M_n B_n}$ be the encoder and the decoder respectively of an arbitrary code of block size $n$ with $\log|R_n| = \log|M_n| = nQ$ such that
\[ \left\| \left( \mathcal{D} \circ \mathcal{N}^{\otimes n} \circ \mathcal{E} \right)\left( \Phi^{R_n M_n} \otimes \Phi^{\wtA_n B_n} \otimes (\phi^{SS'})^{\otimes n} \right) - \Phi^{R_n M_n} \right\|_1 \leqslant \varepsilon, \]
let $\sigma^{R_n B_n {A'}^n S^n} = \mathcal{E}(\Phi^{R_n M_n} \otimes \Phi^{\wtA B_n} \otimes (\phi^{SS'})^{\otimes n})$ and $\omega^{R_n B_n C^n} = \mathcal{N}^{\otimes n}(\sigma)$. Then, by Fannes' inequality (Theorem \ref{thm:fannes}) and the monotonicity of the mutual information (see Section \ref{sec:entropy-properties}) we must have that
\begin{equation}
	I(R_n;C^{n} B_n)_{\omega} \geqslant 2n(Q - d(\varepsilon,n))
\end{equation}
where $d(\varepsilon,n) := \frac{3\varepsilon Q}{2} + \frac{3\varepsilon \log \varepsilon}{n}$.  Notice that
\begin{align}
	I(R_n;B_n C^{n})_{\omega} &= I(B_n;R_n)_{\omega} + I(R_n;C^{n}|B_n)_{\omega}\\
	&= I(R_n;C^{n}|B_n)_{\omega} \label{eqn:weak-converse-2}\\
	&\leqslant I(R_nB_n;C^{n})_{\omega}
\end{align}
where (\ref{eqn:weak-converse-2}) is due to the fact that $R_n$ and $B_n$ are independent. Combining this with $I(R_nB_n;S^{n})_{\sigma} = 0$, we have
\begin{equation}
	I(R_nB_n;C^{n})_{\omega} - I(R_nB_n;S^{n})_{\sigma} \geqslant 2n(Q - d(\varepsilon,n)).
\end{equation}
We will now introduce a few shorthands which will make the notation considerably less cumbersome: we will write $C^i$ instead of $C_1,\ldots,C_i$ and $C_i^j$ instead of $C_i,\ldots,C_j$, and likewise for $S$. Define also
\begin{align}
	X(i) &:= R_nB_n C^{i-1} S_{i+1}^n\\
	Y(i) &:= R_nB_n S_{i+1}^n
	\label{eqn:uv-defs}
\end{align}
Note that these are nothing more than groupings of subsystems. We also define the following sequence of states:
\begin{equation}
	\omega(i) := (\mathcal{N}^{\otimes i} \otimes \ident^{\otimes n-i})(\sigma)
	\label{eqn:sigma-omega-defs}
\end{equation}
In other words, $\omega(i)$ is the result of applying the first $i$ instances of the channel to the state $\sigma$.

We shall now prove the inequality
\begin{multline}
	I(R_nB_n;C^{n})_{\omega} - I(R_nB_n;S^{n})_{\sigma}\\
\leqslant \sum_{i=1}^n \left\{ I(X(i);C_i)_{\omega(i)} - I(X(i);S_i)_{\omega(i-1)} \right\}.
	\label{eqn:toprove}
\end{multline}
Since each term in this sum is of the form $I(A;C)_{\mathcal{N}(\sigma)} - I(A;S)_{\sigma}$ for some $\sigma^{AA'S}$, the highest term is achievable by the direct coding theorem and therefore there exists a state for which $Q \leqslant I(A;C)_{\mathcal{N}(\sigma)} - I(A;S)_{\sigma}$. This allows us to conclude the theorem.

We now proceed in exactly the same way as in \cite{gelfand-pinsker} to establish (\ref{eqn:toprove}): we consider the inequality
\begin{multline}
	I(Y(i);C^i)_{\omega(i)} - I(Y(i);S^i)_{\omega(i-1)}\\
\leqslant \left[ I(Y(i-1);C^{i-1})_{\omega(i-1)} - I(Y(i-1);S^{i-1})_{\omega(i-2)} \right]\\
	+ \left[ I(X(i);C_i)_{\omega(i)} - I(X(i);S_i)_{\omega(i-1)} \right].
	\label{eqn:eq17}
\end{multline}
Summing up all these inequalities from $i=2$ to $i=n$, we obtain (\ref{eqn:toprove}), since $Y(n) = R_nB_n$ and $Y(1) = X(1)$.

Now, to prove (\ref{eqn:eq17}), we use the following identities which follow from the definitions of $X(i)$ and $Y(i)$ and from basic properties of the mutual information:

\begin{align}
	I(Y(i);C^i)_{\omega(i)} &= I(Y(i);C^{i-1})_{\omega(i)} + I(Y(i);C_i|C^{i-1})_{\omega(i)}\\
	I(Y(i);S^i)_{\omega(i-1)} &= I(Y(i);S_i)_{\omega(i-1)} + I(Y(i);S^{i-1}|S_i)_{\omega(i-1)}\\
	I(Y(i-1);S^{i-1})_{\omega(i-1)} &= I(Y(i);S^{i-1}|S_i)_{\omega(i-1)}\\
	I(Y(i-1);C^{i-1})_{\omega(i-1)} &= I(S_i;C^{i-1})_{\omega(i-1)} + I(Y(i);C^{i-1}|S_i)_{\omega(i-1)}\\
	I(X(i);C_i)_{\omega(i)} &= I(C^{i-1};B_i)_{\omega(i)} + I(Y(i); C_i|C^{i-1})_{\omega(i)}\\
	I(X(i);S_i)_{\omega(i-1)} &= I(C^{i-1};S_i)_{\omega(i-1)}d + I(Y(i);S_i|C^{i-1})_{\omega(i-1)}.
\end{align}

Substituting these into (\ref{eqn:eq17}) and using the identity
\begin{equation}
	I(A;B) - I(A;B|C) = I(A;C) - I(A;C|B)
	\label{eqn:mutual-info-identity}
\end{equation}
which holds on any mixed state $\rho^{ABC}$, we get that the difference between the right-hand side and the left-hand side of (\ref{eqn:eq17}) is $I(C^{i-1};C_i)_{\omega(i)}$, which is always nonnegative. This concludes the proof.
\end{proof}

\section{Discussion}\label{sec:si-discussion}
This result further strengthens the parallel between classical information theory problems and their entanglement-assisted quantum counterparts. Indeed, the capacity formula (\ref{eqn:cap}) has the same form as the classical version (\ref{eqn:gp-classique}); the same phenomenon arises in the case of the entanglement-assisted capacities of regular point-to-point channels \cite{BSST02}, multiple-access channels \cite{qmac}, and, for the best coding theorem we know, broadcast channels (see Chapter \ref{chp:bcast}). It would be particularly interesting to have a systematic way in which classical coding theorems could be transformed into entanglement-assisted quantum protocols as it would enable us to import much larger classes of results from classical information theory into the quantum world.

Returning to our result, there is one remaining issue that one would like to solve in order to have a fully satisfactory characterization of the achievable rate region: we currently have no upper bound on the dimension of the $A$ system needed to achieve the capacity in expression (\ref{eqn:cap}). Thus, despite having a single-letter converse, we unfortunately do not have a way to compute the capacity. In the classical case, it is possible to use Carathéodory's theorem \cite{caratheodory} to bound the cardinality of $\mathcal{U}$ in the optimal input distribution. However, in the quantum case, this approach fails due to the fact that the quantum conditional entropy cannot in general be expressed as $H(A|B) = \sum_b p(b) H(A|B=b)$. On the other hand, there is little reason to believe that large dimensions are necessary to achieve the optimal rate, but we do not know how to prove that this is not the case. In fact, one encounters a very similar difficulty when trying to calculate the squashed entanglement \cite{squashed-ent} of a particular state since we have no bound on the size of the subsystem we need to condition on. We therefore leave this issue as an open problem.

One might also wonder about a related problem: whether the capacity can in general be achieved by optimizing only over pure states $\sigma^{AA'S}$. This would imply an upper bound on $|A|$. However, one can show that this cannot be the case: take, for example, a qubit-to-qubit channel which applies one of the four Pauli operations with equal probability, but where $S$ tells the transmitter which one of the four operations is applied. The capacity of such a channel is clearly one qubit per transmission. Suppose that this rate is achievable using a pure state $\sigma^{AA'S}$. Then, we must have $\demi I(A;C)_{\mathcal{N}(\sigma)} = 1$ (since $C$ is two-dimensional) and therefore $\demi I(A;S)_{\sigma} = 0$. However, this last equation together with the fact that $\sigma$ is pure implies that the purification of $S$ must be entirely in $A'$. This is impossible since $S$ is maximally mixed over a four-dimensional system whereas $A'$ is two-dimensional, and hence the optimal $\sigma$ cannot be pure.

\chapter{Quantum broadcast channels}\label{chp:bcast}

\section{Introduction}\label{sec:intro}

Discrete memoryless broadcast channels are channels with one sender and $n$ receivers, modelled using an input set $\mfX$, output sets $\mfY_1, \dots, \mfY_n$, and a probability transition matrix $p(y_1, \ldots, y_n|x)$. When the transmitter selects the input symbol $x_0 \in \mfX$, the output at the receivers is distributed according to $p(y_1,\dots,y_n|x=x_0)$. These can represent, for instance, a radio tower broadcasting a signal to many receivers, each of whom experiences different signal corruption due being closer or further away from the tower, or due to the proximity of buildings. There are many natural tasks that one may want to perform using these channels, such as sending common messages to all the users, sending separate information to each user, sending data to each user privately, or some combination of these tasks. Here we shall focus only on sending separate data to two different receivers that we will call Bob 1 and Bob 2.

One should note in passing that while this definition of broadcast channels is standard in electrical engineering, it may strike computer scientists (and particularly cryptographers) as bizarre. Indeed, computer scientists are used to defining broadcast channels as a \emph{task} to be performed: send the same message to multiple parties, with no notion of noise. Here we think of broadcast channels more as physical objects: a physical channel with one input and multiple outputs, with which we may want to perform a number of different tasks.

Broadcast channels were first introduced by Cover in \cite{cover1}, where he suggested that it may be possible to use them more efficiently than by timesharing between the different users. Since then, several results concerning broadcast channels have been found, such as the capacity of degraded broadcast channels (see, for example, \cite{coverthomas}). Furthermore, these results form the basis of many protocols that are currently used in real multiuser systems, such as cellphone networks.

The best achievable rate region known for general classical broadcast channels is due to Marton \cite{marton}: given a probability distribution $p(x, u_1, u_2) = p(u_1, u_2)p(x|u_1,u_2)$, the following rate region is achievable for the general two-user broadcast channel $p(y_1, y_2|x)$:
\begin{equation}
\label{eqn:marton}
\begin{split}
0 \leqslant R_1 &\leqslant I(U_1; Y_1)\\
0 \leqslant R_2 &\leqslant I(U_2; Y_2)\\
R_1 + R_2 &\leqslant I(U_1; Y_1) + I(U_2; Y_2) - I(U_1; U_2)
\end{split}
\end{equation}
It is conjectured that this characterizes the capacity region of general two-receiver broadcast channels, but despite considerable efforts, no one has been able to prove a converse theorem.

The quantum generalization of broadcast channels was first studied in \cite{allahverdyan-saakian} and \cite{YHD06} as part of a recent effort to develop a network quantum information theory \cite{horo1,horo2,YDH05,LOW06,netcoding1,winter-qmac,klimovitch-qmac,SVW05,bosonic}. In \cite{YHD06}, the authors derived three classes of results, the first one about channels with a classical input and quantum outputs, the second one about sending a common classical message while sending quantum information to one receiver, and the third about sending qubits to one receiver while establishing a GHZ state with the two receivers.

In this chapter, we study quantum broadcast channels using the general techniques developed in this thesis. We look at the case where Alice initially shares a tripartite state $\psi_1^{A_1 B_1 R_1}$ with Bob 1 and a reference, and would like to transfer her share $A_1$ to Bob 1 using the broadcast channel. She would simultaneously like to do the same with $\psi_2^{A_2 B_2 R_2}$ and Bob 2. We first give a one-shot theorem for this task, and then specialize it to the i.i.d.\ case (i.e.\ the channel has the form $\mathcal{N}^{\otimes n}$) in which $\psi_1$ consists of a maximally entangled pair between Alice and a reference, and separate maximally entangled pairs between Alice and Bob 1; the same goes for $\psi_2$ and Bob 2. This corresponds to transmitting qubits with rate-limited entanglement assistance to Bob 1 and Bob 2 simultaneously over the broadcast channel. When we use the maximum possible amount of entanglement assistance, we recover a quantum version of Marton's region. On the other hand, when no entanglement assistance at all is used, the rate region does not appear to have any independent constraint on the sum rate; the information going to Bob 1 and to Bob 2 appear to ``talk past each other''. Interestingly, it turns out that the same phenomenon occurs in the classical scenario of Gaussian multiple-antenna broadcast channels (a particular type of classical broadcast channels) with confidential messages \cite{llps09}. This is perhaps not so surprising, since private classical communication tends to be the closest classical parallel to quantum communication, in which one must inherently keep the information private from the environment.

We then prove a regularized converse for the fully entanglement-assisted case, and give an example of a channel for which the single-letter region is optimal.

\section{Definitions}
Here we define the various concepts needed for this chapter.

\begin{defin}[Quantum broadcast channel]
A quantum broadcast channel is a CPTP map with more than one subsystem as its output, and whose outputs are held by separate receivers.
\end{defin}

In the one-shot case, we will be interested in the following situation: the initial state is $\psi_1^{A_1 B_1 R_1} \otimes \psi_2^{A_2 B_2 R_2}$, where $A_1$ and $A_2$ are held by Alice, $B_1$ and $B_2$ by Bob 1 and Bob 2 respectively, and $R_1$ and $R_2$ are reference systems making the states pure. Alice wants to use the broadcast channel $\mathcal{N}^{A' \rightarrow C_1 C_2}$ to send $A_1$ to Bob 1 and $A_2$ to Bob 2 (of course, Bob 1 gets $C_1$ and Bob 2 gets $C_2$). Hence, we will need to assert the existence of an encoding superoperator $\mathcal{E}^{A_1 A_2 \rightarrow A'}$ and decoders $\mathcal{D}_1^{B_1 C_1 \rightarrow A_1 B_1}$ and $\mathcal{D}_2^{B_2 C_2 \rightarrow A_2 B_2}$ such that
\begin{equation*}
\left\| \left( (\mathcal{D}_1 \otimes \mathcal{D}_2) \circ \mathcal{N} \circ \mathcal{E}\right)\left(\psi_1^{A_1 B_1 R_1} \otimes \psi_2^{A_2 B_2 R_2}\right) - \psi_1^{A_1 B_1 R_1} \otimes \psi_2^{A_2 B_2 R_2} \right\|_1 \leqslant \delta
\end{equation*}
for some $\delta$ that we find suitably small. Note here that the two decoders $\mathcal{D}_1$ and $\mathcal{D}_2$ commute and can be applied in parallel.

In the i.i.d.\ case, we will want to use the broadcast channel $\mathcal{N}^{A' \rightarrow C_1 C_2}$ $n$ times, to transmit separate arbitrary quantum data to Bob 1 and to Bob 2, with separate preshared entanglement with Bob 1 and Bob 2. In other words, $\psi_1^{M_1 \wtA_1 B_1 R_1} =  \Phi^{R_1 M_1} \otimes \Phi^{\wtA_1 B_1}$ where $M_1 \wtA_1$ plays the role $A_1$; likewise for $\psi_2$. Now, for a given protocol for this task, we define the transmission rate to Bob 1 (resp. Bob 2) $Q_1$ (resp. $Q_2$) as $\frac{1}{n} \log|M_1|$ (resp. $\frac{1}{n} \log|M_2|$) and the entanglement consumption rate to Bob 1 (resp. Bob 2) as $E_1 = \frac{1}{n} \log|\wtA_1|$ (resp. $E_2 = \frac{1}{n} \log|\wtA_2|$).

We say that a four-tuple $(Q_1, Q_2, E_1, E_2)$ is achievable if there exists a sequence of encoders $\mathcal{E}_n^{M_{1,n} \wtA_{1,n} \rightarrow {A'^{n}}}$ and decoders $\mathcal{D}_{1,n}^{C_1^n B_{1,n} \rightarrow M_{1,n}}$ and $\mathcal{D}_{2,n}^{C_2^n B_{2,n} \rightarrow M_{2,n}}$, such that
\begin{equation*}
\lim_{n \rightarrow \infty} \left\| \left( (\mathcal{D}_{1,n} \otimes \mathcal{D}_{2,n}) \circ \mathcal{N}^{\otimes n} \circ \mathcal{E}_n \right)\left( \Phi^{R_1 M_1} \otimes \Phi^{\wtA_1 B_1} \otimes \Phi^{R_2 M_2} \otimes \Phi^{\wtA_2 \otimes B_2} \right) - \Phi^{R_1 M_1} \otimes \Phi^{R_2 M_2} \right\|_1 = 0
\end{equation*}
with
\begin{align*}
Q_1 &= \lim_{n \rightarrow \infty} \frac{1}{n} \log|M_{1,n}| = \lim_{n \rightarrow \infty} \frac{1}{n} \log|R_{1,n}|\\
Q_2 &= \lim_{n \rightarrow \infty} \frac{1}{n} \log|M_{2,n}| = \lim_{n \rightarrow \infty} \frac{1}{n} \log|R_{2,n}|\\
E_1 &= \lim_{n \rightarrow \infty} \frac{1}{n} \log|\wtA_{1,n}| = \lim_{n \rightarrow \infty} \frac{1}{n} \log|B_{1,n}|\\
E_2 &= \lim_{n \rightarrow \infty} \frac{1}{n} \log|\wtA_{2,n}| = \lim_{n \rightarrow \infty} \frac{1}{n} \log|B_{2,n}|.
\end{align*}

\section{Direct coding theorem}
We start by proving the one-shot version of the protocol. First, however, we prove a simple technical lemma:

\begin{lem}\label{lem:multidecoupling}
	If we have density operators $\rho^{ABC}, \sigma^A, \omega^{BC}, \tau^{AB}, \eta^C$ such that
\begin{align*}
\left\| \rho^{ABC} - \sigma^A \otimes \omega^{BC} \right\|_1 &\leqslant \varepsilon_1\\
\left\| \rho^{ABC} - \tau^{AB} \otimes \eta^C \right\|_1 &\leqslant \varepsilon_2
\end{align*}
then $\left\| \rho^{ABC} - \sigma^A \otimes \tau^B \otimes \eta^C \right\|_1 \leqslant 2\varepsilon_1 + \varepsilon_2$.
\end{lem}
\vskip 2mm
\begin{proof}
\begin{align*}
\left\| \rho^{ABC} - \sigma^A \otimes \tau^B \otimes \eta^C \right\|_1 &\leqslant \left\| \rho^{ABC} - \sigma^A \otimes \omega^{BC} \right\|_1\\
&\quad + \left\| \sigma^A \otimes \omega^{BC} - \sigma^A \otimes \tau^B \otimes \eta^C \right\|_1\\
&= \varepsilon_1 + \left\| \omega^{BC} - \tau^B \otimes \eta^C \right\|_1\\
&\leqslant \varepsilon_1 + \left\| \omega^{BC} - \rho^{BC} \right\|_1 + \left\| \rho^{BC} - \tau^B \otimes \eta^C \right\|_1\\
&\leqslant 2\varepsilon_1 + \varepsilon_2
\end{align*}
where the first two inequalities are applications of the triangle inequality, and the equality is due to the fact that $\| A \|_1 = \| \sigma \otimes A \|_1$ for any operator $A$ and density matrix $\sigma$. 
\end{proof}

\begin{thm}\label{thm:broadcast-1shot}
For any quantum broadcast channel $\mathcal{N}^{A' \rightarrow C_1 C_2}$, any pair of pure quantum states $\psi_1^{A_1 B_1 R_1}$ and $\psi_2^{A_2 B_2 R_2}$, any pure quantum state $\sigma^{A_1'' A_2'' A'D}$ and any $\varepsilon > 0$, there exists an encoding superoperator $\mathcal{E}^{A_1 A_2 \rightarrow A'}$ and decoding superoperators $\mathcal{D}_1^{C_1 B_1 \rightarrow A_1 B_1}$ and $\mathcal{D}_2^{C_2 B_2 \rightarrow A_2 B_2}$ such that
\begin{multline*}
\left\| ( (\mathcal{D}_1 \otimes \mathcal{D}_2) \circ \mathcal{N} \circ \mathcal{E})(\psi_1^{A_1 B_1 R_1} \otimes \psi_2^{A_2 B_2 R_2}) - \psi_1^{A_1 B_1 R_1} \otimes \psi_2^{A_2 B_2 R_2} \right\|_1\\
 \leqslant 4 \sqrt{2\sqrt{\delta_{\enc}} + \delta_1} + 2 \sqrt{ 2\sqrt{\delta_{\enc}} + \delta_2}
\end{multline*}
where
\begin{equation*}
\delta_{\enc} = 4 \times 2^{\demi H_{\max}^{\varepsilon}(A_1)_{\psi_1} - \demi H_{\min}^{\varepsilon^2/20}(A''_1|A''_2)_{\sigma}} + 5 \times 2^{\demi H_{\max}^{\varepsilon}(A_2)_{\psi_2} - \demi H_{\min}^{\varepsilon}(A''_2)_{\sigma}} + 72\varepsilon
\end{equation*}
and
\begin{align*}
\delta_1 &\leqslant 4 \times 2^{-\demi H_{\min}^{\varepsilon^2/20}(A''_1|EDA''_2 C_2)_{U_{\mathcal{N}} \cdot \sigma} - \demi H_{\min}^{\varepsilon}(A_1|R_1)_{\psi_1}} + 32\varepsilon\\
\delta_2 &\leqslant 5 \times 2^{-\demi H_{\min}^{\varepsilon^2/16}(A''_2|EDA''_1 C_1)_{U_{\mathcal{N}} \cdot \sigma} - \demi H_{\min}^{\varepsilon}(A_2|R_2)_{\psi_2}} + 40\varepsilon.
\end{align*}
\end{thm}
See Figure \ref{fig:broadcast-1shot-real} for an illustration of the purified version of the protocol.

\begin{figure}
    \centering
    \scalebox{0.9}{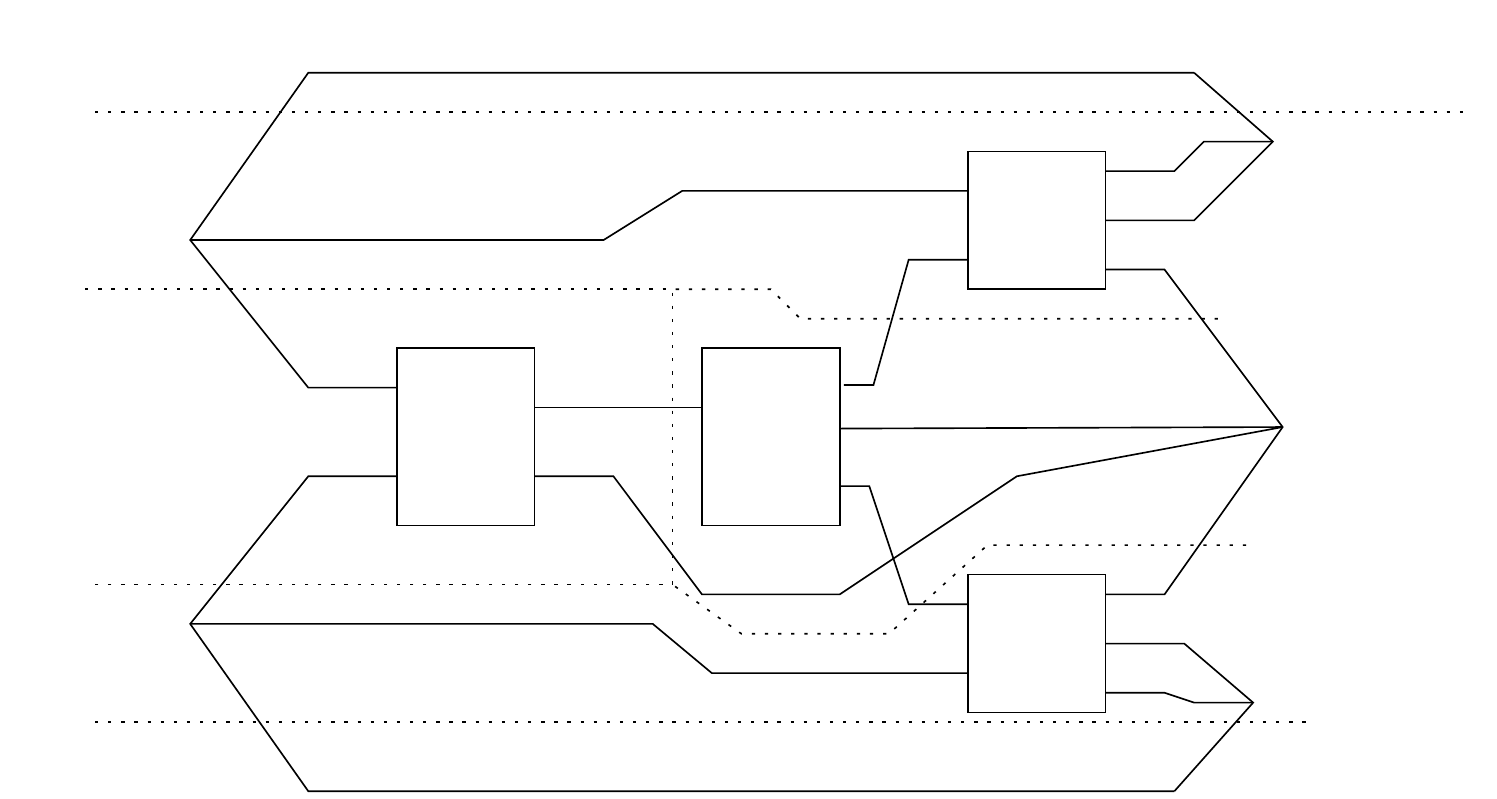}
    \caption{Diagram illustrating Theorem \ref{thm:broadcast-1shot}, with encoder, channel and decoders purified. Each line represents a quantum system, boxes represent isometries, and the horizontal axis represents the passage of time. Lines joined together at either end of the diagram represent pure states. $W$ represents Alice's encoder: she encodes the messages $A_1$ and $A_2$ into the channel input $A'$ and discards a system $D$. The decoders $D_1$ and $D_2$ take the channel outputs $C_1$ and $C_2$ together with Bob 1 and 2's initial systems $B_1$ and $B_2$ and produce $A_1 B_1$ and $A_2B_2$ as output; the result being close to the initial state $\psi$.}
    \label{fig:broadcast-1shot-real}
\end{figure}

\begin{proof}
Let $U_{\mathcal{N}}^{A' \rightarrow C_1 C_2 E}$ and $W^{A_1 A_2 \rightarrow A'D}$ be Stinespring dilations of $\mathcal{N}$ and $\mathcal{E}$ respectively (it will turn out later that the encoder indeed only needs to discard a system of size $D$). To be able to assert the existence of the two decoders, we need to ensure that the two associated decoupling conditions are fulfilled. Those two conditions stipulate that there must exist states $\xi_1$ and $\xi_2$ such that
\begin{gather*}
\left\| \left( U_{\mathcal{N}} W \cdot (\psi_1 \otimes \psi_2) \right)^{R_1 R_2 C_2 B_2 ED} - \psi^{R_1} \otimes \xi_1^{R_2 C_2 B_2 ED} \right\|_1\\
\left\| \left( U_{\mathcal{N}} W \cdot (\psi_1 \otimes \psi_2) \right)^{R_2 R_1 C_1 B_1 ED} - \psi^{R_2} \otimes \xi_2^{R_1 C_1 B_1 ED} \right\|_1
\end{gather*}
are both appropriately small.

To ensure this, let $V_1^{A_1 \rightarrow A''_1}$ and $V_2^{A_2 \rightarrow A''_2}$ be any full-rank partial isometries, $\ket{\tilde{\psi}_1}^{A_1'' B_1 R_1} = V_1 \ket{\psi_1}^{A_1 B_1 R_1}$ and $\ket{\tilde{\psi}_2}^{A_2'' B_2 R_2} = V_2 \ket{\psi_2}^{A_2 B_2 R_2}$, and define the states
\begin{align*}
	\ket{\omega_1(U_2)}^{A''_1A'DR_2B_2} &= \sqrt{|A''_2|} \left( \op_{A''_2 \rightarrow A''_1 A' D}(\ket{\sigma}) U_2^{A''_2} \ket{\tilde{\psi}_2}^{A_2'' B_2 R_2} \right)\\
	\ket{\omega_2(U_1)}^{A''_2A'DR_1B_1} &= \sqrt{|A''_1|} \left( \op_{A''_1 \rightarrow A''_2 A' D}(\ket{\sigma}) U_1^{A''_1} \ket{\tilde{\psi}_1}^{A_1'' B_1 R_1} \right)
\end{align*}
We now use Theorem \ref{thm:smooth-bertha} to get
\begin{multline}\label{eqn:bcast-decoupling1}
	\int \left\| \left(|A''_1| U_{\mathcal{N}} \op_{A''_1 \rightarrow A'DR_2 B_2}(\ket{\omega_1(U_2)}) U_1 \cdot \tilde{\psi}_1^{A''_1 R_1 B_1}\right)^{R_1 EDR_2 B_2 C_2} - \psi_1^{R_1} \otimes \omega_1(U_2)^{C_2 ED R_2 B_2} \right\|_1 dU_1\\
 \leqslant 2^{-\demi H_{\min}^{\varepsilon}(A''_1|EDR_2 B_2 C_2)_{U_{\mathcal{N}} \cdot \omega_1(U_2)} - \demi H_{\min}^{\varepsilon}(A_1|R_1)_{\psi_1}} + 8\varepsilon
\end{multline}
and
\begin{multline}\label{eqn:bcast-decoupling2}
	\int \left\| \left(|A''_2| U_{\mathcal{N}} \op_{A''_2 \rightarrow A'DR_1 B_1}(\ket{\omega_2(U_1)}) U_2 \cdot \tilde{\psi}_2^{A''_2 R_2 B_2}\right)^{R_2 EDR_1 B_1 C_1} - \psi_2^{R_2} \otimes \omega_2(U_1)^{C_1 ED R_1 B_1} \right\|_1 dU_2\\
 \leqslant 2^{-\demi H_{\min}^{\varepsilon}(A''_2|EDR_1 B_1 C_1)_{U_{\mathcal{N}} \cdot \omega_2} - \demi H_{\min}^{\varepsilon}(A_2|R_2)_{\psi_2}} + 8\varepsilon
\end{multline}
Note that the first states on the left-hand side of both inequalities are actually the same state written differently, namely $|A''_1||A''_2| (U_{\mathcal{N}} \op_{A''_1 A''_2 \rightarrow A'D}(\ket{\sigma})(U_1 \otimes U_2) \cdot (\tilde{\psi}_1 \otimes \tilde{\psi}_2))$ (see Lemma \ref{lem:op-switcheroo}). This is close to what we need, but there are still two problems: the encoder in the above is not an isometry, and the smooth-min-entropies should be in terms of $U_{\mathcal{N}} \cdot \sigma$ rather than $U_{\mathcal{N}} \cdot \omega_1$ and $U_{\mathcal{N}} \cdot \omega_2$. To solve the first problem (and temporarily exacerbate the second!), we use Theorem \ref{thm:smooth-bertha} again to get
\begin{multline*}
	\int \left\| |A''_1| \left( \op_{A''_1 \rightarrow R_2 B_2 A'D}(\ket{\omega_1}) U_1 \cdot \tilde{\psi}_1^{A_1 B_1 R_1} \right)^{R_2 B_2 R_1 B_1} - \psi_1^{R_1B_1} \otimes \omega_1(U_2)^{R_2 B_2} \right\|_1 dU_1\\
\leqslant 2^{-\demi H_{\min}^{\varepsilon}(A''_1|R_2B_2)_{\omega_1(U_2)} - \demi H_{\min}^{\varepsilon}(A_1|R_1B_1)_{\psi_1}} + 8\varepsilon
\end{multline*}
and
\begin{multline*}
	\int \left\| |A''_2|\left( \op_{A''_2 \rightarrow A''_1 A'D}(\ket{\sigma}) U_2 \cdot \tilde{\psi}_2^{A''_2 B_2 R_2} \right)^{R_2 B_2} - \psi_2^{R_2 B_2} \right\|_1 dU_2\\
 \leqslant 2^{-\demi H_{\min}^{\varepsilon}(A''_2)_{\sigma} - \demi H_{\min}^{\varepsilon}(A_2|R_2B_2)_{\psi_2}} + 8\varepsilon.
\end{multline*}
Note that, in this last inequality, the first state on the left inside the trace norm is simply $\omega_1$; we can therefore use the triangle inequality to get
\begin{multline*}
	\int \left\| |A''_1| \left( \op_{A''_1 \rightarrow R_2 B_2 A'D}(\ket{\omega_1}) U_1 \cdot \tilde{\psi}_1^{A''_1 B_1 R_1} \right)^{R_2 B_2 R_1 B_1} - \psi_1^{R_1B_1} \otimes \psi_2^{R_2 B_2} \right\|_1 dU_1 dU_2\\
\leqslant \int 2^{-\demi H_{\min}^{\varepsilon}(A''_1|R_2B_2)_{\omega_1(U_2)} - \demi H_{\max}^{\varepsilon}(A_1)_{\psi_1}}dU_2  + 2^{-\demi H_{\min}^{\varepsilon}(A''_2)_{\sigma} - \demi H_{\max}^{\varepsilon}(A_2)_{\psi_2}} + 16\varepsilon 
\end{multline*}
This will allow us to solve our first problem: we will use Uhlmann's theorem on this last inequality to obtain our encoding isometry $W^{A_1 A_2 \rightarrow A'D}$. But before doing this, we will turn our attention to the second problem, namely that of bounding the min-entropies on $\omega_1(U_2)$ and $\omega_2(U_1)$.

There are three such problematic min-entropies: $H_{\min}^{\varepsilon}(A''_1|R_2B_2)_{\omega_1(U_2)}$ in this latest inequality, as well as $H_{\min}^{\varepsilon}(A''_1|EDR_2B_2C_2)_{U_{\mathcal{N}} \cdot \omega_2(U_1)}$ and $H_{\min}^{\varepsilon}(A''_2|EDR_1B_1C_1)_{U_{\mathcal{N}} \cdot \omega_1(U_2)}$ in Equations (\ref{eqn:bcast-decoupling1}) and (\ref{eqn:bcast-decoupling2}) respectively. We will first deal with the first one explicitly; the same technique applies to the other two.

Let $\tilde{\sigma}^{A''_1 A''_2 A'D}$ be a state such that $\| \tilde{\sigma} - \sigma \|_1 \leqslant 2\varepsilon$ and $H_{\min}(A''_1|A''_2)_{\tilde{\sigma}} = H_{\min}^{\varepsilon}(A''_1|A''_2)_{\sigma}$, and let $\mathcal{T}^{A''_2 \rightarrow R_2 B_2}$ be defined as $\mathcal{T}(\xi) = |A''_2| (\op_{A''_2 \rightarrow R_2 B_2}(\ket{\tilde{\psi}_2}) \cdot \xi)$ (and hence, $\omega_1(U_2) = \mathcal{T}(\sigma)$). Furthermore, let $\theta^{A''_2}$ be a positive semi-definite operator such that $\tilde{\sigma}^{A''_1 A''_2} \leqslant \ident^{A''_1} \otimes \theta^{A''_2}$, with $\tr[\theta] = 2^{-H_{\min}^{\varepsilon}(A''_1|A''_2)_{\sigma}}$. Then, since $\mathcal{T}$ is completely positive, we have that, for any $U_2^{A''_2}$,
\[ \mathcal{T}(U_2^{A''_2} \cdot \tilde{\sigma})^{A''_1 R_2 B_2} \leqslant \ident^{A''_1} \otimes \mathcal{T}(U_2^{A''_2} \cdot \theta)^{R_2 B_2}.\]
Hence, if we could be certain that $\mathcal{T}(U_2 \cdot \tilde{\sigma})^{A''_1 R_2 B_2}$ is within $\varepsilon$ in fidelity distance to $\mathcal{T}(U_2 \cdot \sigma)$, we would have shown that $H_{\min}^{\varepsilon}(A''_1|R_2 B_2)_{\omega_1(U_2)} \geqslant H_{\min}^{\varepsilon}(A''_1|A''_2)_{\sigma}$ for any $U_2$. Unfortunately, things are not so easy for us: we will instead have to show that, when averaging over $U_2$, the fidelity distance is not too bad. Note that first that $\int \mathcal{T}(U_2 \cdot \xi) dU_2 = \tr[\xi] \mathcal{T}(\pi^{A''_2})$. Likewise, letting $\tilde{\sigma} - \sigma = \Delta_{+} - \Delta_{-}$ with $\Delta_{\pm}$ positive semi-definite and having disjoint support, and with $\tr[\Delta_{\pm}] \leqslant 2\varepsilon$,
\begin{align*}
\int \| \mathcal{T}(U_2 \cdot \tilde{\sigma}) - \mathcal{T}(U_2 \cdot \sigma) \|_1 dU_2 &= \int \| \mathcal{T}(U_2 \cdot (\Delta_{+} - \Delta_{-}) \|_1 dU_2\\
&\leqslant \int \tr[\mathcal{T}(U_2 \cdot \Delta_{+})] + \tr[\mathcal{T}(U_2 \cdot \Delta_{-})] dU_2\\
&\leqslant 4\varepsilon
\end{align*}
and therefore 
\begin{equation}\label{eqn:df-bound-for-hmin}
\int d_F(\mathcal{T}(U_2 \cdot \tilde{\sigma}), \mathcal{T}(U_2 \cdot \sigma)) dU_2 \leqslant 2\sqrt{\varepsilon}.
\end{equation}
Hence, on average, the fidelity distance is not too bad and we conclude that, on average,
\[ H_{\min}^{2\sqrt{\varepsilon}}(A''_1|R_2 B_2)_{\omega_1(U_2)} \geqslant H_{\min}^{\varepsilon}(A''_1|A''_2)_{\sigma}. \]
Since we want a bound on $H_{\min}^{\varepsilon}(A''_1|R_2B_2)$, we can state this as
\[ H_{\min}^{\varepsilon}(A''_1|R_2B_2) \geqslant H_{\min}^{\varepsilon^2/4}(A''_1|A''_2)_{\sigma}. \]

We can also use the same trick on the other two smooth min-entropies. We now have three inequalities that we want $U_1$ to satisfy and four more inequalities that we need $U_2$ to satisfy. We can use Markov's inequality (see Lemma \ref{lem:markov}) on these to show that there exist a $U_1^{A''_1}$ and a $U_2^{A''_2}$ that satisfy all of them. The inequalities that must be satisfied by $U_1$ are the following:
\begin{multline}\label{eqn:bcast-decoupling-decoder1}
	\left\| |A''_1| \left(U_{\mathcal{N}} \op_{A''_1 \rightarrow A'DR_2 B_2}(\ket{\omega_1(U_2)}) U_1 \cdot \tilde{\psi}_1^{A''_1 R_1 B_1}\right)^{R_1 EDR_2 B_2 C_2} - \psi_1^{R_1} \otimes \omega_1^{ED R_2 B_2 C_2} \right\|_1\\
 \leqslant 4 \times 2^{-\demi H_{\min}^{\varepsilon}(A''_1|EDR_2 B_2 C_2)_{U_{\mathcal{N}} \cdot \omega_1(U_2)} - \demi H_{\min}^{\varepsilon}(A_1|R_1)_{\psi_1}} + 32\varepsilon
\end{multline}
\begin{multline}\label{eqn:bcast-decoupling-encoder1}
	\left\| |A''_1| \left( \op_{A''_1 \rightarrow R_2 B_2 A'D}(\ket{\omega_1(U_2)}) U_1 \cdot \tilde{\psi}_1^{A''_1 B_1 R_1} \right)^{R_2 B_2 R_1 B_1} - \psi_1^{R_1B_1} \otimes \omega_1^{R_2 B_2} \right\|_1\\
\leqslant 4 \times 2^{-\demi H_{\min}^{\varepsilon}(A''_1|R_2B_2)_{\omega_1(U_2)} - \demi H_{\min}^{\varepsilon}(A_1|R_1B_1)_{\psi_1}} + 32\varepsilon
\end{multline}
and
\begin{equation*}
2^{-H_{\min}^{\varepsilon}(A''_2|EDR_1B_1C_1)_{U_{\mathcal{N}} \cdot \omega_2(U_1)}} \leqslant 2^{-H_{\min}^{\varepsilon^2/16}(A''_2|EDA_1C_1)_{U_{\mathcal{N}} \cdot \sigma}}
\end{equation*} 
(The last one is actually a bound on a fidelity distance as in (\ref{eqn:df-bound-for-hmin})).

Likewise, there must exist a $U_2^{A''_2}$ such that
\begin{multline}\label{eqn:bcast-decoupling-decoder2}
	\int \left\| \left(|A''_2| U_{\mathcal{N}} \op_{A''_2 \rightarrow A'DR_1 B_1}(\ket{\omega_2(U_1)}) U_2 \cdot \tilde{\psi}_2^{A''_2 R_2 B_2}\right)^{R_2 EDR_1 B_1 C_1} - \psi_2^{R_2} \otimes \omega_2^{C_1 ED R_1 B_1} \right\|_1 dU_2\\
 \leqslant 5 \times 2^{-\demi H_{\min}^{\varepsilon}(A''_2|EDR_1 B_1 C_1)_{U_{\mathcal{N}} \cdot \omega_2} - \demi H_{\min}^{\varepsilon}(A_2|R_2)_{\psi_2}} + 40\varepsilon
\end{multline}
\begin{multline}\label{eqn:bcast-decoupling-encoder2}
	\int \left\| |A''_2|\left( \op_{A''_2 \rightarrow A''_1 A'D}(\ket{\sigma}) U_2 \cdot \tilde{\psi}_2^{A''_2 B_2 R_2} \right)^{R_2 B_2} - \psi_2^{R_2 B_2} \right\|_1 dU_2\\
 \leqslant 5 \times 2^{-\demi H_{\min}^{\varepsilon}(A''_2)_{\sigma} - \demi H_{\min}^{\varepsilon}(A_2|R_2B_2)_{\psi_2}} + 40\varepsilon.
\end{multline}
\begin{align}
\label{eqn:bcast-minent-bound1}2^{-H_{\min}^{\varepsilon}(A''_1|R_2B_2)_{\omega_1(U_2)}} &\leqslant 2^{-H_{\min}^{\varepsilon^2/20}(A''_1|A''_2)_{\sigma}}\\
2^{-H_{\min}^{\varepsilon}(A''_1|EDR_2B_2C_2)_{U_{\mathcal{N}} \cdot \omega_1(U_2)}} &\leqslant 2^{-H_{\min}^{\varepsilon^2/20}(A''_1|EDA''_2C_2)_{U_{\mathcal{N}} \cdot \sigma}}.
\end{align}

Now, we can combine Equations (\ref{eqn:bcast-decoupling-encoder1}), (\ref{eqn:bcast-decoupling-encoder2}), and (\ref{eqn:bcast-minent-bound1}) to get
\begin{multline}\label{eqn:bcast-decoupling-encoder}
	\left\| |A''_1| \left( \op_{A''_1 \rightarrow R_2 B_2 A'D}(\ket{\omega_1(U_2)}) U_1 \cdot \tilde{\psi}_1^{A''_1 B_1 R_1} \right)^{R_2 B_2 R_1 B_1} - \psi_1^{R_1B_1} \otimes \psi_2^{R_2 B_2} \right\|_1\\
\leqslant 4 \times 2^{-\demi H_{\min}^{\varepsilon^2/20}(A''_1|A''_2)_{\sigma} - \demi H_{\min}^{\varepsilon}(A_1|R_1B_1)_{\psi_1}}\\
+ 5 \times 2^{-\demi H_{\min}^{\varepsilon}(A''_2)_{\sigma} - \demi H_{\min}^{\varepsilon}(A_2|R_2B_2)_{\psi_2}} + 72\varepsilon.
\end{multline}
Using Uhlmann's theorem, we finally get our encoding isometry $W^{A_1 A_2 \rightarrow A'D}$:
\begin{multline}
	\left\| |A''_1| \left( \op_{A''_1 \rightarrow R_2 B_2 A'D}(\ket{\omega_1}) U_1 \cdot \tilde{\psi}_1^{A''_1 B_1 R_1} \right)^{R_2 B_2 R_1 B_1 A'D} - W \cdot (\psi_1^{A_1R_1B_1} \otimes \psi_2^{A_2R_2 B_2}) \right\|_1\\
\leqslant 2 \sqrt{\delta_{\enc}}
\end{multline}
where $\delta_{\enc}$ is defined as the right-hand side of (\ref{eqn:bcast-decoupling-encoder}). Finally, the two decoupling conditions (\ref{eqn:bcast-decoupling-decoder1}) and (\ref{eqn:bcast-decoupling-decoder2}) together with Uhlmann's theorem and Lemma \ref{lem:multidecoupling} yield the existence of the two decoders.
\end{proof}

We can now use this to prove an i.i.d.\ version for both entanglement-assisted and unassisted coding:
\begin{thm}\label{thm:bcast-iid}
Let $\mathcal{N}^{A' \rightarrow C_1 C_2}$ be a quantum broadcast channel with Stinespring extension $U_{\mathcal{N}}^{A' \rightarrow C_1 C_2 E}$, let $\sigma^{A_1 A_2 A'D}$ be any pure state, and define $\rho^{A_1 A_2 C_1 C_2 ED} = U_{\mathcal{N}} \cdot \sigma$. Then, all rates satisfying
\begin{align}
\nonumber 0 \leqslant Q_1 + E_1 &< H(A_1)_{\rho} & Q_1 - E_1 &< I(A_1 \rangle C_1)_{\rho}\\
\label{eqn:bcast-iid-region} 0 \leqslant Q_2 + E_2 &< H(A_2)_{\rho} & Q_2 - E_2 &< I(A_2 \rangle C_2)_{\rho}\\
\nonumber Q_1 + E_1 + Q_2 + E_2 &< H(A_1A_2)_{\rho}
\end{align}
are achievable for quantum transmission with rate-limited entanglement assistance through $\mathcal{N}$. In particular, if we allow $E_1$ and $E_2$ to be maximized (corresponding to fully entanglement-assisted coding), we get a quantum version of Marton's region:
\begin{align*}
Q_1 &< \demi I(A_1;C_1)_{\rho}\\
Q_2 &< \demi I(A_2;C_2)_{\rho}\\
Q_1 + Q_2 &< \demi I(A_1;C_1)_{\rho} + \demi I(A_2;C_2)_{\rho} - \demi I(A_1;A_2)_{\rho}.
\end{align*}
\end{thm}
\begin{proof}
The proof is little more than applying the previous theorem together with the fully quantum AEP (Theorem \ref{thm:fully-quantum-aep}). Consider using the previous theorem on $\mathcal{N}^{\otimes n}$ with input distribution $\sigma^{\otimes n}$ and with transmission and entanglement consumption rates $Q_1$, $Q_2$, $E_1$ and $E_2$. Let $R_1$ and $M_1$ be systems of dimension $2^{nQ_1}$, with $M_1$ representing the quantum information that Alice wants to send to Bob 1, with $R_1$ being the system that purifies it. Furthermore, let $\wtA_1$ and $B_1$ be systems of dimension $2^{nE_1}$ representing Alice's and Bob 1's halves of the preshared entanglement. Replicate all these definitions with subscript 2 for Bob 2. Then, we define $\psi_1 = \Phi^{R_1 M_1} \otimes \Phi^{\wtA_1 B_1}$ and $\psi_2 =  \Phi^{R_2 M_2} \otimes \Phi^{\wtA_2 B_2}$, where $M_1 \wtA_1$ and $M_2 \wtA_2$ play the roles of $A_1$ and $A_2$ from the previous theorem.

To get an error that goes down to zero as $n \rightarrow \infty$, we need to ensure that $\delta_{\enc}$, $\delta_1$ and $\delta_2$ all go down to zero as $n\rightarrow \infty$. By the fully quantum AEP and using the fact that $H_{\max}(A_1)_{\psi_1} = n(Q_1 + E_1)$ and $H_{\max}(A_2)_{\psi_2} = n(Q_2 + E_2)$, $\delta_{\enc}$ goes down to zero if
\begin{align*}
Q_1 + E_1 &< H(A_1|A_2)_{\rho}\\
Q_2 + E_2 &< H(A_2)_{\rho}.
\end{align*}
Likewise, using the fact that $H_{\min}(A_1|R_1)_{\psi_1} = n(E_1 - Q_1)$ and $H_{\min}(A_2|R_2)_{\psi_2} = n(E_2 - Q_2)$, we get that $\delta_1$ goes down to zero if
\[ Q_1 - E_1 < H(A_1|EDA_2C_2)_{\rho} = I(A_1 \rangle C_1)_{\rho} \]
and $\delta_2$ goes to zero if
\[ Q_2 - E_2 < H(A_2|EDA_1C_1)_{\rho} = I(A_2 \rangle C_2)_{\rho}. \]
By switching the roles of Bob 1 and Bob 2, we can also get any rate in the region
\begin{align*}
Q_1 + E_1 &< H(A_1)_{\rho}   & Q_1 - E_1 &< I(A_1 \rangle C_1)_{\rho}\\
Q_2 + E_2 &< H(A_2|A_1)_{\rho}   & Q_2 - E_2 &< I(A_2 \rangle C_2)_{\rho}.
\end{align*}
Taking convex combinations of points in these two regions (which corresponds to timesharing between different protocols) yields the region in the theorem statement.

To get the fully entanglement-assisted region, we simply take linear combinatinons of the various inequalities to get constraints only on $Q_1$ and $Q_2$:
\begin{align*}
Q_1 &< \demi H(A_1)_{\rho} + \demi I(A_1 \rangle C_1)_{\rho} = \demi I(A_1; C_1)_{\rho}\\
Q_2 &< \demi H(A_2)_{\rho} + \demi I(A_2 \rangle C_2)_{\rho} = \demi I(A_2; C_2)_{\rho}\\
Q_1 + Q_2 &< \demi \left[ H(A_1 A_2)_{\rho} + I(A_1 \rangle C_1)_{\rho} + I(A_2 \rangle C_2)_{\rho} \right]\\
&= \demi \left[ H(A_1 A_2)_{\rho} - H(A_1|C_1)_{\rho} - H(A_2|C_2)_{\rho} \right]\\
&= \demi \left[ H(A_1)_{\rho} + H(A_2)_{\rho} - I(A_1;A_2)_{\rho} - H(A_1|C_1)_{\rho} - H(A_2|C_2)_{\rho} \right]\\
&= \demi \left[ I(A_1;C_1)_{\rho} + I(A_2; C_2)_{\rho} - I(A_1;A_2)_{\rho} \right].
\end{align*}
\end{proof}

\section{Regularized converse}
The rate region for the case given in Theorem \ref{thm:bcast-iid} is indeed the capacity for quantum transmission with rate-limited entanglement assistance of quantum broadcast channels provided we regularize over many uses of the channel. It is important to remember, however, that regions defined by very different formulas can nonetheless agree after regularization, so the following theorem should be understood to be only a very weak characterization of the capacity.
\begin{thm}\label{thm:reg-converse}
    The capacity region for rate-limited quantum transmission of a quantum broadcast channel $\mathcal{N}^{A' \rightarrow C_1 C_2}$ is the convex hull of the union of all rate points $(Q_1, Q_2, E_1, E_2)$ satisfying
\begin{align}
	\nonumber Q_1 + E_1 &\leqslant \frac{1}{n}H(A_1)_{\psi}  & Q_1 - E_1 &\leqslant \frac{1}{n} I(A_1 \rangle C_1^{n})_{\psi}\\
	\label{eqn:bcast-converse-region} Q_2 + E_2 &\leqslant \frac{1}{n}H(A_2)_{\psi}  & Q_2 - E_1 &\leqslant \frac{1}{n} I(A_2 \rangle C_2^{n})_{\psi}\\
 \nonumber & Q_1 + Q_2 + E_1 + E_2 \leqslant \frac{1}{n}H(A_1 A_2)_{\psi}
\end{align}
for some state of the form $\ket{\psi}^{A_1 A_2 C_1^{n} C_2^{n} D E^{n}} = {U^{\otimes n}_{\mathcal{N}}}\ket{\phi}^{A_1 A_2 {A'}^{n} D}$, where  $\ket{\phi}$ is a pure state.
\end{thm}
\begin{proof}
    It is immediate from Theorem \ref{thm:bcast-iid} that the region is achievable. We now prove the converse.

    Suppose that $(Q_1, Q_2, E_1, E_2)$ is an achievable four-tuple. That means that there exists a sequence of codes of length $n$ with these rates and with error rate going to 0 as $n \rightarrow \infty$. Consider the code of block size $n$ in this sequence. Let $\psi = \Phi^{R_1 M_1} \otimes \Phi^{\wtA_1 B_1} \otimes \Phi^{R_2 M_2} \otimes \Phi^{\wtA_2 B_2}$ be the input state as in Theorem \ref{thm:bcast-iid}, $\mathcal{E}^{M_1 M_2 \wtA_1\wtA_2 \rightarrow A'^{n}}$ be the encoding superoperator, and let $\rho^{R_1 R_2 C_1^{n} C_2^{n} B_1 B_2 E^n} = U_{\mathcal{N}}^{\otimes n} \cdot \mathcal{E}(\psi)$. We will evaluate entropic quantities with respect to $\rho$.

    Given that Bob 1 must be able to recover a system which purifies $R_1$ from $C_1^{n}$ and $B_1$, we have by Fannes' inequality (Theorem \ref{thm:fannes}) and the monotonicity of the mutual information (see Section \ref{sec:entropy-properties}) that $I(R_1; C_1^{n} B_1) \geqslant 2nQ_1 - n\delta_n$, where $\delta_n \rightarrow 0$ as $n \rightarrow \infty$, and likewise for Bob 2. We also have
\begin{align*}
    2nQ_1 - n\delta_n &\leqslant I(R_1; C_1^{n}B_1)\\
    &= H(R_1) + H(C_1^{n} B_1) - H(R_1 C_1^{n} B_1)\\
    &\leqslant H(R_1) + H(C_1^{n}) + H(B_1) - H(R_1 C_1^{n} B_1)\\
    &= nQ_1 + nE_1 + H(C_1^{n}) - H(R_1 C_1^{n} B_1)\\
    &= nQ_1 + nE_1 + I(R_1 B_1 \rangle C_1^{n})
\end{align*}
where the second line follows from subadditivity, and the third line from the definition of $R_1$ and $B_1$. Hence, if we identify $R_1B_1$ as $A_1$ and likewise for Bob 2, we get
\begin{align}
    Q_1 - E_1 &\leqslant \frac{1}{n} I(A_1 \rangle C_1^{n}) + \delta_n\\
    Q_2 - E_2 &\leqslant \frac{1}{n} I(A_2 \rangle C_2^{n}) + \delta_n
\end{align}
where $\delta_n \rightarrow 0$ as $n \rightarrow \infty$. Since $Q_1 + E_1 = \frac{1}{n} H(A_1)$, $Q_2 + E_2 = \frac{1}{n} H(A_2)$, and $H(A_1A_2) = H(A_1) + H(A_2)$ by construction, this rate point is clearly inside the region in Equation (\ref{eqn:bcast-converse-region}), and it follows that this is indeed the capacity of the channel.
\end{proof}

While one might conjecture that Theorem \ref{thm:reg-converse} characterizes the capacity region of a broadcast channel for quantum transmission with rate-limited entanglement assistance even with the restriction $n=1$, this is false even in the special case of unassisted quantum transmission through a channel with a single receiver \cite{superadd-qcap}. It may however be the true capacity for the fully entanglement-assisted case, but there is no reason to believe that this would be any easier to prove than to prove that Marton's region is optimal in the classical case.

\section{Single-letter example}
In the classical case, the simplest example of a broadcast channel for which Marton's region is optimal is a deterministic channel, i.e.\ a channel where the outputs are completely determined by the inputs. Similarly, we can show that our rate region is optimal for entanglement-assisted quantum transmission through classical deterministic channels. This is perhaps unsurprising since entanglement would be highly unlikely to help classical transmission through a classical channel, but it nonetheless provides an example for which our theorem is optimal.

We say that $\mathcal{N}^{A' \rightarrow C_1 C_2}$ is a classical deterministic broadcast channel if there exist two deterministic functions $f_1:\{1,\ldots,|A'|\} \rightarrow \{1,\ldots,|C_1|\}$ and $f_2:\{1,\ldots,|A'|\} \rightarrow \{1,\ldots,|C_2|\}$ such that $U_{\mathcal{N}}\ket{i} = \ket{f_1(i)}^{C_1} \otimes \ket{f_2(i)}^{C_2} \otimes \ket{i}^E$ for some fixed orthonormal bases on $A'$, $C_1$, $C_2$ and $E$. We claim that any rate point that can be achieved for such a channel is a convex combination of rates that can be achieved via our coding method with input states of the form $\varphi^{A_1 A_2 A'} = \sum_{i = 1}^{|A'|} p_i \ketbra{f_1(i)}^{A_1} \otimes \ketbra{f_2(i)}^{A_2} \otimes \ketbra{i}^{A'}$ for some probability distribution $\{p_i\}$. To prove this, we first need the following observation:

\begin{lem}\label{lem:single-letter}
	Let $f : \{1,\ldots,|D|\} \rightarrow \{1,\ldots,|B|\}$ be a function, and $\ket{\xi}^{ABCD}$ be $\sum_i \alpha_i \ket{\mu_i}^A \otimes \ket{f(i)}^B \otimes \ket{\nu_i}^C \otimes \ket{i}^D$, where $\ket{\mu_i}$ and $\ket{\nu_i}$ are any pure states, and $\ket{i}$ and $\ket{f(i)}$ represent $i$ and $f(i)$ encoded in a standard bases on $D$ and $B$ respectively. Then, $I(A;B)_{\xi} \leqslant H(B)_{\xi}$.
\end{lem}
\begin{proof}
	The lemma follows from the observation that because of the structure of $\xi^{AB}$ and strong subadditivity (see Section \ref{sec:entropy-properties}), $H(B|A) \geqslant H(B|D)$. The latter is a classical conditional entropy and is therefore never negative, which means that $I(A;B)_{\xi} = H(B)_{\xi} - H(B|A)_{\xi} \leqslant H(B)_{\xi}$.
\end{proof}

Armed with this, we can now show the following:
\begin{thm}
	Let $\mathcal{N}^{A' \rightarrow C_1 C_2}$ be a classical deterministic channel. Then, the capacity region for entanglement-assisted quantum transmission on this channel is the same as the achievable rate region given by Theorem \ref{thm:bcast-iid}.
\end{thm}
\begin{proof}
	According to the regularized converse theorem (Theorem \ref{thm:reg-converse}), for any achievable rate point $(Q_1,Q_2)$, there exists a state $\ket{\psi}^{A_1 A_2 C_1^n C_2^n E^n D} = U_{\mathcal{N}}^{\otimes n}\ket{\varphi}^{A_1 A_2 A'^n D}$ such that $Q_1 = \frac{1}{2n} I(A_1;C_1^n)_{\psi} + \delta_n$, $Q_2 = \frac{1}{2n} I(A_2;C_2^n)_{\psi} + \delta_n$, where $\delta_n \geqslant 0$, and $I(A_1;A_2)_{\psi} = 0$. Let $C_{1,i}$ and $C_{2,i}$ be the $i$th copies of $C_1$ and $C_2$ in $C_1^n$ and $C_2^n$, and, for each $i$, let $\psi_i^{A_1 A_2 C_1 C_2} = \sum_{jk} \ket{jkjk}\bra{jk} \psi^{C_{1,i} C_{2,i}} \ket{jk} \bra{jkjk}$, where the $\bra{jkjk}\ket{jk}$ are defined in the classical basis on $C_{1,i}$ and $C_{2,i}$ and in some fixed basis on $A_1, A_2, C_1$ and $C_2$. Then, we can bound the individual rates as follows:

\begin{align}
	Q_1 &\leqslant \frac{1}{2n} I(A_1;C_1^n)_{\psi} + \delta_n\\
	&\leqslant \frac{1}{2n} H(C_1^n)_{\psi} + \delta_n\\
	&\leqslant \frac{1}{2n} \sum_i H(C_{1,i})_{\psi} + \delta_n\\
	&= \frac{1}{2n} \sum_i H(C_{1})_{\psi_i} + \delta_n\\
	&= \frac{1}{n} \sum_i \frac{1}{2}I(A_1;C_1)_{\psi_i} + \delta_n \label{eqn:single-letter-terms}
\end{align}
and likewise for $Q_2$. The second inequality is due to Lemma \ref{lem:single-letter}, with the roles of the $B$ and $D$ subsystems in the lemma played by $C_1^n$ and $E^n$ respectively, and the third inequality makes use the subadditivity of the von Neumann entropy.

We can now do the same thing for the sum rate:
\begin{align}
	\nonumber Q_1 + Q_2 &= \frac{1}{2n} \left\{ I(A_1;C_1^n)_{\psi} + I(A_2;C_2^n)_{\psi} \right\} + 2\delta_n\\
	\nonumber &= \frac{1}{2n} \left\{ H(A_1)_{\psi} + H(A_2)_{\psi} - H(A_1|C_1^n)_{\psi} - H(A_1;C_2^n)_{\psi} \right\} + 2\delta_n\\
	\nonumber &\leqslant \frac{1}{2n} \left\{ H(A_1A_2)_{\psi} - H(A_1A_2|C_1^nC_2^n)_{\psi} \right\} + 2\delta_n\\
	\nonumber &= \frac{1}{2n} I(A_1A_2;C_1^nC_2^n)_{\psi} + 2\delta_n\\
	\label{eqn:single-letter-terms-sumrate} &\leqslant \frac{1}{2n} H(C_1^nC_2^n)_{\psi} + 2\delta_n\\
	\nonumber &\leqslant \frac{1}{2n} \sum_i H(C_{1,i}C_{2,i})_{\psi} + 2\delta_n\\
	\nonumber &= \frac{1}{2n} \sum_i H(C_{1}C_{2})_{\psi_i} + 2\delta_n\\
	\nonumber &= \frac{1}{n} \sum_i \frac{1}{2} \left\{ H(C_1)_{\psi_i} + H(C_2)_{\psi_i} - I(C_1;C_2)_{\psi_i} \right\} + 2\delta_n\\
	\nonumber &= \frac{1}{n} \sum_i \frac{1}{2} \left\{ I(A_1;C_1)_{\psi_i} + I(A_2;C_2)_{\psi_i} - I(A_1;A_2)_{\psi_i} \right\} + 2\delta_n 
\end{align}
where, in the first inequality, we have made use of the fact that $A_1$ and $A_2$ are independent and of the standard inequality $H(AB|CD) \leqslant H(A|C) + H(B|D)$, and the last equality follows from the special form of the $\psi_i$'s.

Since every $i$ in equations (\ref{eqn:single-letter-terms}) and (\ref{eqn:single-letter-terms-sumrate}) corresponds to a rate which is achievable via Theorem \ref{thm:bcast-iid}, this concludes the proof.
\end{proof}

\section{Discussion}\label{sec:bcast-discussion}
We have exhibited and analyzed a new protocol for quantum communication with rate-limited entanglement assistance through quantum broadcast channels. Our protocol achieves the following rate region for every mixed state $\sigma^{A_1 A_2 A'}$:
\begin{equation}
\begin{split}
0 \leqslant Q_1 &\leqslant \frac{1}{2} I(A_1; C_1)_\rho\\
0 \leqslant Q_2 &\leqslant \frac{1}{2} I(A_2; C_2)_\rho\\
Q_1 + Q_2 &\leqslant \frac{1}{2} \left[ I(A_1;C_1)_\rho + I(A_2;C_2)_\rho - I(A_1;A_2)_\rho \right]
\end{split}
\end{equation}
where $\rho^{A_1 A_2 C_1 C_2 E} = U_{\mathcal{N}}^{A' \rightarrow C_1 C_2 E} \cdot \sigma^{A_1 A_2 A'}$.

The corresponding rate region (Equation (\ref{eqn:bcast-iid-region})) is very similar to Marton's region for classical broadcast channels (Equation (\ref{eqn:marton})) \cite{marton}; except for the factors of $1/2$, the two expressions are formally identical. In fact, for classical channels, the rates for entanglement-assisted quantum communication found here can be achieved directly using teleportation between the senders and the receiver, with the classical communication required by teleportation transmitted using Marton's protocol. From this point of view, our results can be viewed as a direct generalization of Marton's region to quantum channels.

Therefore, once again, it is the entanglement-assisted version of the quantum capacity that bears the strongest resemblance to its classical counterpart. The same is true for both the regular point-to-point quantum channel \cite{BSST02} and the quantum multiple-access channel \cite{qmac, state-merging} and, of course, the quantum channels with side information at the transmitter that were discussed in the last chapter. In both those cases, the known achievable rate regions for entanglement-assisted quantum communication are identical to their classical counterparts.  This collection of similarities suggests a fundamental question. To what extent does the addition of free entanglement make quantum information theory similar to classical information theory? 

Of course, the lack of a single-letter converse for Marton's region and, by extension, for our region, leaves open the possibility that the analogy might break down for a new, better broadcast region that remains to be discovered. A first step towards eliminating that uncertainty could be to find a better characterization of the quantum regions we have presented here. The presence of the ``discarded'' system $D$ in Theorem \ref{thm:bcast-iid} is equivalent to optimizing over all mixed states $\phi^{A_1 A_2 A'}$ rather than only over pure states. This is not required for most theorems in quantum information theory, but we have not found a way to prove the regularized converse without allowing for the possibility of mixed states. We leave it as an open problem to determine whether it is possible to demonstrate a converse theorem that does not require allowing mixed states.

Finally, for the unassisted case, it is very interesting to note the absence of an independent constraint on the sum-rate. However, we already know that this region is suboptimal even for channels with a single receiver. It would therefore be desirable to know whether this holds for the true capacity region and whether there is an underlying principle that explains this phenomenon.

\chapter{Locking classical information in quantum states}\label{chp:locking}

One particularly shocking feature of quantum information is the ``information locking'' effect that one sometimes observes. At the general level, it consists of a system in which one encodes classical data into a quantum system with two parts, one part being a very large ``cyphertext'', and the other being a very small key. The strange phenomenon is that it is possible to set up the system in such a way that, given the large portion, one can get almost no information about the classical data by measuring the cyphertext, whereas the key allows one to ``unlock'' this information. This may seem at first somewhat unsurprising, since this is what classical cryptographic systems aim to do, but it must be stressed that this is at the information theory level: the distribution on the classical data given the large portion is almost the same as the prior distribution even if the key is much smaller than the message. Classical encryption cannot achieve this at all: the distribution on the message given the cyphertext is vastly different from the prior distribution unless the key is as large as the message.

Information locking schemes have already been shown to exist by \cite{dhlst03} and \cite{HLSW03}. In \cite{dhlst03}, the authors construct a scheme by encoding the classical information in one of two mutually unbiased bases, and the one-bit classical key simply tells which basis it's encoded in. Without the key, the Shannon entropy about the message is approximately half of the entropy of the message; the key therefore increases the information the receiver has by the same amount.

In \cite{HLSW03}, the authors look at a protocol where one encodes classical information in the computational basis, and then applies one of a few (logarithmic in the number of possible messages) fixed unitaries. The classical key tells which unitary was applied. If the unitaries are chosen according to the Haar measure, then locking occurs with high probability.

In both of these papers, locking was defined in terms of the \emph{accessible information} between the cyphertext and the message, which defined as follows:
\begin{defin}[Accessible information]
Let $\rho^{AB} \in \DD(\sfA \otimes \sfB)$ be a quantum state. Then, the accessible information $I_{\acc}(A;B)$ is defined as
\[ I_{\acc}(A;B)_{\rho} := \sup_{\mathcal{A}, \mathcal{B}} I(X;Y)_{(\mathcal{A} \otimes \mathcal{B})(\rho)}, \]
where $\mathcal{A}^{A \rightarrow X}$ and $\mathcal{B}^{B \rightarrow Y}$ are measurement superoperators, and the supremum is taken over all possible superoperators. In other words, the accessible information is the largest possible mutual information between the results of measurements made on $A$ and $B$.
\end{defin}

Locking was said to occur when the difference in accessible information with and without the key was larger than the size of the key. Here we will instead use the trace distance between the joint distribution of the measurement results and the message and the product of their marginals. This will imply a bound on the mutual information via the Alicki-Fannes inequality (Lemma \ref{lem:alicki-fannes}). 

We now give the formal definition of locking that we will use:
\begin{defin}
	Let $C$ and $K$ be two quantum systems. We call a set of quantum states $\{ \rho_m^{CK} : m \in \{ 1,\dots,N \} \}$ an $\varepsilon$-locking scheme if $\| \rho_i^{CK} - \rho_j^{CK} \|_1 = 2$ whenever $i \neq j$, and for any complete measurement superoperator $\mathcal{M}^{C \rightarrow X}$, we have that
\[
\left\| \mathcal{M}(\omega^{MC}) - \mathcal{M}(\pi^C) \otimes \omega^M \right\|_1 \leqslant \varepsilon
\]
where $\omega^{MC} = \frac{1}{N} \sum_{i=1}^N \ketbra{i}^M \otimes \rho_i^C$ and $\pi^C$ denotes the completely mixed state on $C$.
\end{defin}
In other words, a set of states is a locking scheme if the states are perfectly distinguishable when one has both the cyphertext $C$ and the key $K$, whereas a measurement on $C$ alone yields practically no information about which state was present. Note that the restriction to complete measurements is a natural one, since, if the goal is to maximize the information about the message, keeping a quantum residue is of no use: it can never hurt to measure it until nothing is left.

Note that it is impossible to achieve this classically without making the key almost as long as the message. One can see this by considering that fact that, if one only needs to know an extra $\log K$ bits to reconstruct the message, our probability distribution of the message given the cyphertext must always be supported on at most $K$ distinct messages; such a distribution must necessarily be far away from the uniform distribution over all messages unless $K$ is nearly equal to the number of messages.

The scheme we will construct here is a special case of this model. We consider a scheme where we encode classical information in the computational basis of a quantum system, apply a fixed unitary, and split the system into two components, a large one ($C$) that becomes the cyphertext, and a small one ($K$) that becomes the key.

Note also that an $\varepsilon$-locking scheme also automatically implies locking of the accessible information:
\begin{lem}
Let $\{ \rho^{CK}_{m} : m \in \{1,\dots,N\} \}$ be an $\varepsilon$-locking scheme, and let $\omega^{MC} = \frac{1}{N} \sum^N_{i=1} \ketbra{i}^M \otimes \rho^C_i$. Then,
\[I_{\acc}(M;C)_{\omega} \leqslant \varepsilon \log N + 2\eta (1-\varepsilon) + 2\eta (\varepsilon),\]
where $\eta(x) := -x \log x$ and $\eta(0) = 0$.
\end{lem}
\begin{proof}
Direct application of the Alicki-Fannes inequality (Lemma \ref{lem:alicki-fannes}).
\end{proof}

\section{The locking scheme}
Our information locking scheme is straightforward. To encode $N$ equiprobable classical messages, we embed them via a random partial isometry into a system $CK$ of total dimension at least $N$; $C$ constitutes the cyphertext and $K$ constitutes the key. The key is therefore itself a quantum state; if one prefers a scheme with a classical key as was done in \cite{dhlst03, HLSW03}, one can simply perfectly encrypt the quantum key with a $2\log K$-bit classical key and make the encrypted quantum key part of the cyphertext; the classical key is then the locking key.

To prove that this works, let $\{ \ket{\psi_m} : 1 \leqslant m\leqslant N \}$ be any set of orthonormal pure states in $CK$. We would like to prove that there exists a $U^{CK}$ such that $\{ U^{CK} \ket{\psi_m} \}$ is a good locking scheme. To do this, we will consider the state $\rho^{MCK} = \frac{1}{N}\sum_{x=1}^N \ketbra{m}^M \otimes \ketbra{\psi_m}^{CK}$ and the expression
\[ \left\| \mathcal{M}(\tr_K[U \cdot \rho^{MCK}]) - \mathcal{M}(\pi^C) \otimes \rho^M \right\|_1 \]
for a $U^{CK}$ chosen according to the Haar measure and an arbitrary $\mathcal{M}$. We will show that the average is sufficiently small and that the distribution is sufficiently concentrated around the mean value to ensure that there exists a $U$ that makes this expression small for every $\mathcal{M}$.

\begin{thm}
	Let $\rho^{MCK}$ be a state of the form $\rho^{MCK} = \sum_{m=1}^N \ketbra{m}^M \otimes \rho_m^{CK}$. Then, there exists a $U^{CK}$ such that for every measurement superoperator $\mathcal{M}^{C \rightarrow X}$,
	\begin{equation*}
		\left\| \mathcal{M}(\tr_K[U \cdot \rho^{MCK}]) - \mathcal{M}(\pi^C) \otimes \rho^M \right\|_1 \leqslant 7\varepsilon
	\end{equation*}
as long as $N \geqslant \frac{8\sqrt 2}{\varepsilon}$, $\varepsilon \leqslant e^{-2}$, and $|K| \geqslant \frac{32}{\varepsilon} \sqrt{\log\left( \frac{4N^2}{\varepsilon}\right) \ln(1/\varepsilon) }$.
\end{thm}

To prove this, we will first consider $\mathcal{M}$'s of a very specific form that will allow us to take a union bound:

\begin{defin}[Quasi-measurement]
	We call a superoperator $\mathcal{M}^{C \rightarrow X}$ an $(n,k)$-quasi-measurement if it is of the form $\mathcal{M}(\sigma) = \frac{|C|}{n} \sum_{x=1}^n \ket{x}\bra{\psi_x} \sigma \ket{\psi_x}\bra{x}$ where the $\ket{x}$ are orthonormal, and $\frac{|C|}{n}\sum_{x=1}^n  \ketbra{\psi_x} \leqslant k \ident^C$. 
\end{defin}

The starting point will be the following concentration of measure result:
	
\begin{lem}
	Let $\mathcal{M}$ be an $(n,k)$-quasi-measurement. Then,
\begin{multline*}
	\Pr_U \left\{ \left\| \mathcal{M}(\tr_K[U \cdot \rho^{MCK}]) - \mathcal{M}(\pi^C) \otimes \rho^M \right\|_1 \geqslant 2^{\demi \log k - \demi \log|K|} + r \right\}\\
\leqslant 2e^{-N^2 r^2/16k^2}
\end{multline*}
\end{lem}
\begin{proof}
	The lemma is a direct application of Theorem \ref{thm:bertha-concentration}. We first show that $\max \{ \| \mathcal{M}(\tr_K[X]) \|_1 : X \in \Herm(\sfA), \| X \|_1 \leqslant 1 \} \leqslant k$. Let $X \in \Herm(\sfA)$; then,
\begin{align*}
\left\| \mathcal{M}(\tr_K(X)) \right\|_1 &= \frac{|C|}{n} \left\| \sum_x \ket{x} \bra{\psi_x} \tr_K(X) \ket{\psi_x} \bra{x} \right\|_1\\
&= \frac{|C|}{n} \sum_x \left| \bra{\psi_x} \tr_K(X) \ket{\psi_x} \right|\\
&\leqslant \frac{|C|}{n} \sum_x \bra{\psi_x} |\tr_K(X)| \ket{\psi_x} \\
&= \frac{|C|}{n} \tr \left[ \sum_x \ketbra{\psi_x} |\tr_K(X)|\right] \\
&\leqslant k \| \tr_K(X) \|_1\\
&\leqslant k \| X \|_1
\end{align*}
where the first inequality follows from the matrix inequality $-|Y| \leqslant Y \leqslant |Y|$, which holds for any Hermitian $Y$. Next, let $\omega^{C'K'X} = \mathcal{M}(\tr_K[\Phi^{C'K'CK}])$; we will show that $H_2(C'K'|X)_{\omega} \geqslant -\log k + \log|K|$. Since $\mathcal{M}$ is a quasi-measurement, $\omega$ has the form $\omega^{C'K'X} = \sum_{x=1}^n \alpha_x \ketbra{x}^X \otimes \pi^{K'} \otimes (\psi_x^{T})^{C'}$, where $\alpha_x = \frac{|C|}{n}\bra{\psi_x} \pi \ket{\psi_x}$. Then, $(\omega^X)^{-1/4} = \sum_x \alpha_x^{-1/4} \ketbra{x}$, and we have that
\begin{align*}
2^{-H_2(C'K'|X)_{\omega}} &\leqslant \tr\left[ \left( (\omega^X)^{-1/4} \omega^{C'K'X} (\omega^X)^{-1/4} \right)^2 \right]\\
&= \tr\left[ \sum_x \alpha_x \ketbra{x} \otimes (\pi^{K'})^2 \otimes (\psi_x^{C'})^T \right]\\
&= \frac{1}{|K|}\sum_x \frac{|C|}{n}\bra{\psi_x} \pi \ket{\psi_x}\\
&= \frac{1}{|K|n} \tr \left[ \sum_x \ketbra{\psi_x} \right]\\
&= \frac{k}{|K|}
\end{align*}
We combine this with the fact that $H_2(CK|M)_{\rho} = 0$ to get the lemma.
\end{proof}

At this point, we would like to take a union bound over all possible quasi-measurements to be able to say that there is a nonzero probability that the trace distance above is small for every quasi-measurement $\mathcal{M}$. We do this by introducing an $\varepsilon$-net (see Definition \ref{def:nets}) $\mfN$ over $C$, with $|\mfN| \leqslant \left( \frac{5}{\varepsilon} \right)^{2|C|}$. For any $(n,k)$-quasi-measurement $\mathcal{M}^{C \rightarrow X}$ of the form $\mathcal{M}(\sigma) = \sum_x \alpha_x \ket{x}\bra{\psi_x}\sigma\ket{\psi_x}\bra{x}$, define $\mathcal{M}_{\mfN}$ as $\mathcal{M}_{\mfN}(\sigma) = \sum_x \alpha_x \ket{x}\bra{\psi'_x}\sigma\ket{\psi'_x}\bra{x}$, where $\ket{\psi'_x}$ is the state in $\mfN$ closest to $\ket{\psi_x}$. 

Given a sequence $( \ket{\varphi_x} : 1 \leqslant x \leqslant n, \ket{\varphi_x} \in \mfN )$, we say that is it $\varepsilon$-close to an $(n,k)$-quasi-measurement if there exists an $(n,k)$-quasi-measurement $\mathcal{M}(\sigma) = \frac{|A|}{n} \sum_{x=1}^n \ket{x}\bra{\psi_x}\sigma\ket{\psi_x}\bra{x}$ such that $\| \psi_x - \varphi_x \|_1 \leqslant \varepsilon$ for all $x$. Furthermore, given a sequence $q = \{ \ket{\psi_x} : 1 \leqslant x \leqslant n \}$, we define $\mathcal{M}_q^{C \rightarrow X}$ as $\mathcal{M}_q(\sigma) = \frac{|C|}{n}\sum_{x=1}^n \ket{x}\bra{\psi_x}\sigma \ket{\psi_x}\bra{x}$.

We can now take the desired union bound:

\begin{lem}\label{lem:quasi-measurement-concentration}
	Let $\mfQ \subseteq \mfN^n$ be the set of all sequences of $n$ elements of $\mfN$ that are $\varepsilon$-close to an $(n,k)$-quasi-measurement. Then,
	\begin{multline*}
		\Pr_U \left\{ \exists q \in \mfQ : \left\| \mathcal{M}_q(\tr_K[U \cdot \rho^{MCK}]) - \mathcal{M}_q(\pi^C) \otimes \rho^M \right\|_1 \geqslant 2^{\demi \log k - \demi \log|K|} + 2\varepsilon + r \right\}\\
		\leqslant 2 e^{2n\ln(5/\varepsilon)|C| - N^2r^2/16k^2}
	\end{multline*}
	and therefore, as long as $2n\ln(5/\varepsilon)|C| - N^2r^2/16k^2 < -\ln 2$, there exists a $U$ such that all $q \in \mfQ$ satisfy the above inequality.
\end{lem}
\begin{proof}
	Let $\mathcal{M}_q(\sigma) = \frac{|C|}{n} \sum_x \ket{x}\bra{\psi_x} \sigma \ket{\psi_x}\bra{x}$ and let $\mathcal{M}'(\sigma) = \frac{|C|}{n} \sum_x \ket{x}\bra{\psi'_x} \sigma \ket{\psi'_x}\bra{x}$ be an $(n,k)$-quasi-measurement that is $\varepsilon$-close to $q$. Furthermore, let $\xi_x = \psi_x - \psi'_x$; clearly, for each $x$, $\| \xi_x \|_1 \leqslant \varepsilon$. Then, given any cq-state $\zeta^{MC}$ with $\zeta^C = \pi^C$, we have that
	\begin{multline*}
		\left\| \mathcal{M}(\zeta^{MC} - \zeta^M \otimes \zeta^C) \right\|_1\\
	\begin{split}
		&\leqslant \frac{|C|}{n}\sum_{x=1}^n \left\| \tr_C[\psi_x (\zeta^{MC} - \zeta^M \otimes \zeta^C )] \right\|_1\\
		&= \frac{|C|}{n}\sum_{x=1}^n \left\| \tr_C[(\psi'_x + \xi_x) (\zeta^{MC} - \zeta^M \otimes \zeta^C )] \right\|_1\\
		&\leqslant \frac{|C|}{n}\sum_{x=1}^n \left( \left\| \tr_C[\psi'_x (\zeta^{MC} - \zeta^M \otimes \zeta^C )] \right\|_1 + \left\| \tr_C[\xi_x (\zeta^{MC} - \zeta^M \otimes \zeta^C )] \right\|_1 \right)\\
		&\leqslant \frac{|C|}{n}\sum_{x=1}^n \left( \left\| \tr_C[\psi'_x (\zeta^{MC} - \zeta^M \otimes \zeta^C )] \right\|_1 + \left\| \tr_C[\xi_x \zeta^{MC}] \right\|_1 + \left\| \tr_C[\xi_x (\zeta^M \otimes \zeta^C) ] \right\|_1 \right)\\
		&\leqslant \frac{|C|}{n}\sum_{x=1}^n \left( \left\| \tr_C[\psi'_x (\zeta^{MC} - \zeta^M \otimes \zeta^C )] \right\|_1 + \frac{2\varepsilon}{|C|} \right)\\
		&= \left\| \mathcal{M}'(\zeta^{MC} - \zeta^M \otimes \zeta^C) \right\|_1 + 2\varepsilon
	\end{split}
	\end{multline*}
	Now, we have that $|\mfQ| \leqslant |\mfN^n| \leqslant \left( \frac{5}{\varepsilon} \right)^{2n|C|}$. Hence, by the union bound and Lemma \ref{lem:quasi-measurement-concentration}, we get the lemma.
\end{proof}

We need to use this to get a bound on general measurement superoperators. The idea will be to imagine that, given any measurement operator, we perform $n$ independent measurements on $n$ i.i.d.\ copies of $\rho$. The operator Chernoff bound (Lemma \ref{lem:operator-chernoff}) will then ensure that the resulting sequence of measurement results is an $(n,k)$-quasi-measurement with high probability.

\begin{lem}\label{lem:locking-operator-chernoff}
	Let $\mathcal{M}^{C \rightarrow X}$ be any complete measurement superoperator, with $\mathcal{M}(\pi) = \sum_x \alpha_x \ket{x}\bra{\psi_x}\pi\ket{\psi_x}\bra{x}$, and consider the operator-valued random variable $Y$ which takes the value $\ketbra{\psi_x}$ with probability $\alpha_x \bra{\psi_x} \pi \ket{\psi_x} = \alpha_x/|C|$. Then, $n$ i.i.d.\ copies of $Y$ will fail to be an $(n,k)$-quasi-measurement with probability at most $2|C| e^{-n(k-1)^2/|C|2\ln 2}$.
\end{lem}
\begin{proof}
	$Y$ fulfills all the conditions for the operator Chernoff bound (Lemma \ref{lem:operator-chernoff}) to apply, with $\mbE Y = \pi^C$. This yields
\[ \Pr \left\{ \frac{1}{n} \sum_{j=1}^n Y_j \nleqslant k\pi \right\} \leqslant 2|C| e^{-n(k-1)^2/|C|2\ln 2} \]
and the probability on the left is an upper bound on the probability that the sequence $Y_1,\dots, Y_n$ is not an $(n,k)$-quasi-measurement.
\end{proof}

Putting all the pieces together, we finally get the main theorem of this section:

\begin{thm}
	There exists a $U^{CK}$ such that for all measurement operators $\mathcal{M}^{C \rightarrow X}$,
	\begin{equation*}
		\left\| \mathcal{M}(\tr_K[U \cdot \rho^{MCK}]) - \mathcal{M}(\pi^C) \otimes \rho^M \right\|_1 \leqslant 7\varepsilon
	\end{equation*}
as long as $N \geqslant \frac{8\sqrt 2}{\varepsilon}$, $\varepsilon \leqslant e^{-2}$, and $|K| \geqslant \frac{32}{\varepsilon} \sqrt{\log\left( \frac{4N^2}{\varepsilon}\right) \ln(1/\varepsilon) }$.
\end{thm}
\begin{proof}
	Let $\mathcal{M}^{C \rightarrow X}$ be any complete measurement superoperator of the form $\mathcal{M}(\sigma) = \sum_x \alpha_x \ket{x}\bra{\psi_x}\sigma\ket{\psi_x}\bra{x}$, and define $Y$ to be the operator-valued RV which takes value $\psi_x$ with probability $\alpha_x/|C|$. Let $Q$ be the event that $Y_1,\dots,Y_n$ is an $(n,k)$-quasi-measurement, where the $Y_i$ are i.i.d.\ with the same distribution as $Y$. Now, assuming $U$ fulfills the requirements of Lemma \ref{lem:quasi-measurement-concentration}, we have that
	\begin{multline*}
		\left\| \mathcal{M}(\tr_K[U \cdot \rho^{MCK}]) - \mathcal{M}(\pi^C) \otimes \rho^M \right\|_1\\
	\begin{split}
		&= \sum_x \alpha_x \left\| \tr_C[\psi_x (\tr_K[U \cdot \rho^{MCK}] - \pi^C \otimes \rho^M)] \right\|_1\\
		&= |C| \mbE_Y \left\| \tr_C[Y (\tr_K[U \cdot \rho^{MCK}] - \pi^C \otimes \rho^M)] \right\|_1\\
		&= \frac{|C|}{n} \mbE_{Y_1,\dots,Y_n} \sum_{i=1}^n \left\| \tr_C[Y_i (\tr_K[U \cdot \rho^{MCK}] - \pi^C \otimes \rho^M)] \right\|_1\\
		&= \frac{|C|}{n} \Pr\{ Q \} \mbE \left[ \left. \sum_{i=1}^n \left\| \tr_C[Y_i (\tr_K[U \cdot \rho^{MCK}] - \pi^C \otimes \rho^M)] \right\|_1 \right| Q \right]\\
		&\hspace{5mm} + \frac{|C|}{n} \Pr\{ \bar{Q} \} \mbE \left[ \left. \sum_{i=1}^n \left\| \tr_C[Y_i (\tr_K[U \cdot \rho^{MCK}] - \pi^C \otimes \rho^M)] \right\|_1 \right| \bar{Q} \right]\\
		&\leqslant 2^{\demi \log k - \demi \log|K|} + 4\varepsilon + r + 4|C|^2 e^{-n(k-1)^2/|C|2\ln 2}
	\end{split}
	\end{multline*}
	In the above, we have bounded the first conditional expectation using Lemma \ref{lem:quasi-measurement-concentration}, with the $2\varepsilon$ going to $4\varepsilon$ due to the fact that, by definition, any $(n,k)$-quasi-measurement is $\varepsilon$-close to element of $\mfQ$. The second conditional expectation was simply upper bounded by $2n$ (i.e.\ each trace distance in the sum cannot exceed 2) and we used Lemma \ref{lem:locking-operator-chernoff} to bound $\Pr\{ \bar{Q} \}$.

	All that is left to do is to choose the various constants such that  $2n\ln (5/\varepsilon) |C| - N^2r^2/16k^2 < -\ln 2$ as imposed by Lemma \ref{lem:quasi-measurement-concentration}, and such that $4|C|^2 e^{-n(k-1)^2/|C|2\ln 2} \leqslant \varepsilon$. Setting $k=2$ and $r=\varepsilon$ and doing a few simple computations yields that this is possible as long as
\[ n \leqslant \frac{N^2 \varepsilon^2}{512 |C| \ln(1/\varepsilon)} \]
and
\[ n \geqslant 2|C| \log \frac{4|C|^2}{\varepsilon} \]
given that $N \geqslant \frac{8\sqrt{2}}{\varepsilon}$ and $\varepsilon \leqslant e^{-2}$. It follows that choosing $K$ such that
\[ |K| \geqslant \frac{32}{\varepsilon} \sqrt{\log \left( \frac{4N^2}{\varepsilon}  \right)\ln(1/\varepsilon)}\]
suffices to ensure that there exists a $U^{CK}$ such that
	\begin{equation*}
		\left\| \mathcal{M}(\tr_K[U \cdot \rho^{MCK}]) - \mathcal{M}(\pi^C) \otimes \rho^M \right\|_1 \leqslant 7\varepsilon
	\end{equation*}

\end{proof}

\section{Implications for the security of quantum protocols against quantum adversaries}
When designing quantum cryptographic protocols, it is often necessary to show that a quantum adversary (``Eve'') is left with only a negligible amount of information on some secret string. An initial attempt at formalizing this idea is to say that, at the end of the protocol, regardless of what measurement Eve makes on her quantum system, the mutual information between her measurement result and the secret string is at most $\varepsilon$ (in other words, her accessible information about the message is at most $\varepsilon$). This was often taken as the security definition for quantum key distribution, usually implicitly by simply not considering that the adversary might keep quantum data at the end of the protocol \cite{lc99,sp00,NC2000,gl03,lca05} (see also discussion in \cite{bhlmo05,RK04,krbm07}). In \cite{krbm07}, it is shown that this definition of security is inadequate, precisely because of possible locking effects. Indeed, this security definition does not exclude the possibility that Eve, upon gaining partial knowledge of $S$ after the end of the protocol, could then gain more by making a measurement on her quantum register that depends on the partial information that she has learned. In \cite{krbm07}, the authors exhibit a (somewhat contrived) quantum key distribution protocol which generates a secret $n$-bit key such that, if Eve learns the first $n-1$ bits, she can then learn the remaining bit by measuring her own quantum register. 

The locking scheme presented above allows us to demonstrate a much more spectacular failure of this security definition. We will show that there exists a quantum key distribution protocol that ensures that an adversary has negligible accessible information about the final key, but with which an adversary can recover the entire key upon learning only a very small fraction of it.

\subsection{Description of the protocol}
We will derive this protocol by taking a protocol that is truly secure, and then making Alice send a locked version of the secret string directly to Eve. We will be able to prove that regardless of what measurement Eve makes on her state, she will learn essentially no information on the string, but of course, she only needs to learn a tiny amount of information to unlock what Alice sent her. More precisely, let $P$ be a quantum key distribution protocol such that, at the end of its execution, Alice and Bob share an $n$-bit string, and Eve has a quantum state representing everything that she has managed to learn about the string. We will also assume that $P$ is a truly secure protocol: the string together with Eve's quantum state can be represented as a quantum state $\sigma^{SE}$ such that $\| \sigma^{SE} - \pi^S \otimes \sigma^E \|_1 \leqslant \varepsilon$, where $S$ is a quantum register holding the secret string, and $E$ is Eve's quantum register. Now, we will define the protocol $P'$ to be the following quantum key distribution protocol: Alice and Bob first run $P$ to generate a string $s$ of length $n$, and then Alice splits $s$ into two parts: the first part $s_k$ is of size $\log\left( 32/\varepsilon \sqrt{\log\left( \frac{4\cdot 2^{2n}}{\varepsilon} \right) \ln(1/\varepsilon)} \right)$, and the second part $s_c$ contains the rest of the key. Alice then uses the classical key $s_k$ to create a quantum state in register $C$ that contains a locked version of $s_c$ and sends the system $C$ to Eve.

How secure is $P'$? It is clearly very insecure, since, if Eve ever ends up learning $s_k$ (via a known plaintext attack, for instance), she can then completely recover $s_c$. However, the next theorem shows that, right after the execution of $P'$, Eve cannot make any measurement that will reveal information about the key. In particular, $P'$ satisfies the requirement that Eve's accessible information on the key be very low.

\begin{thm}
Let $P$ and $P'$ be quantum key distribution protocols as defined as above, and let $\rho^{CES}$ be the state at the end of the execution of $P'$: $S$ contains the $n$-bit string $s$, $E$ is Eve's quantum register after the execution of $P$, and $C$ contains the locked version of $s_c$ that Alice sent to Eve. Then, for any measurement superoperator $\mathcal{M}^{CE \rightarrow X}$, there exists a state $\xi^X$ such that
\[ \left\| \mathcal{M}(\rho^{CES}) - \xi^X \otimes \pi^S \right\|_1 \leqslant 2\varepsilon. \]
This also entails that
\[ I_{\acc}(S;CE) \leqslant 2\varepsilon n + 2\eta(1-2\varepsilon) + 2\eta(2\varepsilon) \]
via the Alicki-Fannes inequality (Lemma \ref{lem:alicki-fannes}).
\end{thm}
\begin{proof}
From the definition of $P$, we have that
\begin{equation}
\left\| \rho^{ES} - \pi^S \otimes \rho^E \right\|_1 \leqslant \varepsilon.
\end{equation}
Now, let $\mathcal{C}^{S \rightarrow CS}$ be a superoperator that takes a classical string in $S$, splits it into $s_k$ and $s_c$, creates a locked version of $s_c$ with $s_k$ as the key into the quantum system $C$, and leaves the classical string in $S$ unchanged; this is simply the operation that Alice performs when preparing $C$ for Eve. The above inequality, combined with the monotonicity of the trace distance under CPTP maps yields
\begin{equation}
\left\| \rho^{CES} -  \mathcal{C}(\pi^S) \otimes \rho^E \right\|_1 \leqslant \varepsilon
\label{eqn:qkd1}
\end{equation}
and hence, for any measurement superoperator $\mathcal{M}^{CE \rightarrow X}$,
\begin{equation}
\left\| \mathcal{M}(\rho^{CES}) -  \mathcal{M}(\mathcal{C}(\pi^S) \otimes \rho^E) \right\|_1 \leqslant \varepsilon
\label{eqn:qkd2}
\end{equation}
Consider now the expression $\mathcal{M}^{CE \rightarrow X}(\mathcal{C}(\pi^S) \otimes \rho^E)$: it can be viewed as a measurement on the $C$ system of $\mathcal{C}^{S \rightarrow CS}(\pi^S)$ alone that is implemented by creating the state $\rho^E$ and then measuring $\mathcal{M}^{CE \rightarrow X}$. Furthermore, note that, by the definition of an $\varepsilon$-locking scheme, we have that, for every measurement superoperator $\mathcal{N}^{C \rightarrow X}$,
\begin{equation}
\left\| \mathcal{N}(\mathcal{C}(\pi^S)) - \mathcal{N}(\tr_S[\mathcal{C}(\pi^S)]) \otimes \pi^S \right\|_1 \leqslant \varepsilon.
\end{equation}
Applying this to $\mathcal{M}^{CE \rightarrow X}(\mathcal{C}(\pi^S) \otimes \rho^E)$, we get that
\begin{equation}
\left\| \mathcal{M}(\mathcal{C}(\pi^S) \otimes \rho^E) -  \mathcal{M}(\tr_S[\mathcal{C}(\pi^S)] \otimes \rho^E) \otimes \pi^S \right\|_1 \leqslant \varepsilon.
\label{eqn:qkd3}
\end{equation}
We now use the triangle inequality on Equations (\ref{eqn:qkd2}) and (\ref{eqn:qkd3}) to obtain
\begin{equation}
\left\| \mathcal{M}(\rho^{CES}) - \mathcal{M}(\tr_S[\mathcal{C}(\pi^S)] \otimes \rho^E) \otimes \pi^S \right\|_1 \leqslant 2\varepsilon
\end{equation}
which yields the theorem with $\xi^X := \mathcal{M}(\tr_S[\mathcal{C}(\pi^S)] \otimes \rho^E)$.
\end{proof}

Hence, we have shown that requiring that Eve's accessible information on the generated key be low is not an adequate definition of security for quantum key distribution. We have exhibited a protocol $P'$ which guarantees low accessible information and yet is clearly insecure due to locking effects.

\section{Discussion}
The essence of the locking phenomenon is that it is possible to possess \emph{purely} quantum information about a classical message: the cyphertext by itself must contain a lot of information about the message, since only a tiny key is required to get the message, but none of it can be considered classical, since no measurement succeeds in extracting this information. This phenomenon has particular importance in cryptography: it highlights the need to consider an adversary having access to quantum memory, since it is possible for a protocol to ensure that no adversary has any classical information about a particular string while having a lot of quantum information about it. The adversary then needs only a very small amount of additional information to unlock his quantum information. This essentially means that security definitions in cryptography must take quantum information into account to be composable in the physical world.

The main improvement of this work over previous locking schemes is the fact that locking is defined in terms of a trace distance between measurement outputs rather than in terms of accessible information. This is strictly stronger, and has a more compelling interpretation: measurements made on a locked message cannot be distinguished with more than negligible probability from data generated independently of the message. Furthermore, it demonstrates the failure of cryptographic security definitions based on measurement results even more flagrantly: previous results \cite{krbm07} showed that there exists a quantum key distribution protocol that produces an $n$-bit key about which no adversary can obtain significant information through a measurement, but for which there can exist a quantum adversary who, upon learning the first $n-1$ bits of the key, can then learn the last one by measuring his quantum data. In this work, the quantum adversary only needs to get $\polylog(n)$ bits on the key before being able to reconstruct the entire key, rather than $n-1$ bits.

\chapter{Conclusion}\label{chp:conclusion}
In this thesis, we have developed a set of mathematical tools to solve quantum information theory problems within a unified framework. These tools are based on the idea of \emph{decoupling}: in the quantum world, ensuring that two systems are uncorrelated implies that both of these systems are completely correlated with a third system that purifies the state that they are holding. Hence, the problem of information transmission, which can be viewed as the problem of establishing perfect correlation between a sender and a receiver, can be solved by \emph{destroying} correlation between the sender and an ``environment'' system that purifies the global state. Chapter \ref{chp:decoupling} presents this concept in detail and gives a general theorem that allows us to ensure that two systems are decorrelated. This theorem analyzes the following situation: we have a quantum channel $\mathcal{T}^{A \rightarrow E}$ and a quantum state $\rho^{AR}$, we apply a unitary $U$ on the $A$ system of $\rho$ (a unitary chosen at random uniformly over the set of all unitaries works on average)  and then we send $A$ into the input of $\mathcal{T}$. The result is that the quality of decorrelation only depends on two parameters: one that indicates how easy it is to decorrelate the state and the other that measures how good the channel is at decorrelating. Several different versions of this theorem are presented to adapt it to different uses.

The rest of the thesis then goes on to apply these tools to more concrete information theory problems, allowing us to obtain new theorems as well as many of the most important theorems in the field, often in a more general form. These include the best known achievable rate for quantum transmission through quantum channels and the entanglement-assisted capacities of quantum channels for classical and quantum transmission. It also allowed us to come up with hitherto unknown coding theorems on quantum channels with side-information at the transmitter, as well as quantum broadcast channels.

In all of these cases, the coding theorems followed the same pattern: we first obtain a theorem that applies to a single use of a channel, with the quality of transmission depending on various min- and max-entropies. We then specialize these theorems to the case where the ``single channel'' in question is actually $n$ copies of the same channel, yielding an asymptotic result. In this process, the min- and max-entropies are bounded using the fully quantum asymptotic equipartition property, and turn out to become von Neumann entropies.

We end the thesis in Chapter \ref{chp:locking} with an application of decoupling of a slightly different flavour: locking classical information in quantum states. This involves encoding a classical message into two quantum systems: a large one (the cyphertext) that is almost as large as the message itself, and a very small one (the key). The encoding has the property that, given only the cyphertext, no measurement can yield any significant amount of information about the message, even though the cyphertext and the key together provide full information about the message. In contrast with previous work on locking, the definition of locking used here involves a trace distance between two classical distributions: results of a measurement made on a locked message, and measurement results generated independently of the message; this is both a stronger condition and has the clear operational interpretation that a locked message is virtually indistinguishable from a random state when a measurement is made.

\section{Open problems and future research directions}
There are several open problems and possible research projects that arise out of the results presented in this thesis. Here are some of them:

\textbf{A constructive version of the locking scheme:} The results presented in this thesis involve a random unitary chosen according to the Haar measure in some way. This yields proofs that certain protocols exist, but does not directly give a way of actually constructing them. For almost all of the results in this thesis, however, one can replace the Haar measure by a unitary 2-design, since only the second moment of the Haar measure is necessary for the proofs. But there is one exception: the information locking scheme of Chapter \ref{chp:locking}. This is because its proof relies not only on the second moment of the Haar distribution, but also on its concentration properties (Theorem \ref{thm:bertha-concentration}). Indeed, Theorem \ref{thm:bertha-concentration} states that, in the main decoupling theorem (Theorem \ref{thm:bertha}), not only do we get good decoupling on average when choosing a unitary randomly, but also that this holds with overwhelmingly high probability. Statements of this nature abound in quantum information theory and its applications are far from being limited to information locking: it can be used to show the existence of completely entangled subspaces \cite{generic-entanglement}, to prove the existence of counterexamples to the additivity conjecture \cite{hastings-nonadditivity}. In all of these cases (as well as locking), it would be of great interest to have explicit, constructive examples. Finding a constructive version of Theorem \ref{thm:bertha-concentration} (or perhaps something slightly more general) would most likely achieve this for all of the problems mentioned.

\textbf{Min-entropy bounds for larger classes of states:} For all of the channel coding problems shown in this document, we have proceeded as follows: we first gave a general one-shot coding theorem, and then we used it to give a theorem for memoryless channels. To do this, we used the fully quantum asymptotic equipartition property (Theorem \ref{thm:fully-quantum-aep}, \cite{tcr08}) to get a bound on the smooth min-entropy of i.i.d.\ states. If we had a way to similarly bound the smooth min-entropy of a larger class of states, we could apply it to a larger class of channels, such as various types of channels with memory.

\textbf{Optimality of the one-shot coding theorems:} The various one-shot coding theorems presented were left without converses. However, it seems likely that they are, in fact, optimal, at least for some particular input distributions. Some special cases have already been shown to be optimal, such as state merging \cite{merging-berta-etal}.

\textbf{Systematically relating classical information theory and quantum information theory with free entanglement:} In quantum Shannon theory, it has very often proven to be the case that the quantum problems that bear the strongest resemblance to their classical counterparts are those in which the various participants share entanglement before the protocol starts and are allowed to use it to improve the performance of the protocol. This is the case for information transmission through a regular channel: Shannon showed that the mutual information gives the capacity of a classical channel; and it turns out that the quantum mutual information characterizes the \emph{entanglement-assisted} capacity of quantum channels. This is also true for channels with side-information at the transmission (see Chapter \ref{chp:side-info}) as well as broadcast channels (see Chapter \ref{chp:bcast}). In all of these cases, we get essentially identical expressions for the capacities (or achievable rate regions) in the classical and in the quantum case. This suggests that there might be a general principle at work relating the two. Such a principle would allow us to automatically import large classes of results from the extremely vast body of work in classical information theory directly into quantum information theory. It is not clear at this point, however, to what extent this principle would apply, or what the most appropriate definitions would be.

\bibliographystyle{alpha}
\bibliography{these}

\debutannexes
\annexe{Various technical lemmas} \label{appendix:techlemmas}
In this section, we state (and usually prove) various technical lemmas used at various points throughout the thesis.

The first lemma is a simple application of the triangle inequality:
\begin{lem}\label{lem:remove-typical-proj}
Let $\rho$, $\rho'$ and $\sigma$ be positive semidefinite operators on $A$ such that $\| \rho - \sigma \|_1 \leqslant \varepsilon$, $\tr[\rho'] \leqslant \tr[\sigma]$, and $\rho' \geqslant \rho$. Then, $\| \rho' - \sigma \|_1 \leqslant 2\varepsilon$.
\end{lem}
\begin{proof}
We have that
\begin{align}
\| \rho' - \rho \|_1 &= \tr[\rho' - \rho]\\
&\leqslant \tr[\sigma - \rho]\\
&\leqslant \varepsilon
\end{align}
and hence
\begin{equation}
\| \rho' - \sigma \|_1 \leqslant \| \rho - \sigma \|_1 + \| \rho' - \rho \|_1 \leqslant 2\varepsilon.
\end{equation}
\end{proof}

We then prove the following operator inequalities:
\begin{lem}\label{lem:fiou}
Let $\rho^{AB}$ be positive semidefinite, and let $0 \leqslant P^B \leqslant \ident^B$. Then,
\[ \tr_B[P^B \rho^{AB} P^B] \leqslant \rho^A \]
\end{lem}
\begin{proof}
Let $M^A$ be any positive semidefinite operator. Then,
\begin{align*}
\tr[M^A \tr_B [P^B \rho^{AB} P^B]] &= \tr[(M^A \otimes \ident^B)(P^B \rho^{AB} P^B)]\\
&= \tr[(M^A \otimes {P^B}^2) \rho^{AB}]\\
&\leqslant \tr[(M^A \otimes \ident^B)\rho^{AB}]\\
&= \tr[M^A \rho^A]
\end{align*}
where we have used the fact that tensoring with the identity is the adjoint of the trace superoperator, as well as the fact that ${P^B}^2 \leqslant \ident^B$. Since this is true for every positive semidefinite $M^A$, the lemma follows.
\end{proof}

\begin{lem}\label{lem:fiou2}
	Let $\ket{\psi}^{AB} \in \sfA \otimes \sfB$, $\rho^A \in \Pos(\sfA)$ such that $\rho^A \leqslant \psi^A$. Then, there exists a $P^B \in \Pos(\sfB)$ such that $P^B \leqslant \ident^B$ and $\tr_B[P^B \cdot \psi^{AB}] = \rho^A$.
\end{lem}
\begin{proof}
	Without loss of generality, let $A$ and $B$ be equal to the support of $\psi^A$ and $\psi^B$ respectively. Define the partial isometry $V^{B \rightarrow A} = {\psi^A}^{-1/2} \op_{B \rightarrow A}(\ket{\psi}) = \op_{B \rightarrow A}(\ket{\psi}) {\psi^B_T}^{-1/2}$ where the $T$ subscript denotes transposition. Now,
	\begin{align*}
		\rho^A &= VV\mdag \rho VV\mdag\\
		&= \op_{B \rightarrow A}(\ket{\psi}) {\psi^B_T}^{-1/2} V\mdag \rho V {\psi^B_T}^{-1/2} \op_{B \rightarrow A}(\ket{\psi})\\
		&= \op_{B \rightarrow A}(\ket{\psi}) V\mdag {\psi^A}^{-1/2} \rho {\psi^A}^{-1/2} V \op_{B \rightarrow A}(\ket{\psi})\\
		&= \op_{B \rightarrow A}(\ket{\psi}) {P^B_T}^{2} \op_{B \rightarrow A}(\ket{\psi})\mdag\\
		&= \op_{B \rightarrow A}(P^B \ket{\psi}) \op_{B \rightarrow A}(P^B \ket{\psi})\mdag\\
		&= \tr_B[P^B \cdot \psi^{AB}]
	\end{align*}
	where we have defined ${P^B_T}^2 = V\mdag {\psi^A}^{-1/2} \rho {\psi^A}^{-1/2} V \in \Pos(\sfB)$ and the $T$ subscript denotes transposition. We can now easily check that ${P^B_T}^2 \leqslant \ident^B$ since $\rho \leqslant \psi^A$ implies that ${\psi^A}^{-1/2} \rho {\psi^A}^{-1/2} \leqslant \ident^A$.
\end{proof}

The following lemma comes from Lemma II.4 from \cite{HLSW03}:
\begin{lem}\label{lem:onenorm-twonorm}
	Given two normalized vectors $\ket{\psi}$ and $\ket{\varphi}$ in $\sfA$, we have that
	\[ \left\| \psi - \varphi \right\|_1 \leqslant 2 \left\| \ket{\psi} - \ket{\varphi} \right\|_2 \]
\end{lem}
\begin{proof}
	By Lemma \ref{lem:pseudo-jensen-renato} with $\sigma$ as the projector onto the 2-dimensional support of $\psi - \varphi$, we have that
	\begin{align*}
		\left\| \psi - \varphi \right\|_1 &\leqslant \sqrt{2 \tr[(\psi - \varphi)^2]}\\
		&= 2 \sqrt{1 - \tr[\varphi \psi]}\\
		&= 2 \sqrt{1 - |\braket{\psi}{\varphi}|^2}\\
		&= 2 \sqrt{(1 - |\braket{\psi}{\varphi}|)(1 + |\braket{\psi}{\varphi}|)}\\
		&\leqslant 2 \sqrt{ 2 - 2 |\braket{\psi}{\varphi}| }\\
		&\leqslant 2 \sqrt{ 2 - \braket{\psi}{\varphi} - \braket{\varphi}{\psi} }\\
		&= 2 \sqrt{(\bra{\psi} - \bra{\varphi})(\ket{\psi} - \ket{\varphi})}\\
		&= 2 \| \ket{\psi} - \ket{\varphi} \|_2
	\end{align*}
\end{proof}

The next two lemmas are simple inequalities regarding operator norms:

\begin{lem}\label{lem:norm-prod-matrices}
Let $M^{A \rightarrow B}$ and $N^{B \rightarrow C}$ be arbitrary matrices. Then,
\[ \| NM \|_2 \leqslant \| N \|_2 \| M \|_{\infty} \]
\end{lem}
\begin{proof}
Let $U^{B \rightarrow A}$ be an isometry such that $P^B := MU$ is positive semidefinite (such an isometry can be seen to exist by taking the singular-value decomposition of $M$). Then, we have that
\begin{align}
\| NM \|_2 &= \| NP \|_2\\
&= \sqrt{\tr[NP^2 N\mdag}]\\
&\leqslant \| P \|_{\infty} \sqrt{\tr[NN\mdag]}\\
&= \| M \|_{\infty} \|N\|_2
\end{align}
where the inequality comes from the matrix inequality $P^2 \leqslant \|P\|_{\infty}^2 \ident$.
\end{proof}

\begin{lem}\label{lem:tracenorm-maxu}
Let $M^{A \rightarrow B}$ be an arbitrary matrix. Then,
\[ \| M \|_1 = \max_{V^{B \rightarrow A}} | \tr[VM] | \]
where the maximization is taken over all partial isometries $V^{B \rightarrow A}$.
\end{lem}
\begin{proof}
Let us decompose $M$ as $M = \sum \alpha_j \ket{\psi_j}\bra{\varphi_j}$ where the $\ket{\psi_j}^B$ are orthonormal, as are the $\ket{\varphi_j}^A$, and the $\alpha_j$ are the singular values of $M$. Furthermore, let $W^{B \rightarrow A}$ be a partial isometry such that $W \ket{\psi_j} = \ket{\varphi_j}$. Then,
\begin{align*}
\| M \|_1 &= |\tr[WM]|\\
&\leqslant \max_{V^{B \rightarrow A}} |\tr[VM]|\\
&= \max_V \left| \sum_{j} \alpha_j \tr[V \ket{\psi_j}\bra{\varphi_j}] \right|\\
&\leqslant \max_V \sum_j  \alpha_j \left| \bra{\varphi_j} V \ket{\psi_j} \right|\\
&\leqslant \sum_j \alpha_j\\
&= \| M \|_1
\end{align*}
\end{proof}

The next lemma is simply Markov's inequality, which we use several times to assert the existence of a unitary satisfying many conditions at once:
\begin{lem}[Markov's inequality]\label{lem:markov}
Let $X$ be a random variable which is always positive. Then,
\[ \Pr\{ X \geqslant k \mbE X \} \leqslant \frac{1}{k} \]
Hence, for example, if $f_1,\dots f_k: \mbU \rightarrow \mbR_+$, then, there exists a $U$ such that
\begin{align*}
	f_1(U) &\leqslant (k+1) \mbE f_1(U)\\
	\vdots\\
	f_k(U) &\leqslant (k+1) \mbE f_k(U)\\
\end{align*}
by the union bound.
\end{lem}

The next lemma is known as the operator Chernoff bound and was first proven in \cite{ahlswede-winter}:
\begin{lem}[Operator Chernoff bound]\label{lem:operator-chernoff}
  Let $X_1,\ldots,X_M$ be i.i.d.\ random variables taking values in the operators $\Pos(\sfA)$, with $0\leqslant X_j\leqslant \ident$, with $A=\mbE X_j\geqslant\alpha \ident$, and let $0<\eta \leqslant 1/2$. Then
  \begin{equation}
    \Pr \left\{ \frac{1}{M}\sum_{j=1}^M X_j \nleqslant (1+\eta)A \right\} \leqslant 2|A| \exp\left( -M\frac{\alpha\eta^2}{2\ln 2} \right).
  \end{equation}
\end{lem}

We also need Fannes's inequality \cite{fannes} as well as its relative, the Alicki-Fannes inequality \cite{alicki-fannes}:
\begin{lem}[Fannes's inequality \cite{fannes}]\label{thm:fannes}
Let $\rho$ and $\sigma$ be density operators on $A$ such that $\| \rho - \sigma \|_1 \leqslant 1/e$. Then,
\[ | H(A)_{\rho} - H(A)_{\sigma} | \leqslant \| \rho - \sigma \|_1 \log|A| + \eta\left( \| \rho - \sigma \|_1 \right) \]
where $\eta(x) := -x \log x$ and $e$ is the base of the natural logarithm.
\end{lem}

\begin{lem}[Alicki-Fannes inequality \cite{alicki-fannes}]\label{lem:alicki-fannes}
Given two states $\rho^{AB} \in \DD(\sfA \otimes \sfB)$ and $\sigma^{AB} \in \DD(\sfA \otimes \sfB)$, with $\| \rho^{AB} - \sigma^{AB} \|_1 = \varepsilon$, the following holds:
\[ \left| H(A|B)_{\rho} - H(A|B)_{\sigma} \right| \leqslant 4\varepsilon \log|A| + 2\eta(1-\varepsilon) + 2\eta(\varepsilon) \]
where $\eta$ is defined as above.
\end{lem}

The locking chapter needs the concept of $\varepsilon$-nets. The following definition and lemma were taken from \cite{HLSW03}, but these concepts are used rather extensively in other areas of mathematics, particularly in random matrix theory.
\begin{defin}[$\varepsilon$-net]\label{def:nets}
A set of pure states $\mfN \subseteq \sfA$ is called an $\varepsilon$-net if, for every normalized $\ket{\psi} \in \sfA$, there exists a $\ket{\varphi} \in \mfN$ such that $\| \ket{\psi} - \ket{\varphi} \|_2 \leqslant \varepsilon/2$ and $\| \psi - \varphi \|_1 \leqslant \varepsilon$.
\end{defin}
\begin{lem}[Existence of small nets]
For any Hilbert space $\sfA$ of dimension $|A|$, there exists an $\varepsilon$-net $\mfN \subseteq \sfA$ of size $|\mfN| \leqslant \left( \frac{5}{\varepsilon} \right)^{2|A|}$.
\end{lem}

\end{document}

%% file: reg-1shot-real.pdf_t
\begin{picture}(0,0)%
\includegraphics{reg-1shot-real.pdf}%
\end{picture}%
\setlength{\unitlength}{4144sp}%
\begingroup\makeatletter\ifx\SetFigFontNFSS\undefined%
\gdef\SetFigFontNFSS#1#2#3#4#5{%
  \reset@font\fontsize{#1}{#2pt}%
  \fontfamily{#3}\fontseries{#4}\fontshape{#5}%
  \selectfont}%
\fi\endgroup%
\begin{picture}(6458,2229)(256,-1828)
\put(6166,-286){\makebox(0,0)[lb]{\smash{{\SetFigFontNFSS{12}{14.4}{\rmdefault}{\mddefault}{\updefault}{\color[rgb]{0,0,0}$\ket{\psi}^{ABR}$}%
}}}}
\put(1666,164){\makebox(0,0)[lb]{\smash{{\SetFigFontNFSS{12}{14.4}{\rmdefault}{\mddefault}{\updefault}{\color[rgb]{0,0,0}$R$}%
}}}}
\put(1576,-1501){\makebox(0,0)[lb]{\smash{{\SetFigFontNFSS{12}{14.4}{\rmdefault}{\mddefault}{\updefault}{\color[rgb]{0,0,0}$A$}%
}}}}
\put(2611,-1501){\makebox(0,0)[lb]{\smash{{\SetFigFontNFSS{12}{14.4}{\rmdefault}{\mddefault}{\updefault}{\color[rgb]{0,0,0}$A'$}%
}}}}
\put(3781,-1591){\makebox(0,0)[lb]{\smash{{\SetFigFontNFSS{12}{14.4}{\rmdefault}{\mddefault}{\updefault}{\color[rgb]{0,0,0}$E$}%
}}}}
\put(3736,-916){\makebox(0,0)[lb]{\smash{{\SetFigFontNFSS{12}{14.4}{\rmdefault}{\mddefault}{\updefault}{\color[rgb]{0,0,0}$C$}%
}}}}
\put(5221,-1051){\makebox(0,0)[lb]{\smash{{\SetFigFontNFSS{12}{14.4}{\rmdefault}{\mddefault}{\updefault}{\color[rgb]{0,0,0}$F$}%
}}}}
\put(5716,-781){\makebox(0,0)[lb]{\smash{{\SetFigFontNFSS{12}{14.4}{\rmdefault}{\mddefault}{\updefault}{\color[rgb]{0,0,0}$A$}%
}}}}
\put(2116,-1276){\makebox(0,0)[lb]{\smash{{\SetFigFontNFSS{12}{14.4}{\rmdefault}{\mddefault}{\updefault}{\color[rgb]{0,0,0}$V$}%
}}}}
\put(4681,-511){\makebox(0,0)[lb]{\smash{{\SetFigFontNFSS{12}{14.4}{\rmdefault}{\mddefault}{\updefault}{\color[rgb]{0,0,0}$D$}%
}}}}
\put(271,-511){\makebox(0,0)[lb]{\smash{{\SetFigFontNFSS{12}{14.4}{\familydefault}{\mddefault}{\updefault}{\color[rgb]{0,0,0}$\ket{\psi}^{ABR}$}%
}}}}
\put(3151,-1276){\makebox(0,0)[lb]{\smash{{\SetFigFontNFSS{12}{14.4}{\rmdefault}{\mddefault}{\updefault}{\color[rgb]{0,0,0}$U_{\mathcal{N}}$}%
}}}}
\put(5941,-1141){\makebox(0,0)[lb]{\smash{{\SetFigFontNFSS{12}{14.4}{\rmdefault}{\mddefault}{\updefault}{\color[rgb]{0,0,0}$\ket{\xi}^{EF}$}%
}}}}
\put(6346,-556){\makebox(0,0)[lb]{\smash{{\SetFigFontNFSS{12}{14.4}{\familydefault}{\mddefault}{\updefault}{\color[rgb]{0,0,0}Bob}%
}}}}
\put(406,-1141){\makebox(0,0)[lb]{\smash{{\SetFigFontNFSS{12}{14.4}{\familydefault}{\mddefault}{\updefault}{\color[rgb]{0,0,0}Alice}%
}}}}
\put(5266,-196){\makebox(0,0)[lb]{\smash{{\SetFigFontNFSS{12}{14.4}{\rmdefault}{\mddefault}{\updefault}{\color[rgb]{0,0,0}$B$}%
}}}}
\put(1576,-376){\makebox(0,0)[lb]{\smash{{\SetFigFontNFSS{12}{14.4}{\rmdefault}{\mddefault}{\updefault}{\color[rgb]{0,0,0}$B$}%
}}}}
\put(451,-241){\makebox(0,0)[lb]{\smash{{\SetFigFontNFSS{12}{14.4}{\familydefault}{\mddefault}{\updefault}{\color[rgb]{0,0,0}Bob}%
}}}}
\end{picture}%

%% file: reg-1shot-omega.pdf_t
\begin{picture}(0,0)%
\includegraphics{reg-1shot-omega.pdf}%
\end{picture}%
\setlength{\unitlength}{4144sp}%
\begingroup\makeatletter\ifx\SetFigFontNFSS\undefined%
\gdef\SetFigFontNFSS#1#2#3#4#5{%
  \reset@font\fontsize{#1}{#2pt}%
  \fontfamily{#3}\fontseries{#4}\fontshape{#5}%
  \selectfont}%
\fi\endgroup%
\begin{picture}(3315,1561)(121,-980)
\put(1621,-556){\makebox(0,0)[lb]{\smash{{\SetFigFontNFSS{12}{14.4}{\rmdefault}{\mddefault}{\updefault}{\color[rgb]{0,0,0}$U_{\mathcal{N}}$}%
}}}}
\put(1036,-421){\makebox(0,0)[lb]{\smash{{\SetFigFontNFSS{12}{14.4}{\rmdefault}{\mddefault}{\updefault}{\color[rgb]{0,0,0}$A'$}%
}}}}
\put(1666,434){\makebox(0,0)[lb]{\smash{{\SetFigFontNFSS{12}{14.4}{\rmdefault}{\mddefault}{\updefault}{\color[rgb]{0,0,0}$A''$}%
}}}}
\put(2341,-916){\makebox(0,0)[lb]{\smash{{\SetFigFontNFSS{12}{14.4}{\rmdefault}{\mddefault}{\updefault}{\color[rgb]{0,0,0}$E$}%
}}}}
\put(3421,-241){\makebox(0,0)[lb]{\smash{{\SetFigFontNFSS{12}{14.4}{\rmdefault}{\mddefault}{\updefault}{\color[rgb]{0,0,0}$\ket{\omega}^{A''CE}$}%
}}}}
\put(136,-151){\makebox(0,0)[lb]{\smash{{\SetFigFontNFSS{12}{14.4}{\rmdefault}{\mddefault}{\updefault}{\color[rgb]{0,0,0}$\ket{\sigma}^{A''A'}$}%
}}}}
\put(2341,-286){\makebox(0,0)[lb]{\smash{{\SetFigFontNFSS{12}{14.4}{\rmdefault}{\mddefault}{\updefault}{\color[rgb]{0,0,0}$C$}%
}}}}
\end{picture}%

%% file: si-1shot-real.pdf_t
\begin{picture}(0,0)%
\includegraphics{si-1shot-real.pdf}%
\end{picture}%
\setlength{\unitlength}{4144sp}%
\begingroup\makeatletter\ifx\SetFigFontNFSS\undefined%
\gdef\SetFigFontNFSS#1#2#3#4#5{%
  \reset@font\fontsize{#1}{#2pt}%
  \fontfamily{#3}\fontseries{#4}\fontshape{#5}%
  \selectfont}%
\fi\endgroup%
\begin{picture}(7764,3051)(169,-2425)
\put(1396,-1636){\makebox(0,0)[lb]{\smash{{\SetFigFontNFSS{12}{14.4}{\familydefault}{\mddefault}{\updefault}{\color[rgb]{0,0,0}$S'$}%
}}}}
\put(2026,-2356){\makebox(0,0)[lb]{\smash{{\SetFigFontNFSS{12}{14.4}{\familydefault}{\mddefault}{\updefault}{\color[rgb]{0,0,0}$S$}%
}}}}
\put(2566,-1501){\makebox(0,0)[lb]{\smash{{\SetFigFontNFSS{12}{14.4}{\familydefault}{\mddefault}{\updefault}{\color[rgb]{0,0,0}$D$}%
}}}}
\put(4771,-1591){\makebox(0,0)[lb]{\smash{{\SetFigFontNFSS{12}{14.4}{\familydefault}{\mddefault}{\updefault}{\color[rgb]{0,0,0}$E$}%
}}}}
\put(6706,-826){\makebox(0,0)[lb]{\smash{{\SetFigFontNFSS{12}{14.4}{\familydefault}{\mddefault}{\updefault}{\color[rgb]{0,0,0}$G$}%
}}}}
\put(7516,119){\makebox(0,0)[lb]{\smash{{\SetFigFontNFSS{12}{14.4}{\rmdefault}{\mddefault}{\updefault}{\color[rgb]{0,0,0}$\ket{\psi}^{ABR}$}%
}}}}
\put(7381,-1456){\makebox(0,0)[lb]{\smash{{\SetFigFontNFSS{12}{14.4}{\rmdefault}{\mddefault}{\updefault}{\color[rgb]{0,0,0}$\ket{\xi}^{GED}$}%
}}}}
\put(451,-1951){\makebox(0,0)[lb]{\smash{{\SetFigFontNFSS{12}{14.4}{\familydefault}{\mddefault}{\updefault}{\color[rgb]{0,0,0}$\ket{\phi}^{SS'}$}%
}}}}
\put(2026,-1186){\makebox(0,0)[lb]{\smash{{\SetFigFontNFSS{12}{14.4}{\familydefault}{\mddefault}{\updefault}{\color[rgb]{0,0,0}$V$}%
}}}}
\put(4186,-1186){\makebox(0,0)[lb]{\smash{{\SetFigFontNFSS{12}{14.4}{\familydefault}{\mddefault}{\updefault}{\color[rgb]{0,0,0}$U_{\mathcal{N}}$}%
}}}}
\put(6166,-421){\makebox(0,0)[lb]{\smash{{\SetFigFontNFSS{12}{14.4}{\familydefault}{\mddefault}{\updefault}{\color[rgb]{0,0,0}$D$}%
}}}}
\put(6706,-106){\makebox(0,0)[lb]{\smash{{\SetFigFontNFSS{12}{14.4}{\familydefault}{\mddefault}{\updefault}{\color[rgb]{0,0,0}$B$}%
}}}}
\put(6706,-331){\makebox(0,0)[lb]{\smash{{\SetFigFontNFSS{12}{14.4}{\familydefault}{\mddefault}{\updefault}{\color[rgb]{0,0,0}$A$}%
}}}}
\put(1396,-151){\makebox(0,0)[lb]{\smash{{\SetFigFontNFSS{12}{14.4}{\familydefault}{\mddefault}{\updefault}{\color[rgb]{0,0,0}$B$}%
}}}}
\put(1396,389){\makebox(0,0)[lb]{\smash{{\SetFigFontNFSS{12}{14.4}{\familydefault}{\mddefault}{\updefault}{\color[rgb]{0,0,0}$R$}%
}}}}
\put(1396,-871){\makebox(0,0)[lb]{\smash{{\SetFigFontNFSS{12}{14.4}{\familydefault}{\mddefault}{\updefault}{\color[rgb]{0,0,0}$A$}%
}}}}
\put(3241,-871){\makebox(0,0)[lb]{\smash{{\SetFigFontNFSS{12}{14.4}{\familydefault}{\mddefault}{\updefault}{\color[rgb]{0,0,0}$A'$}%
}}}}
\put(4771,-916){\makebox(0,0)[lb]{\smash{{\SetFigFontNFSS{12}{14.4}{\familydefault}{\mddefault}{\updefault}{\color[rgb]{0,0,0}$C$}%
}}}}
\put(226,-241){\makebox(0,0)[lb]{\smash{{\SetFigFontNFSS{12}{14.4}{\rmdefault}{\mddefault}{\updefault}{\color[rgb]{0,0,0}$\ket{\psi}^{ABR}$}%
}}}}
\put(226,-1186){\makebox(0,0)[lb]{\smash{{\SetFigFontNFSS{12}{14.4}{\familydefault}{\mddefault}{\updefault}{\color[rgb]{0,0,0}Alice}%
}}}}
\put(7381,-511){\makebox(0,0)[lb]{\smash{{\SetFigFontNFSS{12}{14.4}{\familydefault}{\mddefault}{\updefault}{\color[rgb]{0,0,0}Bob}%
}}}}
\end{picture}%

%% file: si-1shot-omega.pdf_t
\begin{picture}(0,0)%
\includegraphics{si-1shot-omega.pdf}%
\end{picture}%
\setlength{\unitlength}{4144sp}%
\begingroup\makeatletter\ifx\SetFigFontNFSS\undefined%
\gdef\SetFigFontNFSS#1#2#3#4#5{%
  \reset@font\fontsize{#1}{#2pt}%
  \fontfamily{#3}\fontseries{#4}\fontshape{#5}%
  \selectfont}%
\fi\endgroup%
\begin{picture}(3495,1329)(121,-748)
\put(136,-241){\makebox(0,0)[lb]{\smash{{\SetFigFontNFSS{12}{14.4}{\rmdefault}{\mddefault}{\updefault}{\color[rgb]{0,0,0}$\ket{\sigma}^{A''A'SD}$}%
}}}}
\put(3601,-241){\makebox(0,0)[lb]{\smash{{\SetFigFontNFSS{12}{14.4}{\rmdefault}{\mddefault}{\updefault}{\color[rgb]{0,0,0}$\ket{\omega}^{A''CED}$}%
}}}}
\put(1531, 74){\makebox(0,0)[lb]{\smash{{\SetFigFontNFSS{12}{14.4}{\rmdefault}{\mddefault}{\updefault}{\color[rgb]{0,0,0}$A'$}%
}}}}
\put(2791, 74){\makebox(0,0)[lb]{\smash{{\SetFigFontNFSS{12}{14.4}{\rmdefault}{\mddefault}{\updefault}{\color[rgb]{0,0,0}$C$}%
}}}}
\put(2791,-286){\makebox(0,0)[lb]{\smash{{\SetFigFontNFSS{12}{14.4}{\rmdefault}{\mddefault}{\updefault}{\color[rgb]{0,0,0}$E$}%
}}}}
\put(1531,-286){\makebox(0,0)[lb]{\smash{{\SetFigFontNFSS{12}{14.4}{\rmdefault}{\mddefault}{\updefault}{\color[rgb]{0,0,0}$S$}%
}}}}
\put(1531,-646){\makebox(0,0)[lb]{\smash{{\SetFigFontNFSS{12}{14.4}{\rmdefault}{\mddefault}{\updefault}{\color[rgb]{0,0,0}$D$}%
}}}}
\put(1531,434){\makebox(0,0)[lb]{\smash{{\SetFigFontNFSS{12}{14.4}{\rmdefault}{\mddefault}{\updefault}{\color[rgb]{0,0,0}$A''$}%
}}}}
\put(2206,-241){\makebox(0,0)[lb]{\smash{{\SetFigFontNFSS{12}{14.4}{\rmdefault}{\mddefault}{\updefault}{\color[rgb]{0,0,0}$U_{\mathcal{N}}$}%
}}}}
\end{picture}%

%% file: broadcast-1shot-real.pdf_t
\begin{picture}(0,0)%
\includegraphics{broadcast-1shot-real.pdf}%
\end{picture}%
\setlength{\unitlength}{4144sp}%
\begingroup\makeatletter\ifx\SetFigFontNFSS\undefined%
\gdef\SetFigFontNFSS#1#2#3#4#5{%
  \reset@font\fontsize{#1}{#2pt}%
  \fontfamily{#3}\fontseries{#4}\fontshape{#5}%
  \selectfont}%
\fi\endgroup%
\begin{picture}(6837,3636)(661,-2929)
\put(2701,-1321){\makebox(0,0)[lb]{\smash{{\SetFigFontNFSS{12}{14.4}{\familydefault}{\mddefault}{\updefault}{\color[rgb]{0,0,0}$W$}%
}}}}
\put(4051,-1321){\makebox(0,0)[lb]{\smash{{\SetFigFontNFSS{12}{14.4}{\familydefault}{\mddefault}{\updefault}{\color[rgb]{0,0,0}$U_{\mathcal{N}}$}%
}}}}
\put(5311,-331){\makebox(0,0)[lb]{\smash{{\SetFigFontNFSS{12}{14.4}{\familydefault}{\mddefault}{\updefault}{\color[rgb]{0,0,0}$D_1$}%
}}}}
\put(5266,-2311){\makebox(0,0)[lb]{\smash{{\SetFigFontNFSS{12}{14.4}{\familydefault}{\mddefault}{\updefault}{\color[rgb]{0,0,0}$D_2$}%
}}}}
\put(2116,-2086){\makebox(0,0)[lb]{\smash{{\SetFigFontNFSS{12}{14.4}{\familydefault}{\mddefault}{\updefault}{\color[rgb]{0,0,0}$B_2$}%
}}}}
\put(2116,-1636){\makebox(0,0)[lb]{\smash{{\SetFigFontNFSS{12}{14.4}{\familydefault}{\mddefault}{\updefault}{\color[rgb]{0,0,0}$A_2$}%
}}}}
\put(2116,-961){\makebox(0,0)[lb]{\smash{{\SetFigFontNFSS{12}{14.4}{\familydefault}{\mddefault}{\updefault}{\color[rgb]{0,0,0}$A_1$}%
}}}}
\put(2116,-241){\makebox(0,0)[lb]{\smash{{\SetFigFontNFSS{12}{14.4}{\familydefault}{\mddefault}{\updefault}{\color[rgb]{0,0,0}$B_1$}%
}}}}
\put(2206,524){\makebox(0,0)[lb]{\smash{{\SetFigFontNFSS{12}{14.4}{\familydefault}{\mddefault}{\updefault}{\color[rgb]{0,0,0}$R_1$}%
}}}}
\put(5266,-1096){\makebox(0,0)[lb]{\smash{{\SetFigFontNFSS{12}{14.4}{\familydefault}{\mddefault}{\updefault}{\color[rgb]{0,0,0}$E$}%
}}}}
\put(4411,-601){\makebox(0,0)[lb]{\smash{{\SetFigFontNFSS{12}{14.4}{\familydefault}{\mddefault}{\updefault}{\color[rgb]{0,0,0}$C_1$}%
}}}}
\put(3331,-1051){\makebox(0,0)[lb]{\smash{{\SetFigFontNFSS{12}{14.4}{\familydefault}{\mddefault}{\updefault}{\color[rgb]{0,0,0}$A'$}%
}}}}
\put(6706,-1276){\makebox(0,0)[lb]{\smash{{\SetFigFontNFSS{12}{14.4}{\familydefault}{\mddefault}{\updefault}{\color[rgb]{0,0,0}$\ket{\xi}^{G_1G_2ED}$}%
}}}}
\put(6571,-16){\makebox(0,0)[lb]{\smash{{\SetFigFontNFSS{12}{14.4}{\familydefault}{\mddefault}{\updefault}{\color[rgb]{0,0,0}$\ket{\psi}^{A_1B_1R_1}$}%
}}}}
\put(6526,-2581){\makebox(0,0)[lb]{\smash{{\SetFigFontNFSS{12}{14.4}{\familydefault}{\mddefault}{\updefault}{\color[rgb]{0,0,0}$\ket{\psi}^{A_2B_2R_2}$}%
}}}}
\put(5851,-2671){\makebox(0,0)[lb]{\smash{{\SetFigFontNFSS{12}{14.4}{\familydefault}{\mddefault}{\updefault}{\color[rgb]{0,0,0}$B_2$}%
}}}}
\put(6256,-376){\makebox(0,0)[lb]{\smash{{\SetFigFontNFSS{12}{14.4}{\familydefault}{\mddefault}{\updefault}{\color[rgb]{0,0,0}$A_1$}%
}}}}
\put(5761,-16){\makebox(0,0)[lb]{\smash{{\SetFigFontNFSS{12}{14.4}{\familydefault}{\mddefault}{\updefault}{\color[rgb]{0,0,0}$B_1$}%
}}}}
\put(6166,-2266){\makebox(0,0)[lb]{\smash{{\SetFigFontNFSS{12}{14.4}{\familydefault}{\mddefault}{\updefault}{\color[rgb]{0,0,0}$A_2$}%
}}}}
\put(3196,-1681){\makebox(0,0)[lb]{\smash{{\SetFigFontNFSS{12}{14.4}{\familydefault}{\mddefault}{\updefault}{\color[rgb]{0,0,0}$D$}%
}}}}
\put(1126,-1366){\makebox(0,0)[lb]{\smash{{\SetFigFontNFSS{12}{14.4}{\familydefault}{\mddefault}{\updefault}{\color[rgb]{0,0,0}Alice}%
}}}}
\put(6976,-331){\makebox(0,0)[lb]{\smash{{\SetFigFontNFSS{12}{14.4}{\familydefault}{\mddefault}{\updefault}{\color[rgb]{0,0,0}Bob 1}%
}}}}
\put(6976,-2311){\makebox(0,0)[lb]{\smash{{\SetFigFontNFSS{12}{14.4}{\familydefault}{\mddefault}{\updefault}{\color[rgb]{0,0,0}Bob2}%
}}}}
\put(2161,-2851){\makebox(0,0)[lb]{\smash{{\SetFigFontNFSS{12}{14.4}{\familydefault}{\mddefault}{\updefault}{\color[rgb]{0,0,0}$R_2$}%
}}}}
\put(676,-421){\makebox(0,0)[lb]{\smash{{\SetFigFontNFSS{12}{14.4}{\familydefault}{\mddefault}{\updefault}{\color[rgb]{0,0,0}$\ket{\psi}^{A_1B_1R_1}$}%
}}}}
\put(676,-2176){\makebox(0,0)[lb]{\smash{{\SetFigFontNFSS{12}{14.4}{\familydefault}{\mddefault}{\updefault}{\color[rgb]{0,0,0}$\ket{\psi}^{A_2B_2R_2}$}%
}}}}
\put(4681,-1546){\makebox(0,0)[lb]{\smash{{\SetFigFontNFSS{12}{14.4}{\familydefault}{\mddefault}{\updefault}{\color[rgb]{0,0,0}$C_2$}%
}}}}
\put(6346,-1681){\makebox(0,0)[lb]{\smash{{\SetFigFontNFSS{12}{14.4}{\familydefault}{\mddefault}{\updefault}{\color[rgb]{0,0,0}$G_2$}%
}}}}
\put(6346,-916){\makebox(0,0)[lb]{\smash{{\SetFigFontNFSS{12}{14.4}{\familydefault}{\mddefault}{\updefault}{\color[rgb]{0,0,0}$G_1$}%
}}}}
\end{picture}%